\pdfoutput=1
\documentclass[times,review,preprint,authoryear]{elsarticle}

\usepackage{jcomp}
\usepackage{framed,multirow}

\usepackage{amssymb}
\usepackage{latexsym}

\usepackage{url}
\usepackage{xcolor}

\usepackage{overpic}
\usepackage{subcaption}
\usepackage[fleqn]{amsmath}
\usepackage{amsfonts,bm}
\newtheorem{theorem}{Theorem}
\newtheorem{remark}{Remark}
\newproof{proof}{Proof}

\usepackage{hyperref}
 \hypersetup{
    colorlinks=true,       
    linkcolor=red,          
     citecolor=green,        
     urlcolor=cyan           
 }

\definecolor{newcolor}{rgb}{.8,.349,.1}

\journal{Journal of Computational Physics}

\begin{document}

\verso{C. Wang, M.-C. Lai and Z. Zhang}

\begin{frontmatter}

\title{Energy stable methods for phase-field simulation of droplet impact with surfactants}%

\author[1,2]{Chenxi Wang}
\ead{wangcx2017@mail.sustech.edu.cn}
\author[3]{Ming-Chih Lai}
\ead{mclai@math.nctu.edu.tw}
\author[2,4]{Zhen Zhang\corref{cor1}}
\cortext[cor1]{Corresponding author.}
\ead{zhangz@sustech.edu.cn}

\address[1]{Department of Mathematics,
          Harbin Institute of Technology,
          Harbin 150001, P.R. China}
\address[2]{Department of Mathematics,
            Southern University of Science and Technology (SUSTech), Shenzhen 518055, P.R. China}
\address[3]{Department of Applied Mathematics, National Yang Ming Chiao Tung University, 1001, Ta Hsueh Road, Hsinchu 300, Taiwan}
\address[4]{National Center for Applied Mathematics (Shenzhen), Guangdong Provincial Key Laboratory of Computational Science and Material Design, SUSTech International Center for Mathematics, Southern University of Science and Technology (SUSTech), Shenzhen 518055, P.R. China}


\begin{abstract}
This paper is devoted to the numerical study of droplet impact on solid substrates in presence of surfactants. We formulate the problem in an energetically variational framework and introduce an incompressible Cahn-Hilliard-Navier-Stokes system for the phase-field modeling of two-phase flows. Flory-Huggins potential and generalized Navier boundary condition are used to account for soluble surfactants and moving contact lines. 
Based on the convex splitting and pressure
stabilization technique, we develop unconditionally energy stable schemes for this model. The discrete energy dissipation law for the original energy is rigorously proved for the first-order scheme. The numerical methods are implemented using finite difference method in three-dimensional cylindrical coordinates with axisymmetry.
Using the proposed methods for this model, we systematically study the impact dynamics of both clean and contaminated droplets (with surfactants) in a series of numerical experiments. In general, the dissipation in the impact dynamics of a contaminated drop is smaller than that in the clean case, and topological changes are more likely to occur for contaminated drops. Adding surfactants can significantly influence the impact phenomena, leading to the enhancement of adherence effect on hydrophilic surfaces and splashing on hydrophobic surfaces. Some quantitative agreements with experiments are also obtained.
\end{abstract}


\begin{keyword}
\KWD
\newline
Phase-field
\newline
Surfactants
\newline
Energy stability
\newline
Droplet impact
\newline
Moving contact line
\end{keyword}

\end{frontmatter}

\section{Introduction}\label{sec_Introduction}
The effect of surfactants, namely surface-active substances, has been the subject of intense study for long time and is of great interest in many industrial applications. Surfactants can be used as emulsifiers, cleaning detergents, foaming agents, wetting agents, and dispersants \citep{probstein2005physicochemical, Eggleton2001tip, branger2002accelerated}. Recently, surfactants have been widely used in microsystems with the presence of interfaces, where the capillary effect dominates the inertia of fluids \citep{Baret2012}. An important application is surfactant-based inkjet printing, in which surfactants can help the absorbance of ink onto hydrophobic surfaces \citep{Kommeren2018}. The presence of surfactants in multiphase flow has great effects on the dynamics of interface by altering the interfacial tension. In addition, when surfactants are not uniformly distributed on the interface, the dynamical behavior of interface would also be significantly affected due to Marangoni effect \citep{garcke2014diffuse}. These complex dynamics bring difficulty in experimental study of multiphase flow with surfactants. As an alternative, modelling and numerical simulations have played an increasingly significant role, although there remain challenging tasks.

For decades, modeling and numerical simulations of surfactant-driven hydrodynamics have attracted much attention. The surfactant transport model was first presented by \cite{S90} and then revised in \cite{wong1996surfactant} by introducing convection-diffusion system on surface. Based on this model, many numerical methods have been developed to simulate multiphase flow with surfactants. These include front tracking method \citep{Muradoglu2008}, level-set method \citep{de2017level}, volume-of-fluid method \citep{James2004}, immersed boundary/immersed interface method \citep{chen2014conservative}, phase-field method \citep{liu2010phase,zhu2019numerical}, and lattice Boltzmann method \citep{liu2020modelling}. Among these methods, the phase-field method is of particular interests due to its versatility in modeling as well as simulations. In phase-field models, an order parameter is introduced to label the two different phases with different values, and the sharp interface is diffused and represented by a smooth function. The phase-field method was first used to study interfacial dynamics with surfactants in \cite{Laradji1992}, resulting in the phase-field surfactant (PFS) model. Numerous works have shown great performance of the PFS model in simulations of multiphase systems with surfactants \citep{teng2012simulating, gu2014energy, yang2021improved}. Recently, the PFS models were extensive investigated in the presence of hydrodynamics \citep{zhu2019numerical, yang2021novel}. In these models, two order parameters are introduced to represent local volume fraction of one phase (usually called phase-field variable) and surfactant concentration respectively. Besides the classical Ginzburg-Landau free energy, the Flory-Huggins potential with logarithmic terms is typically inserted into the total free energy in order to model mixing entropy due to the addition of surfactants. Moreover, a nonlinearly coupled term in both phase-field variable and surfactant concentration is present to account for the adsorption of surfactants on the fluid interface \citep{van2006diffuse, engblom2013diffuse}.

As long as a total free energy is given, a phase-field model can be derived as a gradient flow system which satisfies energy dissipation law. For numerical methods of such systems, energy stability is one of the most important properties that need to be take care of. The stiffness arising from nonlinear terms in the free energy also makes it difficult to design efficient numerical schemes with large time steps. A number of techniques have been developed for the construction of energy stable schemes. Using convex splitting approach, unconditionally energy stable methods for the Cahn-Hilliard equation were developed by \cite{eyre1998unconditionally}. The major idea is to implicitly treat the convex terms of energy functional and explicitly treat the concave terms. This technique was successfully applied to numerically solve the PFS model \citep{gu2014energy}. In a similar manner, the so-called stabilization method was proposed by introducing a stabilization term to avoid solving nonlinear system \citep{chen1998applications}. Despite high efficiency and easy-to-implement, these two methods are not easily employed in developing high-order schemes with unconditional energy stability. Recently, by generalizing the Lagrange multiplier method \citep{GuillenGonzalez2013}. \cite{MR3633666} proposed an invariant energy quadratization (IEQ) approach with an auxiliary function representing the square root of the nonlinear part in the free energy. The original system is then transformed into an equivalent one with quadratized free energy, for which unconditionally energy stable and high-order linear decoupled schemes can be easily obtained. When a scalar auxiliary variable (SAV) is used instead of an auxiliary function, one can immediately obtain SAV approach which was first introduced by \cite{shen2018scalar} and is more efficient. Since their births, the IEQ and SAV approaches have attracted great interest and have been applied in constructing high-order schemes for many dissipative systems, including the Cahn-Hilliard-Navier-Stokes (CHNS) model \citep{MR4131825}, the PFS model \citep{yang2017linear, yang2021improved,qin2020fully} and other systems \citep{MR3815548, zhu2019numerical,yang2021novel,huang2020highly,Cheng2020}. It is worth mentioning that the IEQ/SAV approach is provable to be energy stable with respect to the modified energy instead of the original energy, which limits its use in physical systems sensitive to energy changes. Some other widely used energy stable methods also exist, for instance, fully implicit schemes and exponential time-differencing schemes \citep{xu2019stability,MR4253790}.

When two immiscible fluids interacts with a solid substrate, a contact line forms as the intersection of fluid-fluid interface with the solid boundary. In equilibrium, the angle between the two interfaces, also known as contact angle, is related to the surface tensions of the three interfaces and determined by Young-Dupr\'{e} equation \citep{Young1805} (Fig.~\ref{MCL}). If one fluid is displaced by the other along the substrate, moving contact line (MCL) must be considered. As is well known, the incompatibility of the moving contact line with the conventional no-slip boundary condition emerges in continuum hydrodynamics, since non-integrable stress singularity is generated at the contact line due to the use of no-slip boundary condition. In the past decades, a number of approaches were proposed to resolve MCL problems by using hydrodynamic models, molecular dynamic models, and diffuse-interface models \citep{Jacqmin00, qian2003molecular, ren2007boundary, Eggers04a, DeConinck08, Yue10, sui2014numerical}. Recent modeling and simulation studies showed that the presence of surfactants can further influence contact line dynamics by changing fluid-fluid interfacial tension, leading to more complicated multiphase hydrodynamics \citep{lai2010numerical, xu2014level, zhang2014derivation, zhao2021thermodynamically, zhu2019thermodynamically, zhu2020phase, Wang2022}.
This makes surfactants appealing in the processes of spray coating, wetting, and many biological applications \citep{de2004capillarity}.
\begin{figure}[t!]
\centering
\includegraphics[scale=0.69]{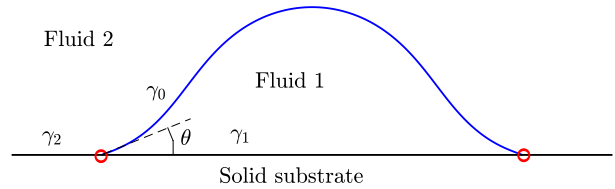}
\caption{\small A fluid-fluid interface intersects a solid substrate. The contact line (points) is represented by the red points; $\gamma_0$, $\gamma_1$ and $\gamma_2$ are respectively the surface tensions of the fluid-fluid interface and the two fluid-solid interfaces; $\theta$ is the contact angle. In equilibrium, the static contact angle $\theta_s$ satisfies the Young-Dupr\'{e} equation $\gamma_0\cos\theta_s=\gamma_2-\gamma_1$.}
\label{MCL}
\end{figure}

Recent experiments indicated that surfactants play a very important role in controlling droplet impact dynamics \citep{gatne2009surfactant, pan2016controlling}, which is worth extensive study due to its potential applications in inkjet printing. The purpose of the present work is to develop a computational model that can be applied to systematically investigate surfactant-induced droplet impact on solid substrates. We aim to develop efficient numerical methods which are not only stable when using large time steps but also thermodynamically consistent with the physical model. Among those models for multiphase flows with moving contact lines, diffuse interface models with generalized Navier boundary condition (GNBC) are of particular interests since they can easily handle topological changes during impact processes while being energetically variational \citep{qian2003molecular, qian2006variational}. By taking into accounts the surfactant transport, the incompressible two-phase flow and the GNBC, we model the droplet impact process by the binary fluid phase-field surfactant system with moving contact lines (NS-PFS-MCL). This model allows for variable density and viscosity and was shown to yield energy dissipation law in a variational framework \citep{zhu2019thermodynamically}. By using the convex splitting technique and the pressure stabilization method \citep{Yu2017}, we propose unconditionally energy stable schemes for the numerical discretization of the NS-PFS-MCL system.  Different from previous works \citep{zhu2019numerical,yang2021improved,yang2021novel,yang2017linear,qin2020fully,zhu2020phase}, we show that the proposed scheme satisfies the discrete energy law without modifying the original free energy. The proposed schemes are efficiently implemented in three-dimensional axisymmetric cylindrical coordinates and applied to the numerical study of droplet impact on solid substrates with satisfactory accuracy.

Another contribution of the paper is to systematically investigate the effect of surfactants on droplet impact dynamics on solid substrates, which is highly demanded in numerical works. Specifically, different impact phenomena, such as adherence, bouncing/partial bouncing and splashing, are numerically compared and analyzed from the viewpoint of dissipations for both clean and contaminated (by surfactants) droplets. Following the experimental and numerical works on impact dynamics of clean droplets \citep{Richard2002,Rioboo2002,Renardy2003,Rioboo2002, Zhang2016}, we choose typical dimensionless parameters (the Reynolds number, the Weber number and the wettability of solid surface) in our numerical simulations, giving rise to seven representative examples for demonstration of the surfactant effect on modifying impact dynamics. The overall effect of surfactants is to make the droplet more hydrophilic\slash hydrophobic on hydrophilic\slash hydrophobic surfaces: adding surfactants into an impact droplet can strengthen the adherence effect on hydrophilic surfaces (demonstrated in Example 7); the presence of surfactants makes a splashing droplet break up into more small drops on hydrophobic surfaces (shown in Example 5). In general, the dissipation in the impact dynamics of a contaminated drop is smaller than that in the clean case, and topological changes are more likely to occur for contaminated drops. Moreover, adding surfactants may induce some qualitative changes in droplet impact phenomena. A clean adherent droplet could experience complete bouncing when surfactants are added (Example 2); a clean droplet which completely rebounds becomes partially bounced in the presence of surfactants (Example 3). Last but not least, we obtain quantitative agreements with experimental results for impact dynamics in the case without surfactants \citep{Rioboo2002}, and simulate impact processes for contaminated drops as comparisons.


The rest of this paper is organized as follows. In Sect.~\ref{sec_Model}, we introduce the binary fluid PFS
model with MCL and describe its energy dissipation law. In Sect.~\ref{sec_Method}, the semi-discrete schemes are proposed, and the energy stability property is proved for the first-order scheme. An efficient numerical implementation  are also discussed. The accuracy test and discrete energy law are numerically validated in Sect.~\ref{sec_Simulation}. We also present extensive numerical experiments to demonstrate the effect of surfactants on hydrodynamics with moving contact lines. In particular, droplet impact dynamics are systematically studied. The paper is concluded in Sect.~\ref{sec_Conclusion} with discussions.

\section{Mathematical model}\label{sec_Model}
In this section, we present the binary fluid PFS model with MCL. In particular, The GNBC for the dynamics of moving contact lines is imposed at solid wall. The energy dissipation law for this system is obtained.

\subsection{Governing equations}\label{sec_GE}
We introduce two conserved order parameters $\phi$ and $\psi$ to represent local volume fraction of one phase and surfactant concentration respectively. $\phi$ is usually called phase-field variable which labels the two different phases with label $-1$ and $1$. $\psi$ is a percentage ranging from 0 to 1.

We first give the total free energy functional except for kinetic energy as follows \citep{zhu2019thermodynamically}:
\begin{equation} \label{total_F}
 E_{f}(\phi,\psi)=E_{GL}(\phi)+E_{sur}(\psi)+ E_{ad}(\phi,\psi)+E_{wf}(\phi).
\end{equation}

In this expression, $E_{GL}$ is given by the Ginzburg-Landau free energy with double-well potential in the bulk,
\begin{equation}\label{E_GL}
E_{GL}(\phi)= \int_{\Omega}\frac{ \mathrm{C n}^2}{2} |\nabla\phi|^{2} +\frac{1}{4}(\phi^2-1)^2\mathrm{d} x,
\end{equation}
where $\mathrm{Cn}$ is the Cahn number, which is determined by interface thickness and characteristic length.
The dimensionless interfacial tension between two fluids is given by $\gamma=2\sqrt{2}/3$. The quartic function $(\phi^2-1)^2/4$ is ``double well'' potential which expresses a preference at pure phases $\phi=\pm1$.

$E_{sur}$ is the energy representing mixing entropy in binary surfactant-fluid system:
\begin{equation}\label{E_sur}
E_{sur}(\psi)=\int_\Omega \mathrm{Pi}\big[\psi \ln \psi+(1-\psi) \ln (1-\psi)\big]\mathrm{d}x,
\end{equation}
where $\mathrm{Pi}$ is the temperature-dependent constant. Here, a Flory-Huggins free energy is rather preferred as it possesses Langmuir isotherm relation \citep{van2006diffuse}.

The interfacial energy accounting for the adsorption of surfactant is given by
\begin{equation}\label{E_AD}
E_{ad}(\phi, \psi)=\int_\Omega \frac{1}{2\mathrm{Ex}}\psi \phi^{2} -\frac{1}{4}\psi(1-\phi^{2})^{2} \mathrm{d}x,
\end{equation}
where $\mathrm{Ex}$ determines the bulk solubility.
Hereby, $\psi \phi^{2}/(2\mathrm{Ex})$ is an enthalpic term penalizing free surfactant in the respective phases, whereas the term $-\psi(1-\phi^{2})^{2}/4$ represents local attraction of surfactant to an existing interface.
Other approximations for $E_{ad}$ are also available, for instance,
\begin{equation*}
E_{ad}(\phi, \psi)=\int_\Omega \frac{1}{2\mathrm{Ex}} \psi \phi^{2}-\frac{1}{2}\psi|\nabla \phi|^{2}\mathrm{d}x,\quad\quad\text{or}\quad\quad
E_{ad}(\phi, \psi)=\int_\Omega \frac{1}{2\mathrm{Ex}} \psi \phi^{2}-\frac{1}{4} \psi(1-\phi^{2})\mathrm{d}x.
\end{equation*}
According to \cite{engblom2013diffuse}, \eqref{E_AD} performs best in numerical experiments in comparison with the other two formulae. Hence, we will use \eqref{E_AD} in our study. It is worth noting that the interplay of $\mathrm{Pi}$ and $\mathrm{Ex}$ yields the Langmuir adsorption isotherm \citep{engblom2013diffuse,zhu2019thermodynamically}.

Finally, we choose the wall energy $E_{wf}$ as
\begin{equation}\label{E_wf}
  E_{wf}(\phi)= \mathrm{Cn}\int_\Gamma\gamma_{wf}(\phi)\mathrm{d}S,  \quad\quad\gamma_{wf}(\phi)=-\frac{\sqrt{2}}{3}\cos \theta_{s} \sin \big(\frac{\pi \phi }{2} \big)+\frac{\gamma_1+\gamma_2}{2},
\end{equation}
where $\gamma_1$ and $\gamma_2$ are two fluid-solid surface tensions, $\theta_{s}$ is Young's angle determined through Young-Dupr\'{e} equation $2\sqrt{2}/3\cos\theta_s=\gamma_2-\gamma_1$, and $\sin (\frac{\pi \phi }{2})$ is an interpolation function between $\pm1$. In consideration of asymptotic properties, another interpolation function  $\frac12\phi(3-\phi^2)$ can also be used \citep{Xu2018}. In this work, we prefer to use $\sin (\frac{\pi \phi }{2})$ as its all derivatives are smooth and bounded.

After taking the variational derivatives of the free energy $E_f$ with respect to $\psi$ and $\phi$, one obtains the chemical potentials
\begin{align}
&\mu_{\psi}=\frac{\delta E_f}{\delta \psi}=\mathrm{P i} \ln \Big(\frac{\psi}{1-\psi}\Big)+\frac{1}{2 \mathrm{E x}} \phi^{2}-\frac{1}{4}(\phi^{2}-1)^{2},\label{ch_psi}\\
&\mu_{\phi}=\frac{\delta E_f}{\delta \phi}=-\mathrm{C n}^{2} \Delta \phi+\phi^{3}-\phi+\frac{1}{\mathrm{E x}} \psi \phi-\psi (\phi^{3}-\phi).\label{ch_phi}
\end{align}
Then the Cahn–Hilliard type governing equations are given by
\begin{align}
&\psi_{t}+\nabla\cdot(\mathbf{u}\psi)=\frac{1}{\mathrm{Pe}_{\psi}} \nabla \cdot M_{\psi} \nabla \mu_{\psi},\label{psi}\\
&\phi_{t}+\nabla\cdot(\mathbf{u}\phi)=\frac{1}{\mathrm{P e}_{\phi}} \Delta \mu_{\phi},\label{phi}
\end{align}
where $\mathrm{Pe}_\psi$ and $\mathrm{Pe}_\phi$ are the P\'{e}clet numbers, and $\mathbf{u}$ is the velocity field. We choose the degenerate mobility $M_{\psi}=\psi(1-\psi)$ which leads to the Fickian diffusion \citep{van2006diffuse}.

In addition, to take variable density and viscosity into account, we interpolate between two fluid densities and viscosities, i.e.,
\begin{equation*}
\rho(\phi)=\frac{1-\phi}{2}+\lambda_\rho \frac{1+\phi}{2}, \quad\quad \eta(\phi)=\frac{1-\phi}{2}+\lambda_\eta \frac{1+\phi}{2},
\end{equation*}
where $\lambda_{\rho}$ and $\lambda_{\eta}$ are respectively the density ratio and viscosity ratio.
Then using \eqref{phi}, we can modify the continuity equation as
\begin{equation}\label{continuity-eq}
\rho_t+\nabla\cdot(\mathbf{u}\rho)+\nabla\cdot \mathbf{J}_\rho=0,
\end{equation}
where the diffusive flux $\mathbf{J}_\rho=\frac{1-\lambda_\rho} {2\mathrm{Pe}_\phi}\nabla\mu_\phi$ is included.

Following \cite{Abels2012}, we can derive the momentum equation for hydrodynamics as
\begin{equation}\label{NS}
\rho (\mathbf{u}_{t}+ \mathbf{u} \cdot \nabla \mathbf{u})+\mathbf{J}_\rho\cdot\nabla\mathbf{u}+\nabla p=\frac{1}{\mathrm{R e}} \nabla \cdot \big(\eta D(\mathbf{u})\big)-\frac{1}{\mathrm{We}\mathrm{Cn}}( \phi\nabla\mu_{\phi} + \psi \nabla\mu_{\psi}),
\end{equation}
where $p$ is the pressure, and $D(\mathbf{u})=\nabla \mathbf{u}+(\nabla \mathbf{u})^{\top}$ represents twice the strain rate. Here, Reynolds number $\mathrm{Re}$ and Weber number $\mathrm{We}$ are the dimensionless parameters, and $\phi\nabla\mu_{\phi} + \psi \nabla\mu_{\psi}$ is the capillary force. Moreover, \eqref{NS} is supplemented with the incompressible condition
\begin{equation}\label{incompressible}
\nabla \cdot \mathbf{u}=0.
\end{equation}

\subsection{Boundary conditions}\label{sec_BC}
On the solid wall $\Gamma$, the relaxation boundary condition associated with the $\phi$-equation \eqref{phi} reads
\begin{equation}\label{DBC}
\phi_{t}+\mathbf{u}_{\mathbf{\tau}} \cdot\nabla_{\mathbf{\tau}} \phi=-\frac{1}{\mathrm{P e}_{s}} L(\phi),
\end{equation}
where $\mathrm{Pe}_s$ is the P$\acute{\text{e}}$clet number, $\nabla_\tau=(\mathbf{I}-\mathbf{n}\otimes\mathbf{n})\nabla$ is the surface gradient, and $\mathbf{u}_{\tau}=(\mathbf{I}-\mathbf{n}\otimes\mathbf{n})\mathbf{u}$ is
the tangential velocity on $\Gamma$, with $\mathbf{n}$ being the outward normal vector to $\Gamma$. In addition, $L(\phi)$ is the chemical potential on $\Gamma$ defined by
\begin{equation}\label{DBC_L}
L(\phi)=\mathrm{Cn} \partial_{n} \phi+\gamma_{w f}'(\phi).
\end{equation}

In consideration of contact line dynamics, GNBC is imposed on $\Gamma$:
\begin{equation}\label{GNBC}
\frac{ \mathbf{u}_{s}}{\mathrm{L}_{s} l_{s}}=\frac{ L(\phi) \nabla_{\tau} \phi }{\mathrm{Ca}\eta} -\partial_{n} \mathbf{u}_{\tau},
\end{equation}
where $l_s(\phi)=(1-\phi)/2+\lambda_{l_s}(1+\phi)/2$ is the dimensionless interpolating slip length, $\mathrm{L}_{s}$ and $\lambda_{l_s}$ are the slip length of fluid 1 and the slip length ratio respectively, $\mathrm{Ca}=\mathrm{We}/\mathrm{Re}$ is the Capillary number. Given the wall velocity $\mathbf{u}_w$, $\mathbf{u}_s=\mathbf{u}_\tau-\mathbf{u}_w$ is the slip velocity on $\Gamma$. GNBC establishes a linear response relation between the slip velocity $\mathbf{u}_s$ and the sum of the viscous stress $\partial_{n} \mathbf{u}_{\tau}$ and the uncompensated Young stress $L(\phi) \nabla_{\tau} \phi$. Moreover, we also need the following impermeability conditions
\begin{equation}\label{BC1}
\partial_{n} \mu_\phi=0, \quad \quad \partial_{n} \mu_\psi=0, \quad \quad \mathbf{n} \cdot \mathbf{u}=0 \quad \quad\text { on } ~ \Gamma.
 \end{equation}

On the rest of boundaries, we impose the following natural boundary conditions, i.e.,
\begin{equation}\label{BC2}
 \partial_{n} \phi=0, \quad\quad \partial_{n} \mu_\phi=0, \quad \quad
 \partial_{n} \mu_\psi=0, \quad \quad \partial_{n} \mathbf{u}_{\tau}=\mathbf{0}, \quad\quad   \mathbf{n}\cdot \mathbf{u} =0 \quad \quad \text { on } ~ \partial\Omega/\Gamma.
\end{equation}

\subsection{Energy law of the model}\label{sec_EL}
The system of equations and the associated boundary conditions presented in Sections \ref{sec_GE} and \ref{sec_BC} constitutes a dissipative system, whose energy dissipation law will be shown in this subsection. To this end, we first define the total energy $E_{tot}$ of the system as the sum of the kinetic energy $E_k$ and the free energy $E_f$:
\begin{equation}\label{E_tot}
E_{tot}(\mathbf{u},\phi,\psi) = E_k(\mathbf{u},\rho)+E_f(\phi,\psi),
\end{equation}
where $E_k(\mathbf{u},\rho) = \mathrm{We} \mathrm{C n}\int_{\Omega}\rho|\mathbf{u}|^{2}\mathrm{d} x/2 $ and $E_f(\phi,\psi)$ is given in \eqref{total_F}.
Five types of dissipations contribute to the dissipative mechanism: the viscous dissipation ($R_v$), the slip dissipation ($R_s$), the diffusion dissipation in phase-field dynamics ($R_d$), the diffusion dissipation in surfactant dynamics ($R_e$), and the dissipation due to phase-field relaxation at the solid surface ($R_r$), namely,
\begin{equation}
\begin{aligned}
&R_v = \frac{\mathrm{C a} \mathrm{C n}}{2}\big\|\sqrt{\eta} D(\mathbf{u})\big\|^{2}:=\frac{\mathrm{C a} \mathrm{C n}}{2}\int_{\Omega}\eta \big\|D(\mathbf{u})\big\|_F^{2}\mathrm{d}\mathbf{x}, \qquad && R_s= \frac{\mathrm{C a} \mathrm{C n}}{\mathrm{L}_{s}}\big\|\sqrt{\eta/l_s} \mathbf{u}_{s}\big\|_{\Gamma}^{2},\\
&R_{d} = \frac{1}{\mathrm{P e}_{\phi}}\big\|\nabla \mu_{\phi}\big\|^{2},\qquad\quad R_{e} = \frac{1}{\mathrm{P e}_{\psi}}\big\|\sqrt{M_{\psi}} \nabla \mu_{\psi}\big\|^{2},
&& R_r = \frac{\mathrm{C n}}{\mathrm{P e}_{s}}\big\|L(\phi)\big\|_\Gamma^{2},
\end{aligned}
\end{equation}
where $\|\cdot\|$ and $\|\cdot\|_\Gamma$ are the $L^2$ norm in $\Omega$ and on $\Gamma$ respectively, and $\|\cdot\|_F$ is the Frobenius norm for the matrix.
With these concepts in hand, we can prove the following energy law.
\begin{theorem}\label{thm1}
 Assume $\mathbf{u}_w =\mathbf{0}$, then the system \eqref{ch_psi}--\eqref{BC2} is a dissipative system satisfying the following energy dissipation law
\begin{equation}\label{continue_law}
\frac{\mathrm{d} }{\mathrm{d} t} E_{tot}= -(R_v+R_s+R_d+R_e+R_r)\leq0.
\end{equation}
\end{theorem}

\begin{proof}
We first use the modified continuity equation $\rho_t+\nabla\cdot(\mathbf{u}\rho)+\nabla\cdot \mathbf{J}_\rho=0$ and integration by parts to obtain
\begin{equation}\label{thm1_1}
\begin{aligned}
\frac{\mathrm{d}}{\mathrm{d}t}\int_\Omega\frac{1}{2}\rho|\mathbf{u}|^{2}\mathrm{d} x= \int_\Omega\rho\mathbf{u}\cdot\mathbf{u}_t
+\frac{1}{2}\rho_t|\mathbf{u}|^{2}\mathrm{d} x=\int_\Omega\mathbf{u}\cdot(\rho\mathbf{u}_t+\rho\mathbf{u}\cdot\nabla\mathbf{u}+\mathbf{J}_\rho\cdot\nabla\mathbf{u})\mathrm{d}x.
\end{aligned}
\end{equation}
Then by taking the inner product of \eqref{NS} with $\mathrm{We} \mathrm{C n}\mathbf{u}$ in $\Omega$ and using \eqref{thm1_1}, we have
\begin{equation}\label{thm1_2}
\begin{aligned}
\frac{\mathrm{d}}{\mathrm{d}t}E_k+R_v-\mathrm{C a} \mathrm{C n}\int_\Gamma\eta \mathbf{u}_\tau\cdot\partial_n \mathbf{u}_\tau \mathrm{d}S
=-\int_\Omega \phi \mathbf{u}\cdot\nabla\mu_\phi+\psi\mathbf{u}\cdot\nabla\mu_\psi\mathrm{d} x,
\end{aligned}
\end{equation}
where we have used the following integration by parts,
\begin{equation*}
-\int_\Omega\mathbf{u}\cdot\big[\nabla\cdot(\eta D(\mathbf{u}))\big]\mathrm{d} x
=\frac{1}{2}\big\|\sqrt{\eta} D(\mathbf{u})\big\|^{2}-\int_\Gamma\eta \mathbf{u}_\tau\cdot\partial_n \mathbf{u}_\tau \mathrm{d}S.
\end{equation*}
By taking the inner product of \eqref{GNBC} with $\mathrm{Ca} \mathrm{C n}\eta \mathbf{u}_\tau$ on $\Gamma$ and using $\mathbf{u}_w=\mathbf{0}$, we arrive at
\begin{equation}\label{thm1_3}
R_s=
-\mathrm{C a} \mathrm{C n}\int_\Gamma\eta \mathbf{u}_\tau\cdot\partial_n \mathbf{u}_\tau \mathrm{d}S+\mathrm{C n}\int_\Gamma L(\phi)\mathbf{u}_\tau\cdot \nabla_\tau \phi \mathrm{d}S.
\end{equation}
Taking the inner products of \eqref{psi} and \eqref{ch_psi} with $\mu_\psi$ and $\psi_t$ respectively in Ω, and collecting the results together with the homogeneous Neumann boundary conditions in \eqref{BC1} and \eqref{BC2}, we have
\begin{equation}\label{thm1_4}
-\frac{1}{\mathrm{P e}_{\psi}}\big\|\sqrt{M_{\psi}} \nabla \mu_{\psi}\big\|^{2}
=-\int_\Omega\psi\mathbf{u}\cdot\nabla\mu_\psi\mathrm{d} x + \int_\Omega \bigg(\mathrm{P i} \ln \Big(\frac{\psi}{1-\psi}\Big)+\frac{1}{2 \mathrm{E x}} \phi^{2}-\frac{1}{4}(\phi^{2}-1)^{2}\bigg)\psi_t \mathrm{d} x.
\end{equation}
Similar manipulations for \eqref{phi} and \eqref{ch_phi} lead to
\begin{equation}\label{thm1_5}
-\frac{1}{\mathrm{P e}_{\phi}}\big\|\nabla \mu_{\phi}\big\|^{2}
=-\int_\Omega\phi\mathbf{u}\cdot\nabla\mu_\phi\mathrm{d} x + \int_\Omega \big(-\mathrm{C n}^{2} \Delta \phi+\phi^{3}-\phi+\frac{1}{\mathrm{E x}} \psi \phi-\psi (\phi^{3}-\phi)\big)\phi_t \mathrm{d} x.
\end{equation}
Moreover, taking the inner products of \eqref{DBC} and \eqref{DBC_L} with $\mathrm{Cn} L(\phi)$ and $\mathrm{Cn}\phi_t$ respectively on $\Gamma$, we have
\begin{equation}\label{thm1_6}
-R_r=\mathrm{Cn} \int_\Gamma L(\phi)\mathbf{u}_\tau \cdot\nabla_\tau\phi\mathrm{d} S+\mathrm{Cn}\int_\Gamma \big(\mathrm{Cn} \partial_{n} \phi+\gamma_{w f}'(\phi)\big)\phi_t \mathrm{d} S.
\end{equation}
To conclude, we combine \eqref{thm1_2}--\eqref{thm1_6} and obtain
\begin{equation*}
\begin{aligned}
\frac{\mathrm{d} }{\mathrm{d} t} E_{tot}=&~\frac{\mathrm{d} }{\mathrm{d} t}E_k+\frac{\mathrm{d} }{\mathrm{d} t}E_f\\
    =&-(R_v+R_s+R_d+R_e+R_r)\leq0.\quad\Box
    \end{aligned}
\end{equation*}

\end{proof}

\begin{remark}
It is worth noting that we only consider the fixed solid wall in our work, i.e., $\mathbf{u}_w =\mathbf{0}$, in the assumption of Theorem \ref{thm1}. Under this condition, the governing system satisfies an energy dissipation law. This assumption will continue to be considered in the discrete energy law in Theorem \ref{thm2}.
In the case where $\mathbf{u}_w\neq\mathbf{0}$, external work power $-\mathrm{L}_{s}^{-1}\int_\Gamma
 l_{s}^{-1} \eta \mathbf{u}_{s}\cdot\mathbf{u}_w\mathrm{d}S$ done by the wall to the flow will be included in the energy law \eqref{continue_law}.
\end{remark}

\section{Numerical method}\label{sec_Method}
In this section, we construct a set of energy stable schemes. We start with
a first-order time-discrete scheme, whose energy stability can be rigorously proved. Then we extend it to a second-order version.

Since we are only interested in $\phi\in [-1, 1]$, we can modify $F(\phi)$ to have a quadratic growth rate for $|\phi|>1$ \citep{shen2010numerical}:
\begin{equation*}
\widehat{F}(\phi)= \begin{cases}\big(\phi+1\big)^{2}, &\quad \text { if }~ \phi<-1, \\ \big(\phi^{2}-1\big)^{2}/4, &\quad \text { if }~-1 \leq \phi \leq 1, \\ \big(\phi-1\big)^{2}, & \quad\text { if }~ \phi>1.\end{cases}
\end{equation*}
It is proved that this truncated $\hat{F}(\phi)$ with quadratic growth at infinity can guarantee the boundedness of $\phi$ in the Ginzburg-Landau energy $E_{GL}$ .
Correspondingly, we define $\hat{f}(\phi)=\hat{F}'(\phi)$. Hereafter, the notation $~\hat{}~$ is omitted for convenience.

By using the convex splitting technique for the terms in \eqref{ch_psi} and using the stabilization method \citep{Gao2014} for \eqref{ch_phi} and \eqref{phi} with relaxation boundary condition \eqref{DBC}-\eqref{DBC_L}, we can construct a linear scheme. In addition, a splitting method based on pressure stabilization \citep{Yu2017} is implemented to decouple the computation of velocity from pressure for the variable density Navier–Stokes equations with the GNBC.

\subsection{First-order scheme}
We give a first-order temporal discretized scheme for the system \eqref{ch_psi}--\eqref{BC2}:
\vspace{2mm}

\emph{Step 1}: We compute $\psi^{n+1}$, $\mu_\psi^{n+1}$, $\phi^{n+1}$, $\mu_\phi^{n+1}$, $\mathbf{u}^{n+1}$ by

\begin{align}
&\frac{\psi^{n+1}-\psi^n}{\Delta t}+\nabla\cdot(\mathbf{u}^{n+1}\psi^n)=\frac{1}{\mathrm{Pe}_{\psi}} \nabla \cdot M_{\psi}^{n+1} \nabla \mu_{\psi}^{n+1},\label{scheme_psi}\\
&\mu_{\psi}^{n+1}=\mathrm{P i} \ln \Big(\frac{\psi^{n+1}}{1-\psi^{n+1}}\Big)+\frac{1}{2 \mathrm{E x}} (\phi^n)^{2}-\frac{1}{4}\big((\phi^n)^{2}-1\big)^{2},\label{scheme_ch_psi}\\
&\frac{\phi^{n+1}-\phi^n}{\Delta t}+\nabla\cdot(\mathbf{u}^{n+1}\phi^n)=\frac{1}{\mathrm{P e}_{\phi}} \Delta \mu_{\phi}^{n+1},\label{scheme_phi}\\
&\mu_{\phi}^{n+1}=-\mathrm{C n}^{2} \Delta \phi^{n+1}+s_1(\phi^{n+1}-\phi^n)+f(\phi^n)+\frac{1}{\mathrm{E x}} \psi^{n+1} \phi^{n+1} -\psi^{n+1}\big((\phi^n)^{3}-\phi^{n+1}\big),\label{scheme_ch_phi}\\
&\begin{aligned}
&\rho^{n} \frac{\mathbf{u}^{n+1}-\mathbf{u}^{n} }{\Delta t}  + \rho^{n+1} \mathbf{u}^n \cdot \nabla \mathbf{u}^{n+1} +\mathbf{J}_\rho^{n+1}\cdot\nabla\mathbf{u}^{n+1}+\nabla (2p^n-p^{n-1})\\
&\quad\quad\quad=\frac{1}{\mathrm{R e}} \nabla \cdot \big(\eta^{n+1} D(\mathbf{u}^{n+1})\big)
-\frac{1}{\mathrm{We}\mathrm{Cn}}(\phi^n\nabla\mu_{\phi}^{n+1}  + \psi^n\nabla\mu_{\psi}^{n+1} )\\
&\quad\quad\quad\quad-\frac{1}{2} \frac{\rho^{n+1}-\rho^{n}}{\Delta t} \mathbf{u}^{n+1}-\frac{1}{2} \nabla \cdot(\rho^{n+1} \mathbf{u}^{n}) \mathbf{u}^{n+1}-\frac{1}{2}( \nabla \cdot \mathbf{J}_\rho^{n+1}) \mathbf{u}^{n+1},
\end{aligned}\label{scheme_NS}
\end{align}
with boundary conditions
\begin{align}
&\frac{\phi^{n+1}-\phi^{n}}{\Delta t}+\mathbf{u}_{\tau}^{n+1}\cdot \nabla_{\tau} \phi^{n}=-\frac{1}{\mathrm{Pe}_s}L_\phi^{n+1}  &&\text { on }~ \Gamma, \label{scheme_DBC}\\
&\frac{ \mathbf{u}_{s}^{n+1}}{\mathrm{L}_{s} l_{s}^{n+1}}=\frac{ L_\phi^{n+1} \nabla_{\tau} \phi^n }{\mathrm{Ca}\eta^{n+1}} -\partial_{n} \mathbf{u}_{\tau}^{n+1} &&\text { on } ~ \Gamma,\label{scheme_GNBC}\\
&\partial_{n} \mu_\phi^{n+1}=0, \quad\quad\partial_{n} \mu_\psi^{n+1}=0,\quad\quad  \mathbf{n}\cdot \mathbf{u}^{n+1}=0  &&\text { on } ~\Gamma, \label{scheme_BC}\\
& \partial_{n} \phi^{n+1}=0,\quad\quad \partial_{n} \mu_\phi^{n+1}=0, \quad\quad \partial_{n} \mu_\psi^{n+1}=0,\quad\quad
 \partial_{n} \mathbf{u}_{\tau}^{n+1}=\mathbf{0}, \quad\quad \mathbf{n}\cdot \mathbf{u}^{n+1}=0 &&\text { on }~ \partial\Omega/\Gamma,
\end{align}
where
\begin{align}
&L_\phi^{n+1} =\mathrm{Cn}\partial_{n} \phi^{n+1}+s_{2}(\phi^{n+1}-\phi^{n})+\gamma_{wf}^{\prime}(\phi^{n}), \label{scheme_DBC_L}\\
&\rho^{n+1}=(1-\phi^{n+1}) / 2+\lambda_{\rho}(1+\phi^{n+1}) / 2 , \quad\quad\eta^{n+1}=(1-\phi^{n+1}) / 2+\lambda_{\eta}(1+\phi^{n+1}) / 2, \\
&l_s^{n+1}=(1-\phi^{n+1})/2+\lambda_{l_s}(1+\phi^{n+1})/2,\quad\quad M_\psi^{n+1}=\psi^{n+1}(1-\psi^{n+1}),\\
&f(\phi^n)=\begin{cases}2(\phi^n+1), & \text { if } ~\phi^n<-1, \\ (\phi^n)^{3}-\phi^n, & \text { if }~-1 \leq \phi^n \leq 1, \\ 2(\phi^n-1), & \text { if } ~\phi^n>1,\end{cases}\quad\quad\mathbf{J}_\rho^{n+1}=\frac{1-\lambda_\rho}{2\mathrm{Pe}_\phi}\nabla\mu_\phi^{n+1}, \quad\quad \mathbf{u}_s^{n+1}=\mathbf{u}_\tau^{n+1}-\mathbf{u}_w.
\end{align}

\vspace{2mm}

\emph{Step 2}: Update $p^{n+1}$ by solving
\begin{equation}\label{scheme_pressure}
  \Delta(p^{n+1}-p^n)=\frac{\bar{\rho}}{\Delta t}\nabla\cdot\mathbf{u}^{n+1},
\end{equation}
with $\bar{\rho}=\min(1,\lambda_\rho)$ and boundary condition
\begin{equation*}
 \mathbf{n}\cdot \nabla p^{n+1}=0 \quad\quad\text { on } ~\partial\Omega.
\end{equation*}

It should be noted that when discretizing the momentum equation \eqref{NS}, we include three more terms in \eqref{scheme_NS}, which serves as a first-order approximation of $\frac{1}{2}(\rho_t+\nabla\cdot(\mathbf{u}\rho)+\nabla\cdot \mathbf{J}_\rho)\mathbf{u}$ at $t^{n+1}$. These additional terms are approximately zero due to the modified continuity equation \eqref{continuity-eq}. Hence, \eqref{scheme_NS} is indeed a consistent first-order approximation to \eqref{NS} in time.

For the proposed scheme \eqref{scheme_psi}--\eqref{scheme_pressure}, we have the following discrete energy law.
\begin{theorem}\label{thm2}
Assume $\mathbf{u}_w =\mathbf{0}$, $0<\psi^{n+1}<1$, $s_1\geq1$, and $s_2\geq\lvert \frac{\sqrt{2}\pi^2}{24}\cos \theta_{s}\rvert$,  then the scheme \eqref{scheme_psi}--\eqref{scheme_pressure} is unconditionally energy stable, and satisfies the following discrete dissipation law:
\begin{equation}\label{discrete_law}
\begin{aligned}
E_{tot}^{n+1}-E_{tot}^n\leq &-\frac{\Delta t\mathrm{CaCn}}{2}\big\|\sqrt{\eta^{n+1}} D(\mathbf{u}^{n+1})\big\|^{2} -\frac{\Delta t}{\mathrm{Pe}_{\psi}}\big\|\sqrt{M_{\psi}^{n+1}} \nabla \mu_{\psi}^{n+1}\big\|^{2}-\frac{\Delta t}{\mathrm{Pe}_{\phi}}\big\|\nabla \mu_{\phi}^{n+1}\big\|^{2}\\
&-\frac{\Delta t\mathrm{CaCn}}{\mathrm{L}_{s}}\big\|\sqrt{\eta^{n+1}/l_s^{n+1}}\mathbf{u}_{s}^{n+1}\big\|_{\Gamma}^{2}-\frac{\Delta t\mathrm{Cn}}{\mathrm{Pe}_{s}}\big\|L_\phi^{n+1}\big\|_{\Gamma}^{2} \leq 0,
\end{aligned}
\end{equation}
where
\begin{equation*}
E^n_{tot}= E_k(\mathbf{u}^n,\rho^n)+E_{GL}(\phi^n)+E_{sur}(\psi^n)+E_{ad}(\phi^n,\psi^n)+E_{wf}(\phi^n)+\frac{\Delta t^2 \mathrm{WeCn}}{2\bar{\rho}}\big\|\nabla p^n\big\|^2,
\end{equation*}
where $E_{GL}$, $E_{sur}$, $E_{ad}$, $E_{wf}$ and $E_k$ are defined in \eqref{E_GL}--\eqref{E_wf} and \eqref{E_tot} respectively.
\end{theorem}

\begin{proof}
We first take the inner products of \eqref{scheme_psi} and \eqref{scheme_ch_psi} with $\Delta t\mu_{\psi}^{n+1}$ and $\psi^{n+1}-\psi^n$ in $\Omega$ respectively. Summing the results up and using the integration by parts and the homogeneous
Neumann boundary condition, we obtain
\begin{equation}\label{thm2_psi1}
\begin{aligned}
    &\int_\Omega (\psi^{n+1}-\psi^n)\bigg(\mathrm{P i} \ln \Big(\frac{\psi^{n+1}}{1-\psi^{n+1}}\Big)+\frac{1}{2 \mathrm{E x}} (\phi^n)^{2}-\frac{1}{4}\big((\phi^n)^{2}-1\big)^{2}\bigg)\mathrm{d} x-\Delta t\int_\Omega\psi^{n}\mathbf{u}^{n+1}\cdot\nabla\mu_\psi^{n+1}\mathrm{d} x \\ =&-\frac{\Delta t}{\mathrm{Pe}_{\psi}}\big\|\sqrt{M_{\psi}^{n+1}} \nabla \mu_{\psi}^{n+1}\big\|^{2}.
    \end{aligned}
\end{equation}
Using the inequality
\begin{equation*}
 (a-b) \ln a\geq(a \ln a-a)-(b \ln b-b),
\end{equation*}
we can estimate the left side of \eqref{thm2_psi1} and recast it into an inequality
\begin{equation}\label{thm2_psi2}
\begin{aligned}
  &E_{sur}(\psi^{n+1})-E_{sur}(\psi^{n})+E_{ad}(\phi^n,\psi^{n+1})-E_{ad}(\phi^n,\psi^{n})
  -\Delta t\int_\Omega\psi^{n}\mathbf{u}^{n+1}\cdot\nabla\mu_\psi^{n+1}\mathrm{d} x  \\
  \leqslant& -\frac{\Delta t}{\mathrm{Pe}_{\psi}}\big\|\sqrt{M_{\psi}^{n+1}} \nabla \mu_{\psi}^{n+1}\big\|^{2}.
\end{aligned}
\end{equation}

By taking the inner products of \eqref{scheme_phi} and \eqref{scheme_ch_phi} with $\Delta t\mu_\phi^{n+1}$ and $\phi^{n+1}-\phi^n$ respectively in $\Omega$ and combining them, we have
\begin{equation}\label{thm2_phi1}
\begin{aligned}
  &\int_\Omega(\phi^{n+1}-\phi^n)\Big(-\mathrm{C n}^{2} \Delta \phi^{n+1}+s_1(\phi^{n+1}-\phi^n)+f(\phi^n) -\psi^{n+1}\big((\phi^n)^{3}-\phi^{n+1}\big)+\frac{1}{\mathrm{E x}} \psi^{n+1} \phi^{n+1}\Big)\mathrm{d} x\\
  =&\Delta t\int_\Omega\phi^{n}\mathbf{u}^{n+1}\cdot\nabla\mu_\phi^{n+1}\mathrm{d} x -\frac{\Delta t}{\mathrm{Pe}_{\phi}}\big\|\nabla \mu_{\phi}^{n+1}\big\|^{2}.
\end{aligned}
\end{equation}
By using integration by parts and the inequality
\begin{equation}\label{thm2_formula1}
  a(a-b)=\frac12|a|^2-\frac12|b|^2+\frac12|a-b|^2\geq\frac{1}{2}(\lvert a\rvert^2-\lvert b\rvert^2)
\end{equation}
recursively, we can recast \eqref{thm2_phi1} into
\begin{equation}\label{thm2_phi2}
\begin{aligned}
&\Delta t\int_\Omega\phi^{n}\mathbf{u}^{n+1}\cdot\nabla\mu_\phi^{n+1}\mathrm{d} x -\frac{\Delta t}{\mathrm{Pe}_{\phi}}\big\|\nabla \mu_{\phi}^{n+1}\big\|^{2}\\
\geq&~\int_\Omega\mathrm{C n}^{2}(\nabla\phi^{n+1}-\nabla\phi^n)\nabla\phi^{n+1} \mathrm{d} x-\mathrm{C n}^{2}\int_\Gamma(\phi^{n+1}-\phi^n) \partial_{n} \phi^{n+1}\mathrm{d} S\\
&+\int_\Omega\bigg\{ \big(s_1-\frac{1}{2}f'(\sigma^n)\big)\lvert\phi^{n+1}
-\phi^n\rvert^2+F(\phi^{n+1})-F(\phi^n)+ \frac{1}{2\mathrm{E x}}\psi^{n+1}\big(\lvert\phi^{n+1}\rvert^{2}-
\lvert\phi^{n}\rvert^{2}\big)\\
&+\frac{1}{2}\psi^{n+1}\Big[\big(\lvert\phi^{n}\rvert^{2}-\lvert\phi^{n+1}\rvert^{2}\big)
\lvert\phi^{n}\rvert^{2}+\big(\lvert\phi^{n+1}\rvert^{2}-\lvert\phi^{n}\rvert^{2}\big)\Big]\bigg\}\mathrm{d} x\\
  \geq&~E_{GL}(\phi^{n+1})-E_{GL}(\phi^{n})+E_{ad}(\phi^{n+1},\psi^{n+1})-E_{ad}(\phi^n,\psi^{n+1})\\
  &-\mathrm{C n}^{2}\int_\Gamma(\phi^{n+1}-\phi^n) \partial_{n} \phi^{n+1}\mathrm{d} S,
\end{aligned}
\end{equation}
where the first inequality is obtained by the Taylor expansion of $F(\phi)$ at $\phi^n$,
\begin{equation*}
\begin{aligned}
f(\phi^n)(\phi^{n+1}-\phi^n)
=F(\phi^{n+1})-F(\phi^n)-\frac{1}{2}f'(\sigma^n)\lvert\phi^{n+1}-\phi^n\rvert^2,
\end{aligned}
\end{equation*}
for some $\sigma^n\in\big[\min \{\phi^{n+1}, \phi^n\}, \max \{\phi^{n+1}, \phi^n\}\big]$, and the last inequality is due to the condition that $s_1\geq1$.
Then, by taking the inner products of the
equation \eqref{scheme_DBC} and \eqref{scheme_DBC_L} with $\Delta t\mathrm{Cn} L_\phi^{n+1}$ and $\mathrm{Cn}(\phi^{n+1}-\phi^n)$ respectively on $\Gamma$, we have
\begin{equation}\label{thm2_DBC1}
\begin{aligned}
  &\mathrm{Cn}\int_\Gamma\big(\phi^{n+1}-\phi^n\big)\big(\mathrm{Cn} \partial_{n} \phi^{n+1}+s_2(\phi^{n+1}-\phi^n)+\gamma'_{wf}(\phi^n) \big) \mathrm{d}S\\
  =&-\Delta t\mathrm{Cn}\int_\Gamma L_\phi^{n+1}\mathbf{u}_{\tau}^{n+1}\cdot \nabla_{\tau} \phi^{n} \mathrm{d}S-\frac{\Delta t\mathrm{Cn}}{\mathrm{Pe}_{s}}\big\|L_\phi^{n+1}\big\|_{\Gamma}^{2}.
  \end{aligned}
\end{equation}
A Taylor expansion of $\gamma_{wf}$ at $\phi^n$ leads to
\begin{equation*}
\begin{aligned}
\gamma'_{wf}(\phi^n)(\phi^{n+1}-\phi^n)
=\gamma_{wf}(\phi^{n+1})-\gamma_{wf}
(\phi^n)-\frac{1}{2}\gamma''_{wf}(\zeta^n)\lvert\phi^{n+1}-\phi^n\rvert^2,
\end{aligned}
\end{equation*}
for some $\zeta^n\in\big[\min \{\phi^{n+1}, \phi^n\}, \max \{\phi^{n+1}, \phi^n\}\big]$.
Then \eqref{thm2_DBC1} can be recast as
\begin{equation}\label{thm2_DBC2}
\begin{aligned}
&-\Delta t\mathrm{Cn}\int_\Gamma L_\phi^{n+1}\mathbf{u}_{\tau}^{n+1}\cdot \nabla_{\tau} \phi^{n}-\frac{\Delta t\mathrm{Cn}}{\mathrm{Pe}_{s}}\big\|L_\phi^{n+1}\big\|_{\Gamma}^{2}\\
=&~\mathrm{Cn}\int_\Gamma\big(s_2-\frac{1}{2}\gamma''_{wf}(\zeta^n)\big)\lvert\phi^{n+1}
-\phi^n\rvert^2+\gamma_{wf}(\phi^{n+1})-\gamma_{wf}(\phi^n) \mathrm{d}S+\mathrm{Cn}^2\int_\Gamma(\phi^{n+1}-\phi^n) \partial_{n} \phi^{n+1}\mathrm{d}S\\
\geq &~ E_{wf}(\phi^{n+1})-E_{wf}(\phi^{n})+\mathrm{Cn}^2\int_\Gamma(\phi^{n+1}-\phi^n) \partial_{n} \phi^{n+1}\mathrm{d}S,
\end{aligned}
\end{equation}
where the last inequality is due to the condition that $s_2\geq\lvert \frac{\sqrt{2}\pi^2}{24}\cos \theta_{s}\rvert$.

We now consider the Navier-Stokes equations with the GNBC.
Using integration by parts and boundary conditions $\partial_{n} \mu_\phi^{n+1}=0$ and $\mathbf{n}\cdot\mathbf{u}^{n}=0$, we can derive
\begin{align}
&\int_\Omega \big((\rho^{n+1} \mathbf{u}^{n} \cdot \nabla) \mathbf{u}^{n+1}+\frac{1}{2} (\nabla \cdot(\rho^{n+1} \mathbf{u}^{n})) \mathbf{u}^{n+1}\big)\cdot\mathbf{u}^{n+1}\mathrm{d}x=\frac{1}{2}\int_{\partial\Omega}\mathbf{n} \cdot(\rho^{n+1} \mathbf{u}^{n}) \lvert\mathbf{u}^{n+1}\rvert^2\mathrm{d}x=0, \label{thm2_NS1}\\
&\int_\Omega\big(\mathbf{J}_\rho^{n+1} \cdot \nabla \mathbf{u}^{n+1}+\frac{1}{2}(\nabla \cdot \mathbf{J}_\rho^{n+1}) \mathbf{u}^{n+1}\big)\cdot\mathbf{u}^{n+1}\mathrm{d}x=\frac{1}{2}\int_{\partial\Omega}\mathbf{n}\cdot \mathbf{J}_\rho^{n+1} \lvert\mathbf{u}^{n+1}\rvert^2\mathrm{d}x=0.\label{thm2_NS2}
\end{align}
Taking the inner product of \eqref{scheme_NS} with $\Delta t\mathrm{WeCn} \mathbf{u}^{n+1}$ and using the equalities \eqref{thm2_NS1} and \eqref{thm2_NS2} and integration by parts, we obtain
\begin{equation}\label{thm2_NS3}
\begin{aligned}
&\mathrm{WeCn}\int_\Omega\rho^n(\mathbf{u}^{n+1}-\mathbf{u}^{n})\cdot\mathbf{u}^{n+1}\mathrm{d}x
-\Delta t\mathrm{WeCn}\int_\Omega(2p^n-p^{n-1})\nabla\cdot\mathbf{u}^{n+1}\mathrm{d}x\\
=&-\frac{\Delta t\mathrm{CaCn}}{2}\big\|\sqrt{\eta^{n+1}} D(\mathbf{u}^{n+1})\big\|^{2}+\Delta t\mathrm{CaCn}\int_\Gamma\eta^{n+1}\mathbf{u}_\tau^{n+1}\cdot\partial_n\mathbf{u}_\tau^{n+1}\mathrm{d}S\\
&-\Delta t\int_\Omega( \phi^n \nabla\mu_{\phi}^{n+1}+ \psi^n\nabla\mu_{\psi}^{n+1})\cdot\mathbf{u}^{n+1}\mathrm{d}x-\frac{\mathrm{WeCn}}{2}\int_\Omega(\rho^{n+1}-\rho^n)|\mathbf{u}^{n+1}|^2\mathrm{d}x,
\end{aligned}
\end{equation}
where we have used the equality
\begin{equation*}
\int_\Omega\nabla(2p^n-p^{n-1})\cdot\mathbf{u}^{n+1}\mathrm{d}x=-\int_\Omega(2p^n-p^{n-1})\nabla\cdot\mathbf{u}^{n+1}\mathrm{d}x,
\end{equation*}
due to $\mathbf{n}\cdot\mathbf{u}^{n+1}=0$ on $\partial\Omega$.

We now consider each term in \eqref{thm2_NS3} separately. Direct calculations lead to
\begin{equation}\label{thm2_NS4}
\begin{aligned}
&\mathrm{WeCn}\int_\Omega\rho^n(\mathbf{u}^{n+1}-\mathbf{u}^{n})\cdot\mathbf{u}^{n+1}\mathrm{d}x+\frac{\mathrm{WeCn}}{2}\int_\Omega(\rho^{n+1}-\rho^n)|\mathbf{u}^{n+1}|^2\mathrm{d}x\\
=&~\frac{\mathrm{WeCn}}{2}\int_\Omega\rho^{n+1}|\mathbf{u}^{n+1}|^2\mathrm{d}x-\frac{\mathrm{WeCn}}{2}\int_\Omega\rho^n|\mathbf{u}^{n}|^2\mathrm{d}x
+\frac{\mathrm{WeCn}}{2}\int_\Omega\rho^n|\mathbf{u}^{n+1}-\mathbf{u}^{n}|^2\mathrm{d}x\\
=&~E_k(\mathbf{u}^{n+1},\rho^{n+1})-E_k(\mathbf{u}^{n},\rho^{n})+\frac{\mathrm{WeCn}}{2}\int_\Omega\rho^n|\mathbf{u}^{n+1}-\mathbf{u}^{n}|^2\mathrm{d}x.
\end{aligned}
\end{equation}
According to \eqref{scheme_pressure} and the Helmholtz-Hodge decomposition, $\nabla(p^{n+1}-p^n)$ is the unique projection of $\frac{\bar{\rho}}{\Delta t}\mathbf{u}^{n+1}$ on the linear subspace of gradient fields with homogeneous Neumann boundary condition. Thus, $\nabla(p^{n+1}-p^n)-\nabla(p^{n}-p^{n-1})$ is the unique projection of $\frac{\bar{\rho}}{\Delta t}(\mathbf{u}^{n+1}-\mathbf{u}^{n})$. Moreover,
$$\big\|\nabla(p^{n+1}-p^n)-\nabla(p^{n}-p^{n-1})\big\|\leq\frac{\bar{\rho}}{\Delta t}\big\|\mathbf{u}^{n+1}-\mathbf{u}^{n}\big\|.$$
Then using \eqref{thm2_formula1}, we have
\begin{equation}\label{thm2_NS5}
\begin{aligned}
&\int_\Omega(2p^n-p^{n-1})\nabla\cdot\mathbf{u}^{n+1}\mathrm{d}x\\
=&-\frac{\Delta t}{\bar{\rho}}\int_\Omega(p^{n+1}-2p^n+p^{n-1})\Delta(p^{n+1}-p^n)\mathrm{d}x
+\frac{\Delta t}{\bar{\rho}}\int_\Omega p^{n+1}\Delta(p^{n+1}-p^n)\mathrm{d}x\\
=&~\frac{\Delta t}{2\bar{\rho}}\big\|\nabla(p^{n+1}-2p^n+p^{n-1})\big\|^2-\frac{\Delta t}{2\bar{\rho}}\big(\big\|\nabla p^{n+1}\big\|^2-\big\|\nabla p^{n}\big\|^2\big)-\frac{\Delta t}{2\bar{\rho}}\big\|\nabla(p^n-p^{n-1})\big\|^2\\
\leq&~\frac{\bar{\rho}}{2\Delta t}\big\|\mathbf{u}^{n+1}-\mathbf{u}^{n} \big\|^2-\frac{\Delta t}{2\bar{\rho}}\big(\big\|\nabla p^{n+1}\big\|^2-\big\|\nabla p^{n}\big\|^2\big)\\
\leq&~\frac{1}{2\Delta t}\int_\Omega\rho^{n} |\mathbf{u}^{n+1}-\mathbf{u}^{n} |^2\mathrm{d}x-\frac{\Delta t}{2\bar{\rho}}\big(\big\|\nabla p^{n+1}\big\|^2-\big\|\nabla p^{n}\big\|^2\big)
\end{aligned}
\end{equation}
where the last inequality is due to the definition
of $\bar{\rho}$. According to \eqref{scheme_GNBC} and $\mathbf{u}_w=\mathbf{0}$, the boundary term in \eqref{thm2_NS3} is rewritten as
\begin{equation}\label{thm2_NS6}
\begin{aligned}
\int_\Gamma\eta^{n+1}\mathbf{u}_\tau^{n+1}\cdot\partial_n\mathbf{u}_\tau^{n+1}\mathrm{d}S=&~\int_\Gamma \mathbf{u}_\tau^{n+1}\cdot\Big(\frac{ L_\phi^{n+1} \nabla_{\tau} \phi^n }{\mathrm{Ca}}-
\frac{ \eta^{n+1}\mathbf{u}_{s}^{n+1}}{\mathrm{L}_{s} l_{s}^{n+1}}\Big)\mathrm{d}S\\
=&~\frac{1}{\mathrm{Ca}}\int_\Gamma L_\phi^{n+1}  \mathbf{u}_\tau^{n+1}\cdot \nabla_{\tau} \phi^n \mathrm{d}S-\frac{1}{\mathrm{L}_{s}}\big\|\sqrt{\eta^{n+1}/l_s^{n+1}}\mathbf{u}_{s}^{n+1}\big\|_{\Gamma}^{2}.
\end{aligned}
\end{equation}

Combing \eqref{thm2_NS4}--\eqref{thm2_NS6}, we can recast \eqref{thm2_NS3} into
\begin{equation}\label{thm2_NS7}
\begin{aligned}
&E_k(\mathbf{u}^{n+1},\rho^{n+1})-E_k(\mathbf{u}^{n},\rho^{n})+\frac{\Delta t^2\mathrm{WeCn}}{2\bar{\rho}}\big(\big\|\nabla p^{n+1}\|^2-\|\nabla p^{n}\big\|^2\big)\\
\leq&-\frac{\Delta t\mathrm{CaCn}}{2}\big\|\sqrt{\eta^{n+1}} D(\mathbf{u}^{n+1})\big\|^{2}
-\frac{\Delta t\mathrm{CaCn}}{\mathrm{L}_{s}}\big\|\sqrt{\eta^{n+1}/l_s^{n+1}}\mathbf{u}_{s}^{n+1}\big\|_{\Gamma}^{2}\\
&-\Delta t\int_\Omega(\phi^n\nabla\mu_{\phi}^{n+1}  + \psi^n\nabla\mu_{\psi}^{n+1} )\cdot\mathbf{u}^{n+1}\mathrm{d}x+\Delta t\mathrm{Cn}\int_\Gamma L_\phi^{n+1} \mathbf{u}_\tau^{n+1}\cdot \nabla_{\tau} \phi^n \mathrm{d}S,
\end{aligned}
\end{equation}

Finally, collecting \eqref{thm2_psi2}, \eqref{thm2_phi2}, \eqref{thm2_DBC2} and \eqref{thm2_NS7}, we arrive at the desired energy inequality.\quad$\Box$
\end{proof}

\begin{remark}\label{RK_psi}
At a first glance, the proposed scheme \eqref{scheme_psi}-\eqref{scheme_ch_psi} presents a difficulty in solving a nonlinear system in $\psi^{n+1}$ due to the implicit treatment of the mobility $M_\psi^{n+1}$. However, this nonlinearity is weak since the second order derivative term is linear in $\psi^{n+1}$ and nonlinearity only enters lower order terms, i.e., \eqref{scheme_psi}-\eqref{scheme_ch_psi} is equivalent to
\begin{equation}\label{scheme_psi0}
\frac{\psi^{n+1}-\psi^{n}}{\Delta t}+\nabla \cdot (\mathbf{u}^{n+1}\psi^{n})=\frac{\mathrm{Pi}}{\mathrm{Pe}_{\psi}} \Delta \psi^{n+1} + \frac{1}{\mathrm{Pe}_{\psi}} \nabla \cdot M_{\psi}^{n+1}\nabla\Big(\frac{1}{2 \mathrm{E x}} (\phi^{n})^{2}-\frac{1}{4}\big((\phi^{n})^{2}-1\big)^{2}\Big).
\end{equation}
The numerical stiffness of this equation mainly comes from the second order term $\Delta \psi$ which is already treated implicitly. To further remove the nonlinearity, we can replace $M_\psi^{n+1}$ in \eqref{scheme_psi0} by $M_\psi^n$, leading to a linear system in $\psi^{n+1}$,
\begin{equation}\label{scheme_psi1}
\frac{\psi^{n+1}-\psi^{n}}{\Delta t}+\nabla \cdot (\mathbf{u}^{n+1}\psi^{n})=\frac{\mathrm{Pi}}{\mathrm{Pe}_{\psi}} \Delta \psi^{n+1} + \frac{1}{\mathrm{Pe}_{\psi}} \nabla \cdot M_{\psi}^{n}\nabla\Big(\frac{1}{2 \mathrm{E x}} (\phi^{n})^{2}-\frac{1}{4}\big((\phi^{n})^{2}-1\big)^{2}\Big),
\end{equation}
associated with the boundary conditions  $\mathrm{Pi}\partial_n\psi^{n+1}+
M_{\psi}^{n}\partial_n\big((\phi^{n})^{2}/(2\mathrm{E x})-((\phi^{n})^{2}-1)^{2}/4\big)=0$ on $\Gamma$ and $\partial_n\psi^{n+1}=0$ on $\partial\Omega/\Gamma$. This recovers the scheme in \cite{gu2014energy} in the absence of velocity field $\mathbf{u}^{n+1}$. However, in comparison to \eqref{scheme_psi0}, the scheme \eqref{scheme_psi1} does not yields the discrete energy stability rigorously, since the following inequality does not hold:
\begin{equation*}
  \int_\Omega\bigg(M_\psi^{n+1}\nabla\Big(\mathrm{Pi}\ln \frac{\psi^{n+1}}{1-\psi^{n+1}}\Big) + M_\psi^{n}\nabla\Big(\frac{1}{2 \mathrm{E x}} (\phi^{n})^{2}-\frac{1}{4}\big((\phi^{n})^{2}-1\big)^{2}\Big)\bigg)\cdot\nabla\mu_\psi^{n+1}\mathrm{d}x\geq0.
\end{equation*}
It should be also noted that if the mobility in \eqref{scheme_psi}-\eqref{scheme_ch_psi} is treated explicitly as $M_\psi^n$, we recover the scheme in \cite{Wang2022}, which has been proved to be energy stable and bound-preserving in the absence of velocity field.

\end{remark}

\subsection{Second-order scheme}
Following the discussion of updating $\psi^{n+1}$ in Remark \ref{RK_psi}, we propose the linear and decoupled second-order scheme \eqref{scheme_psi2}--\eqref{scheme_pressure2} using second-order BDF discretization:

\vspace{2mm}

\emph{Step 1}: We first compute $\psi^{n+1}$, $\mu_\phi^{n+1}$, $\phi^{n+1}$, $\mu_\phi^{n+1}$, $\mathbf{u}^{n+1}$ by

\begin{align}
&\frac{3\psi^{n+1}-4\psi^n+\psi^{n-1}}{2\Delta t}+\nabla\cdot(\mathbf{u}^{*}\psi^*)=\frac{\mathrm{Pi}}{\mathrm{Pe}_{\psi}} \Delta \psi^{n+1} + \frac{1}{\mathrm{Pe}_{\psi}} \nabla \cdot M_{\psi}^{*}\nabla\Big(\frac{1}{2 \mathrm{E x}} (\phi^{*})^{2}-\frac{1}{4}\big((\phi^{*})^{2}-1\big)^{2}\Big),\label{scheme_psi2}\\
&\mu_{\psi}^{n+1}=\mathrm{P i} \ln \Big(\frac{\psi^{n+1}}{1-\psi^{n+1}}\Big)+\frac{1}{2 \mathrm{E x}} (\phi^*)^{2}-\frac{1}{4}\big((\phi^*)^{2}-1\big)^{2},\label{scheme_ch_psi2}\\
&\frac{3\phi^{n+1}-4\phi^n+\phi^{n-1}}{2\Delta t}+\nabla\cdot(\mathbf{u}^{*}\psi^*)=\frac{1}{\mathrm{P e}_{\phi}} \Delta \mu_{\phi}^{n+1},\label{scheme_phi2}\\
&\mu_{\phi}^{n+1}=-\mathrm{C n}^{2} \Delta \phi^{n+1}+s_1(\phi^{n+1}-\phi^*)+f(\phi^*)+\frac{1}{\mathrm{E x}} \psi^{n+1} \phi^{n+1} -\psi^{n+1}\big((\phi^*)^{3}-\phi^{n+1}\big),\label{scheme_ch_phi2}\\
&\begin{aligned}
&\rho^{n+1} \frac{3\mathbf{u}^{n+1}-4\mathbf{u}^{n} +\mathbf{u}^{n-1} }{2\Delta t}  + \rho^{n+1} \mathbf{u}^* \cdot \nabla \mathbf{u}^{*} +\mathbf{J}_\rho^{n+1}\cdot\nabla\mathbf{u}^{*}+\nabla (p^n+\frac{4}{3}\varphi^{n}-\frac{1}{3}\varphi^{n-1})\\
&\quad\quad\quad=\frac{1}{\mathrm{R e}} \big(\eta^{n+1}\Delta\mathbf{u}^{n+1} + \nabla\eta^{n+1}\cdot D(\mathbf{u}^{*})\big)
-\frac{1}{\mathrm{We}\mathrm{Cn}}(\phi^{n+1} \nabla \mu_{\phi}^{n+1}+ \psi^{n+1} \nabla\mu_{\psi}^{n+1}),
\end{aligned}\label{scheme_NS2}
\end{align}
with boundary conditions
\begin{align}
&\frac{3\phi^{n+1}-4\phi^{n}+\phi^{n-1}}{2\Delta t}+\mathbf{u}_{\tau}^{*}\cdot \nabla_{\tau} \phi^{*}=-\frac{1}{\mathrm{Pe}_s}L_\phi^{n+1}  &&\text { on } ~\Gamma, \label{scheme_DBC2}\\
&\frac{ \mathbf{u}_{s}^{n+1}}{\mathrm{L}_{s} l_{s}^{n+1}}=\frac{ L_\phi^{n+1} \nabla_{\tau} \phi^{n+1} }{\mathrm{Ca}\eta^{n+1}} -\partial_{n} \mathbf{u}_{\tau}^{n+1} &&\text { on }~ \Gamma,\label{scheme_GNBC2}\\
&\partial_{n} \mu_\phi^{n+1}=0,\quad\quad \mathrm{Pi}\partial_{n} \psi^{n+1}+M_{\psi}^{*}\partial_{n}\Big(\frac{1}{2 \mathrm{E x}} (\phi^{*})^{2}-\frac{1}{4}\big((\phi^{*})^{2}-1\big)^{2}\Big)=0,\quad\quad \mathbf{n}\cdot\mathbf{u}^{n+1} =0  &&\text { on } ~\Gamma, \label{scheme_BC2_1}\\
& \partial_{n} \phi^{n+1}=0,\quad\quad
\partial_{n} \mu_\phi^{n+1}=0, \quad\quad  \partial_{n} \psi^{n+1}=0, \quad\quad\partial_{n} \mathbf{u}_{\tau}^{n+1}=\mathbf{0},\quad\quad \mathbf{n}\cdot\mathbf{u}^{n+1}=0 &&\text { on }~ \partial\Omega/\Gamma, \label{scheme_BC2_2}
\end{align}
where
\begin{align}
&\psi^*=2\psi^{n}-\psi^{n-1}, \quad\quad \phi^*=2\phi^{n}-\phi^{n-1}, \quad\quad \mathbf{u}^{*}=2\mathbf{u}^{n}-\mathbf{u}^{n-1},\\
&L_\phi^{n+1} =\mathrm{Cn}\partial_{n} \phi^{n+1}+s_{2}(\phi^{n+1}-\phi^{*})+\gamma_{wf}^{\prime}(\phi^{*}), \label{scheme_DBC_L2}\\
&\rho^{n+1}=(1-\phi^{n+1}) / 2+\lambda_{\rho}(1+\phi^{n+1}) / 2 , \quad\quad\eta^{n+1}=(1-\phi^{n+1}) / 2+\lambda_{\eta}(1+\phi^{n+1}) / 2, \\
&l_s^{n+1}=(1-\phi^{n+1})/2+\lambda_{l_s}(1+\phi^{n+1})/2,\quad\quad M_\psi^*=\psi^*(1-\psi^*),\\
&
f(\phi^*)=\begin{cases}2(\phi^*+1), & \text { if }~ \phi^*<-1, \\ (\phi^*)^{3}-\phi^*, & \text { if }~-1 \leq \phi^* \leq 1, \\ 2(\phi^*-1), & \text { if } ~\phi^*>1,\end{cases}\quad\quad\mathbf{J}_\rho^{n+1}=\frac{1-\lambda_\rho}{2\mathrm{Pe}_\phi}\nabla\mu_\phi^{n+1}, \quad\quad \mathbf{u}_s^{n+1}=\mathbf{u}_\tau^{n+1}-\mathbf{u}_w.
\end{align}

\vspace{2mm}

\emph{Step 2}: Update $p^{n+1}$ by solving
\begin{equation}\label{scheme_pressure2}
\begin{aligned}
  &\Delta\varphi^{n+1}=\frac{3\bar{\rho}}{2\Delta t}\nabla\cdot\mathbf{u}^{n+1},\\
  &p^{n+1}=p^n+\varphi^{n+1}-\eta^{n+1}\nabla\cdot\mathbf{u}^{n+1},
\end{aligned}
\end{equation}
with $\bar{\rho}=\min(1,\lambda_\rho)$ and boundary condition
\begin{equation*}
  \mathbf{n}\cdot\nabla \varphi^{n+1}=0 \quad\quad\text { on } ~\partial\Omega.
\end{equation*}

The proposed scheme \eqref{scheme_psi2}--\eqref{scheme_pressure2} is decoupled for $\psi^{n+1}$, $\phi^{n+1}$, $\mathbf{u}^{n+1}$ and $p^{n+1}$, and constitutes linear systems for $\psi^{n+1}$, $\phi^{n+1}$,  $\mathbf{u}^{n+1}$ and $p^{n+1}$. Thus, they can be solved efficiently by using linear solvers.

\begin{remark}
The discrete energy law of the second-order schemes for the NS-PFS-MCL model is still open. The main difficulties arise from the singularity in Flory-Huggins potential, the nonlinear coupling $\psi\phi^2$ and $\psi(\phi^2-1)^2$ between phase-field variable and surfactant concentration, the nonlinear coupling between the variables $\psi$, $\phi$ and the velocity through convection terms and stress terms, and the coupling in the variable density and viscosity. The nonlinear coupling between the velocity
and the phase-field variable on the MCL boundary poses another challenge for proving the energy stability.
Even though the energy stability for the second-order scheme is not easy to be proved analytically, the numerical result in Sect.~\ref{sec_energy} has shown that it is quite stable with the property of energy decay.
\end{remark}

\subsection{Numerical implementation}\label{implementation}
We first regularize the logarithmic potential $G(\psi)=\psi \ln \psi+(1-\psi) \ln (1-\psi)$ in \eqref{E_sur} from the domain $(0,1)$ to $(-\infty,+\infty)$ to avoid the singularity. For any $\xi>0$, the regularized logarithmic potential \citep{MR1182511} can be written as
\begin{equation*}
\hat{G}(\psi)= \begin{cases}\psi \ln \psi+\frac{(1-\psi)^{2}}{2 \xi}+(1-\psi) \ln \xi-\frac{\xi}{2}, & \quad \text { if } ~ \psi \geq 1-\xi, \\ \psi \ln \psi+(1-\psi) \ln (1-\psi), & \quad \text { if } ~ \xi < \psi < 1-\xi, \\ (1-\psi) \ln (1-\psi)+\frac{\psi^{2}}{2 \xi}+\psi \ln \xi-\frac{\xi}{2}, & \quad \text { if } ~\psi \leq \xi,\end{cases}
\end{equation*}
and when $\xi\rightarrow0$, $\hat{G}(\psi)\rightarrow G(\psi)$. In the numerical experiments, we focus on the NS-PFS-MCL model formulated with the regularized function $\hat{G}(\psi)$. This modification makes the computation of $\mu_\psi^{n+1}$ in \eqref{scheme_ch_psi} and \eqref{scheme_ch_psi2} free from singularities. The $~\hat{}~$ is omitted in the notation for convenience.

\vspace{0.5cm}

The finite difference method is adopted for space discretization in our numerical implements.
The first-order scheme \eqref{scheme_psi}--\eqref{scheme_pressure} form a nonlinear coupled system.
Nevertheless, the coupling system can be solved by either using a simple sub-iteration process or decoupling the system with a lagged velocity for the convection
terms. In our simulations, instead of using sub-iterations,
we treat the nonlinear convection terms explicitly in \eqref{scheme_psi}--\eqref{scheme_DBC}, i.e., $\nabla\cdot(\mathbf{u}^{n+1}\psi^n)$ as $\nabla\cdot(\mathbf{u}^{n}\psi^n)$ in \eqref{scheme_psi}, $\nabla\cdot(\mathbf{u}^{n+1}\phi^n)$ as $\nabla\cdot(\mathbf{u}^{n}\phi^n)$ in \eqref{scheme_phi},  $\rho^{n+1} \mathbf{u}^n \cdot \nabla \mathbf{u}^{n+1} +\mathbf{J}_\rho^{n+1}\cdot\nabla\mathbf{u}^{n+1}$ as $\rho^{n+1} \mathbf{u}^n \cdot \nabla \mathbf{u}^{n} +\mathbf{J}_\rho^{n+1}\cdot\nabla\mathbf{u}^{n}$ in \eqref{scheme_NS}, $\mathbf{u}_{\tau}^{n+1}\cdot \nabla_{\tau} \phi^{n}$ as $\mathbf{u}_{\tau}^{n} \cdot \nabla_{\tau} \phi^{n}$ in \eqref{scheme_DBC}, which can decouple Cahn-Hilliard type equations \eqref{scheme_psi}--\eqref{scheme_ch_phi} from Navier-Stokes equation \eqref{scheme_NS}.
Then $\psi^{n+1}$, $\phi^{n+1}$, $\mathbf{u}^{n+1}$ and $p^{n+1}$ can be solved in order. We also decouple velocity components in \eqref{scheme_NS} by writing the viscous stress term as
\begin{equation*}
\nabla \cdot\big(\eta D(\mathbf{u})\big)=\eta \Delta \mathbf{u}+\nabla \eta \cdot D(\mathbf{u}),
\end{equation*}
where the first term is treated implicitly and the second term is treated explicitly.
This could greatly remove stiffness of viscous term, even though the second term would introduce some mild CFL constraint. The resulting system has only variable coefficients in diagonal positions, which can be solved efficiently by preconditioned conjugate gradient method.

\vspace{0.5cm}

We solve the system under the assumption of three-dimensional axisymmetry with respect to the $z$-axis. We have the following cylindrical domain
\begin{equation}\label{domain}
\Omega=\{ ~(r, z) ~|  ~ 0 \leq r \leq R, ~  0 \leq z \leq L~\},
\end{equation}
where $R$, $L$ are dimensionless parameters. Here, $r=0$ corresponds to the z-axis (the centerline), $z=0$ is the solid wall $\Gamma$, where the GNBC is imposed.
Correspondingly, $\mathbf{u} = (u_r, u_z)$, where $u_r$, $u_z$ denote velocities along the $r$, $z$ directions respectively.
Now, we introduce the mathematical expressions for the gradient,
divergence and Laplace operators in axisymmetric cylindrical coordinates:
\begin{equation*}
\begin{aligned}
&\nabla=\Big(\frac{\partial}{\partial r}, \frac{\partial}{\partial z}\Big), \quad\quad  \nabla \cdot\mathbf{u}=\frac{1}{r} \frac{\partial(r u_r)}{\partial r}+ \frac{\partial u_z}{\partial z},\quad \quad \Delta=\nabla \cdot \nabla=\frac{1}{r} \frac{\partial}{\partial r}\Big(r \frac{\partial}{\partial r}\Big)+\frac{\partial^{2}}{\partial z^{2}},\\
&\nabla \cdot\big[\eta D(\mathbf{u})\big]=\nabla \cdot\big[\eta\big(\nabla \mathbf{u}+(\nabla \mathbf{u})^{\top}\big)\big]=\left(\begin{array}{c}
\frac{2}{r} \frac{\partial}{\partial r}\big(r \eta \frac{\partial u_{r}}{\partial r}\big)-\frac{\eta}{r^{2}} u_{r}+\frac{\partial}{\partial z}\big(\eta \frac{\partial u_{r}}{\partial z}\big)+\frac{\partial}{\partial z}\big(\eta \frac{\partial u_{z}}{\partial r}\big) \\
\frac{1}{r} \frac{\partial}{\partial r}\big(r \eta \frac{\partial u_{r}}{\partial z}\big)+\frac{1}{r} \frac{\partial}{\partial r}\big(r \eta \frac{\partial u_{z}}{\partial r}\big)+2 \frac{\partial}{\partial z}\big(\eta \frac{\partial u_{z}}{\partial z}\big)
\end{array}\right).
\end{aligned}
\end{equation*}
Moreover,  the GNBC \eqref{scheme_DBC}--\eqref{scheme_BC} on $\Gamma$ ($z=0$) becomes
\begin{align}
&\frac{\phi^{n+1}-\phi^{n}}{\Delta t}+\partial_{r} (u_{r}^{n+1} \phi^{n})=-\frac{1}{\mathrm{Pe}_s}L_\phi^{n+1}, \\
&\frac{ u_{s}^{n+1}}{\mathrm{L}_{s} l_{s}^{n+1}}=\frac{ L_\phi^{n+1} \partial_{r} \phi^n }{\mathrm{Ca}\eta^{n+1}} -\partial_{z} u_{r}^{n+1}, \\
&\partial_{z} \mu_\phi^{n+1}=0,\quad\quad \mathrm{Pi}\partial_{z} \psi^{n+1}+M_{\psi}^{n}\partial_{z}\Big(\frac{1}{2 \mathrm{E x}} (\phi^{n})^{2}-\frac{1}{4}\big((\phi^{n})^{2}-1\big)^{2}\Big)=0,\quad\quad u_z^{n+1} =0,
\end{align}
with
\begin{equation*}
  L_\phi^{n+1} =\mathrm{Cn}\partial_{z} \phi^{n+1}+s_{2}(\phi^{n+1}-\phi^{n})+\gamma_{wf}^{\prime}\left(\phi^{n}\right),
  \quad\quad  u_{s}^{n+1}=u_{r}^{n+1}-u_w.
\end{equation*}
The remaining boundary conditions are also modified as follows:
\begin{equation*}
  \begin{aligned}
&\partial_z\phi^{n+1}=0,\quad\quad \partial_z\mu_\phi^{n+1}=0,\quad\quad \partial_z\psi^{n+1}=0,\quad\quad \partial_zu_r^{n+1}=0,\quad\quad u_z^{n+1}=0  &\quad\quad\quad\text{ on }~ z=L,\\
& \partial_r\phi^{n+1}=0,\quad\quad \partial_r\mu_\phi^{n+1}=0,\quad\quad \partial_r\psi^{n+1}=0,\quad\quad u_r^{n+1}=0,\quad\quad \partial_ru_z^{n+1}=0   &\quad\quad\quad\text{ on }~ r=0,\\
& \partial_r\phi^{n+1}=0,\quad\quad \partial_r\mu_\phi^{n+1}=0,\quad\quad \partial_r\psi^{n+1}=0,\quad\quad u_r^{n+1}=0,\quad \quad\partial_ru_z^{n+1}=0   &\quad\quad\quad\text{ on }~ r=R.
  \end{aligned}
\end{equation*}
Similarly, the second-order BDF scheme \eqref{scheme_psi2}--\eqref{scheme_pressure2} can be rewritten accordingly. Furthermore, we will introduce the spatial discretization for these schemes in \ref{appx}.

We use efficient iterative methods to solve the resulting schemes, e.g., generalized minimal residual method (GMRES) with the FFT preconditioner under the periodic boundary conditions. It should be noted that due to Neumann boundary conditions for $\psi^{n+1}$ and $p^{n+1}$, we can accelerate the computation using discrete cosine transform in $z$-direction. Then their $z$-components can be decoupled in spectral space, and the resulting linear system in $r$-direction can be efficiently solved by direct elimination. This acceleration can be further improved when the system is solved in Cartesian coordinates, where $\phi^{n+1}$ can also be solved using FFT efficiently.

\begin{remark}
The presented schemes are applicable to the cases of high density ratios and viscosity ratios due to the use of the pressure stabilization method in Navier-Stokes equations.
It is necessary to pay special attention to the discretization of convection terms for the high density ratio case.
Third order weighted essentially non-oscillatory (WENO) schemes \citep{Liu1994} could be used to reduce the undershoot and overshoot around the interface. Moreover, for the purpose of volume conservation in phase-field variable, the numerical truncation of $\phi$ will be also adopted, in which $\phi$ can be redistributed as follows \citep{Gao2014}:
\begin{enumerate}
\item Truncate the undershoot/overshoot values:
\begin{equation*}
\hat{\phi}= \begin{cases}1, & \quad \text{if}~\phi>1+10^{-7},
 \\ -1, &\quad \text{if}~\phi<-(1+10^{-7}).\end{cases}
\end{equation*}

\item Compute the difference $D$ of the two different total masses:
\begin{equation*}
D = \int_{\Omega} \phi \mathrm{d}x-\int_{\Omega} \hat{\phi} \mathrm{d}x.
\end{equation*}

\item Uniformly distribute $D$ at $N_{D}$ points as $D / N_{D}$, where $N_{D}$ is the total number of cells that fall into the interface region $(-0.999<\phi<0.999)$. 
\end{enumerate}
However, in our simulations, this redistribution is not used due to low density ratio except for the last two examples in Sect.~\ref{sec_Reality}.
\end{remark}

\section{Numerical simulations}\label{sec_Simulation}
In this section, we will show several numerical experiments to validate the accuracy and stability of the proposed schemes, and study the influence of surfactants on the droplet impact dynamics. We first obtain the convergence rates of the proposed first-order and second-order schemes in Sect.~\ref{sec_accuracy}. The property of energy decay of these schemes is numerically validated and the validity of the model is confirmed in Sect.~\ref{sec_energy}.
Then, in Sect.~\ref{sec_impact}, we simulate droplets impacting on the solid surfaces with different values of Reynolds number, Weber number and Young's angle in two cases: one with a clean droplet and the other with a contaminated droplet by surfactants. Moreover, we compare our numerical results to the experimental results and observe almost quantitative agreement in Sect.~\ref{sec_Reality}, which gives a prediction of droplet impact dynamics with surfactants in real experiments.

Unless otherwise specified, the default values of some parameters are chosen as follows,
\begin{align*}
& \mathrm{Cn}=0.01, \quad\quad\mathrm{Pe}_\phi=100, \quad\quad \mathrm{Pe}_s=0.002,\quad\quad \mathrm{Pi}=0.1841, \\
&  \mathrm{Ex}=1, \quad\quad  \lambda_\rho=0.1, \quad\quad \lambda_\eta=0.5, \quad\quad \lambda_{l_s}=1.
\end{align*}
The computational domain $\Omega$ is determined by \eqref{domain} with $R=L=1$, i.e., $(r,z)\in[0,1]\times[0,1]$, except for Sect.~\ref{sec_Reality}, in which the domain is selected according to the real experiments. The grids is $(n_r,n_z)=(N,N)$ and $h=1/N$ is the spatial resolution, where $N$ is chosen appropriately for different sections.
Different values of $N$ and $\Delta t$ are chosen for accuracy test in Sect.~\ref{sec_accuracy}, whereas we use a uniform mesh with $h=0.005$ and $\Delta t=10^{-4}$ in Sect.~\ref{sec_energy}--\ref{sec_impact}.
Moreover, we assume droplet impacts occur under the following conditions:
(i) in axisymmetric motion;
(ii) with non-ionic surfactant;
(iii) with low surfactant concentration.

\subsection{Accuracy test}\label{sec_accuracy}
In this section, we choose that
\begin{equation*}
\mathrm{Re}=20,\quad\quad  \mathrm{We}=2, \quad\quad \mathrm{L}_{s}=0.1, \quad\quad \theta_s=60^\circ,\quad\quad \mathrm{Pe}_\psi=10,
\end{equation*}
and set the following initial values:
\begin{equation*}
\begin{aligned}
  &\psi(r, z, t=0)=0.01,\\
&\phi(r, z, t=0)=\tanh \big(50\sqrt{2}\times(\sqrt{r^{2}+z^{2}}-0.5)\big),\\
&\mathbf{u}(r, z, t=0)=\mathbf{0}.
\end{aligned}
\end{equation*}

For the spatial convergence test,  $N=50$, 100, 200, 400, 800 are chosen and the numerical solutions are obtained by applying the first-order and second-order (in time) schemes at $t=0.4$. We choose the solution on the finest mesh as the reference, and use $\Delta t=10^{-5}$ for the temporal resolution so that the errors arising from the time discretization are negligible compared with the spatial errors.
As shown in Table \ref{T4.1}, for both schemes, the results demonstrate second order accuracy in space for all quantities $\psi$, $\phi$ and $\mathbf{u}$. 
\begin{table}[ht!]
\centering
\caption{Spatial errors and convergence rates for $\psi$, $\phi$ and $\mathbf{u}$ at $t=0.4$ with fixed $\Delta t=10^{-5}$ and different grids $N$. The errors are computed by comparing with the reference solution on the finest mesh with $N=800$. $e(\cdot)$ is the error in $l_2$ norm.}
	\begin{tabular}{ | c | c | c | c | c | c |  c | c | c |}
	\hline
  & \multicolumn{8}{|c|}{1st-order scheme}  \\
		\hline
		$N$ & $e(\psi)$	& order	 & $e(\phi)$	& order &  $e(u_r)$	& order & $e(u_z)$	& order\\
		\hline
		50 	    & 5.10E-2		&  -	 & 5.21E-2  	&  - 	 & 3.58E-2  	&  - 	    &  3.61E-2 	&  - \\
		\hline
		100     & 1.38E-2 		& 1.89	 & 1.27E-2 	    &  2.04  & 9.01E-3 	    &  1.99 	&  9.28E-3 	&  1.96 \\
		\hline
		200 	& 3.32E-3		& 2.05	 & 3.04E-3  	&  2.06	 & 2.03E-3  	&  2.15 	&  2.23E-3 	&  2.06\\
		\hline
		400 	& 7.70E-4		& 2.11	 & 7.09E-4  	&  2.10	 & 4.54E-4  	&  2.16 	&  4.58E-4 	&  2.28\\
		\hline
& \multicolumn{8}{|c|}{2nd-order scheme}  \\
		\hline
		$N$ & $e(\psi)$	& order	 & $e(\phi)$	& order &  $e(u_r)$	& order & $e(u_z)$	& order\\
		\hline
		50 	    & 4.10E-2		&  -	 & 4.61E-2  	&  - 	 & 3.01E-2  	&  - 	    &  3.04E-2 	&  - \\
		\hline
		100     & 1.02E-2 		& 2.01	 & 1.11E-2 	    &  2.05  & 7.68E-3   	&  1.97 	&  7.65E-3 	&  1.99 \\
		\hline
		200 	& 2.48E-3		& 2.04	 & 2.53E-3  	&  2.14	 & 1.80E-3  	&  2.09 	&  1.81E-3 	&  2.08\\
		\hline
		400 	& 5.35E-4		& 2.21	 & 6.27E-4  	&  2.01	 & 4.15E-4  	&  2.12 	&  3.97E-4 	&  2.19\\
		\hline
	\end{tabular}\label{T4.1}
\end{table}

In the convergence test for temporal discretization, the numerical solutions are obtained by applying the first-order and second-order schemes with different time steps $\Delta t = 1/200, 1/400, 1/800, 1/1600$ at $t=0.4$. For different $\Delta t$, the corresponding spatial resolution is $h=2\Delta t$.
The reference solution is obtained by applying the second-order scheme with $\Delta t=10^{-4}$ and $h=5\times10^{-4}$. As shown in Table \ref{T4.2}, the schemes  achieve first order and second order accuracy in time respectively.
\begin{table}[ht!]
\centering
\caption{Temporal errors and convergence rates for $\psi$, $\phi$ and $\mathbf{u}$ at $t=0.4$ with different step sizes $\Delta t$. Corresponding spatial resolution $h=2\Delta t$ with $h=1/N$. The errors are computed by comparing with the reference solution obtained by the second-order scheme with $\Delta t=10^{-4}$ and $h=5\times10^{-4}$. $e(\cdot)$ is the error in $l_2$ norm.}
	\begin{tabular}{ | c | c | c | c | c | c |  c | c | c |}
	\hline
  & \multicolumn{8}{|c|}{1st-order scheme}  \\
		\hline
		$\Delta t$ & $e(\psi)$	& order	 & $e(\phi)$	& order &  $e(u_r)$	& order & $e(u_z)$	& order\\
		\hline
		1/200   & 3.56E-2       &  -	 & 3.98E-2      &  -        &  2.51E-2 	&  - 	 & 2.66E-2  	&  - 	\\
		\hline
		1/400   & 1.84E-2		&  0.95  & 2.05E-2		&  0.96	 	&  1.26E-2 	&  0.99  & 1.40E-2  	&  0.93 	\\
		\hline
		1/800   & 7.96E-3 		&  1.21	 & 1.03E-2 		&  0.99		&  5.94E-3 	&  1.09  & 6.42E-3  	&  1.12 	\\
		\hline
		1/1600  & 3.51E-3		&  1.18	 & 4.84E-3		&  1.09     &  2.55E-3 	&  1.22	 & 2.68E-3  	&  1.26 	\\
		\hline
& \multicolumn{8}{|c|}{2nd-order scheme}  \\
		\hline
		$\Delta t$ & $e(\psi)$	& order	 & $e(\phi)$	& order &  $e(u_r)$	& order & $e(u_z)$	& order\\
		\hline		
        1/200    &	1.19E-2  	&  -     & 1.26E-2      &  -        &  9.37E-3 	&  -    & 9.89E-3       & -	 \\
		\hline
		1/400    &	3.21E-3  	&  1.89  & 3.47E-3		& 1.86  	&  2.49E-3 	&  1.91 & 2.58E-3		& 1.94	 	\\
		\hline
		1/800    &	7.86E-4   	&  2.03  & 8.92E-4 		& 1.96 	    &  6.19E-4 	&  2.01 & 6.81E-4 		& 1.92	 	 \\
		\hline
		1/1600   &	1.82E-4  	&  2.11  & 2.14E-4		& 2.06		&  1.58E-4 	&  1.97 & 1.73E-4		& 1.98	\\
				\hline
	\end{tabular}\label{T4.2}
\end{table}

\subsection{Energy stability and model validation}\label{sec_energy}
In this section, the values of the parameters and the initial conditions are the same as Sect.~\ref{sec_accuracy}, except that $\psi(r, z, t=0)=0.02 +0.01\chi$, where $\chi$ is a random variable uniformly distributed in $[0,1]$.
We perform the computation until $t=15$ so that the system almost reaches steady state.

Fig.~\ref{Total_Energy} reveals that the total energies always decay with time for both first-order and second-order schemes, and the difference between them is almost negligible. In addition, the evolution curves of different parts in total energies are shown in Fig.~\ref{Energy}. It can be observed that the adsorption energy $E_{ad}$ decreases with time, whereas the mixing entropy $E_{sur}$ increases with time. This reveals the adsorption of surfactant onto the two-phase interface, which can also be observed in Fig.~\ref{2D_final}.
The Ginzburg-Landau energy $E_{GL}$ increases and the wall energy $E_{wf}$ decreases until stable states are achieved, indicating the dilation of two-phase interface and wetting area in wetting process. From the evolution of kinetic energy $E_k$, we know that the nonzero velocity field is generated at the beginning according to the initial contact line motion, and attenuates rapidly around about $t=0.4$, after which the system achieves the quasi-static state.
In summary, the interplay among interface evolution, surfactant diffusion, adsorption kinetic, and contact line dynamics become dominant in a long period until the system is near the equilibrium state.
\begin{figure}[t!]
\centering
\begin{subfigure}{0.49\linewidth}
\centering
\includegraphics[scale=0.58]{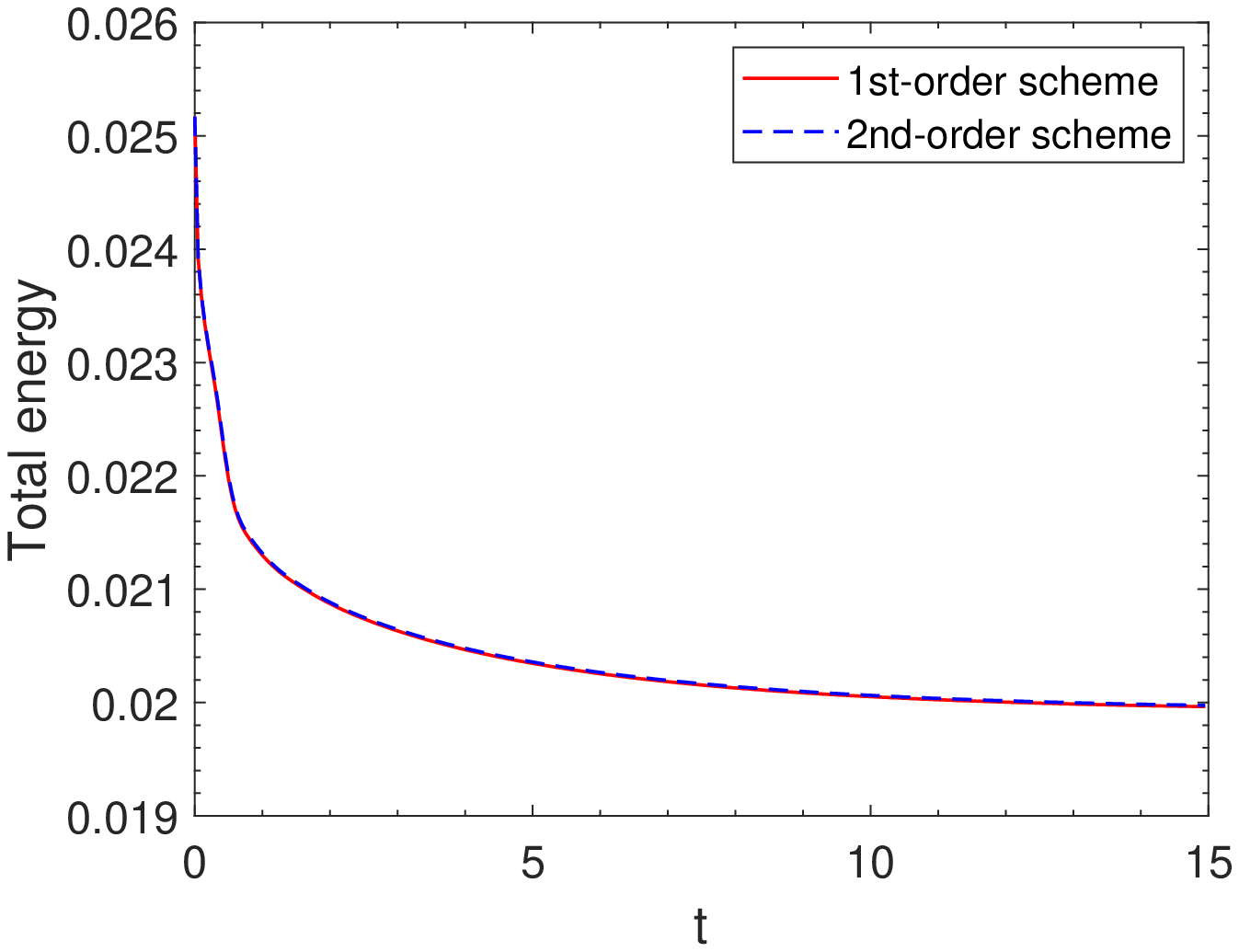}
\caption{}
    \label{Total_Energy}
\end{subfigure}
\begin{subfigure}{0.49\linewidth}
\centering
\includegraphics[scale=0.58]{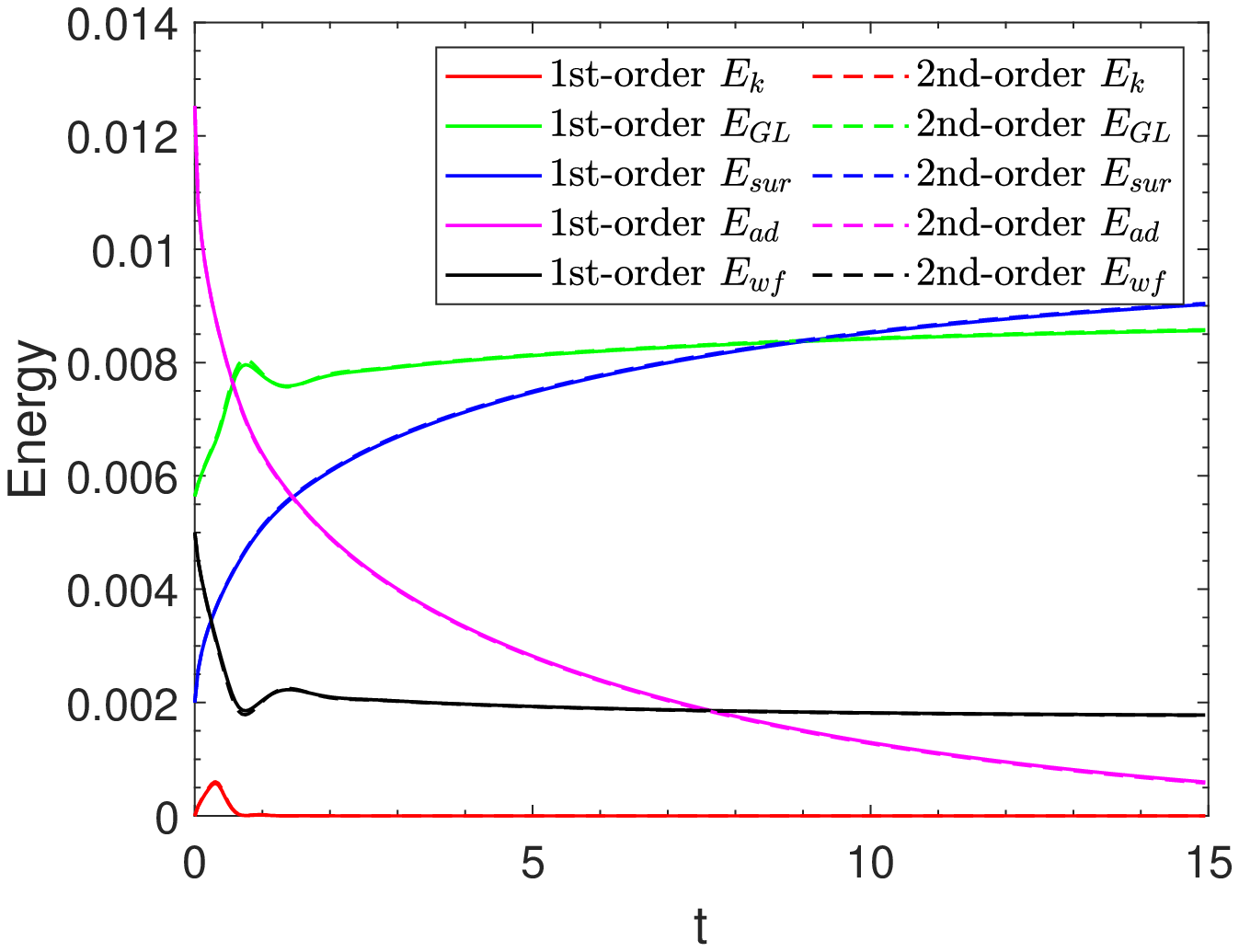}
\caption{}
    \label{Energy}
    \end{subfigure}
    \caption{Evolution of different types of energies obtained from first-order scheme (solid lines) and second-order scheme (dashed lines). (a) The total energy $E_{tot}$; (b) different parts in the total energy.}
\end{figure}
\begin{figure}[t!]
\centering
\begin{subfigure}{0.49\linewidth}
\centering
\includegraphics[scale=0.43]{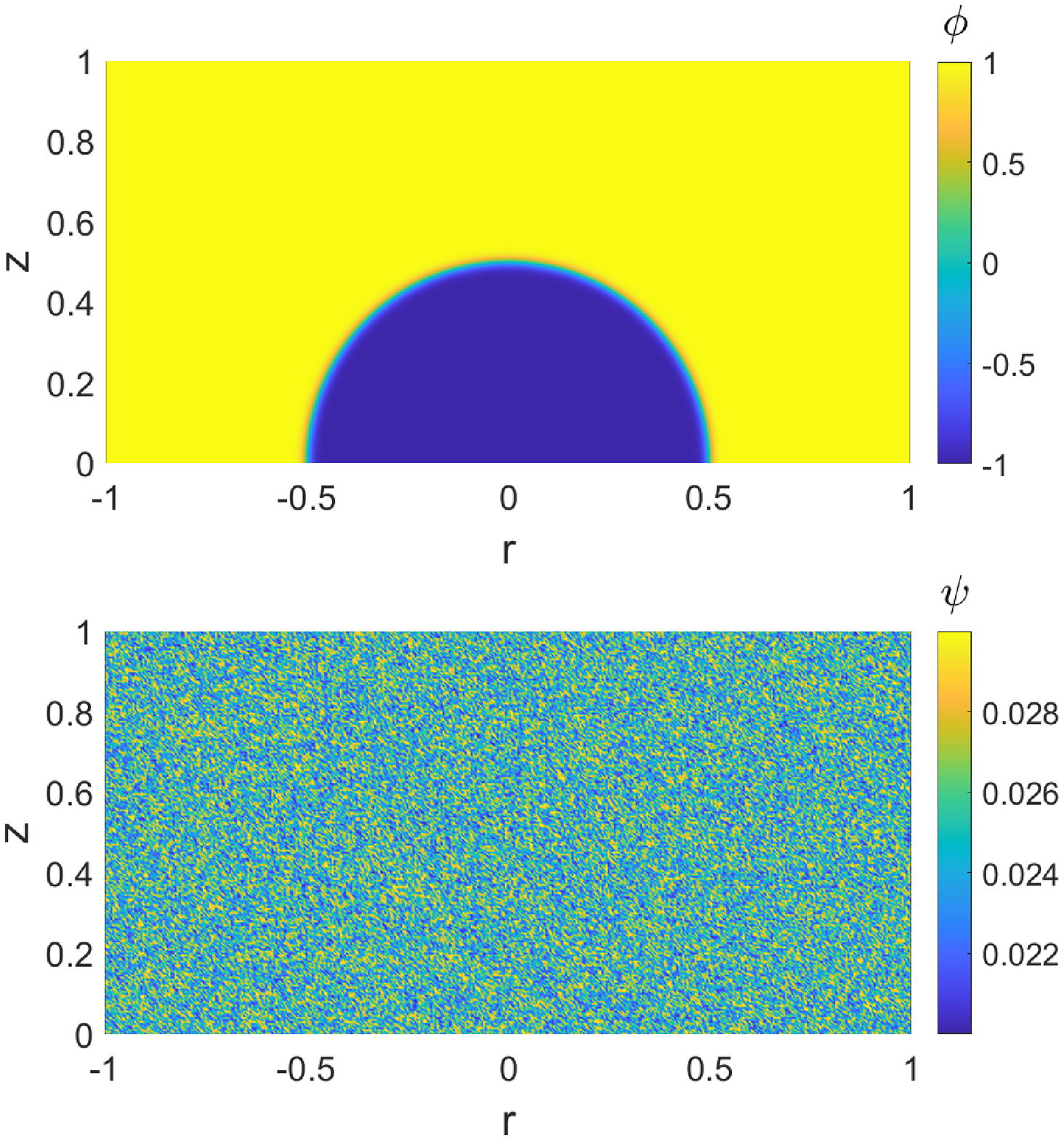}
\caption{}
    \label{2D_initial}
\end{subfigure}
\begin{subfigure}{0.49\linewidth}
\centering
\includegraphics[scale=0.43]{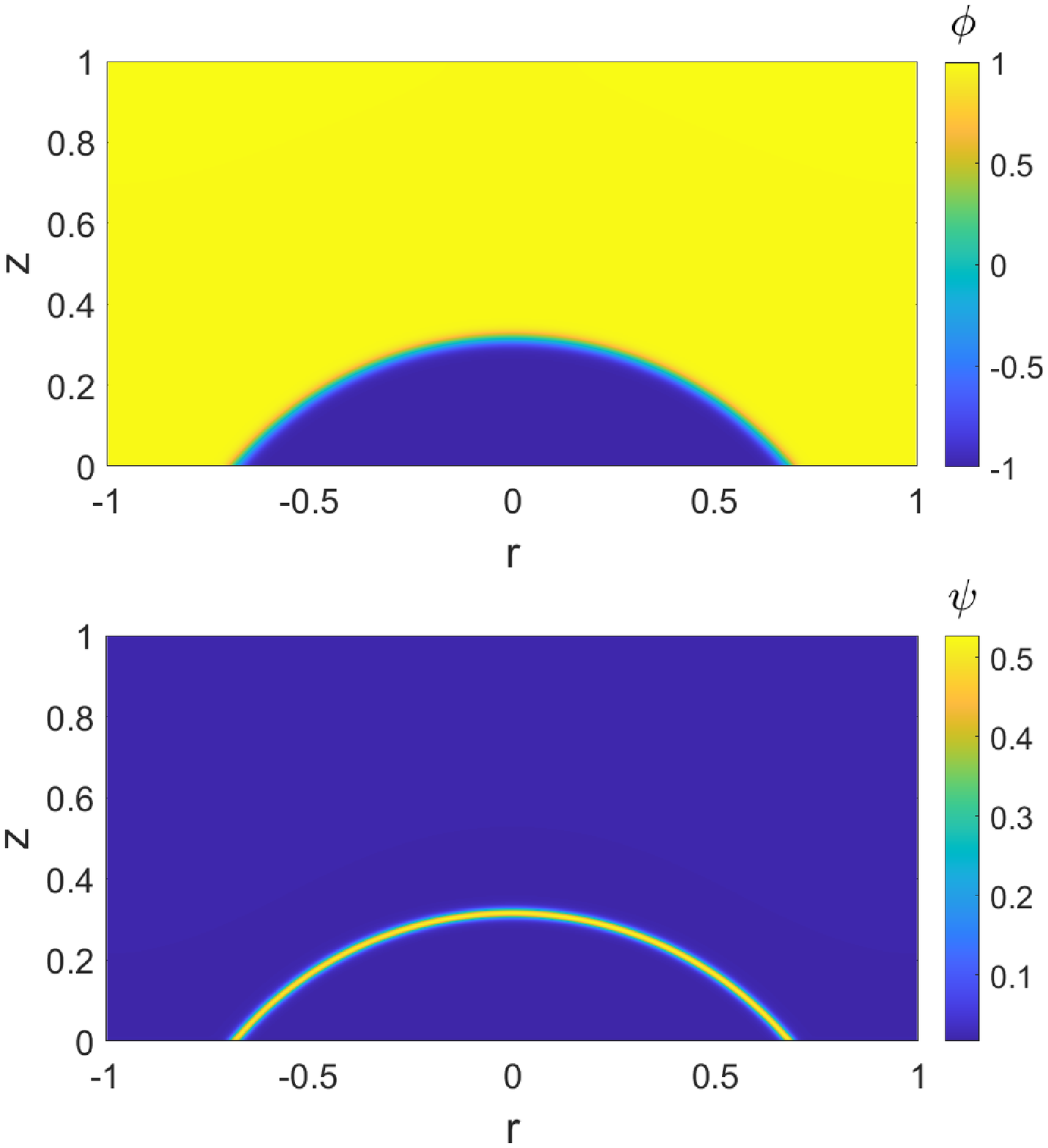}
\caption{}
    \label{2D_final}
    \end{subfigure}
    \caption{Cross-sectional plots of (a) initial droplet profile (upper panel) and surfactant distribution (lower panel); (b) final droplet profile (upper panel) and surfactant distribution (lower panel). Hydrophilic surface ($\theta_s=60^\circ$) is considered. The final state is at $t=15$.}
    \label{2D}
\end{figure}
\begin{figure}[t!]
\hspace{-21mm}
\includegraphics[scale=0.5]{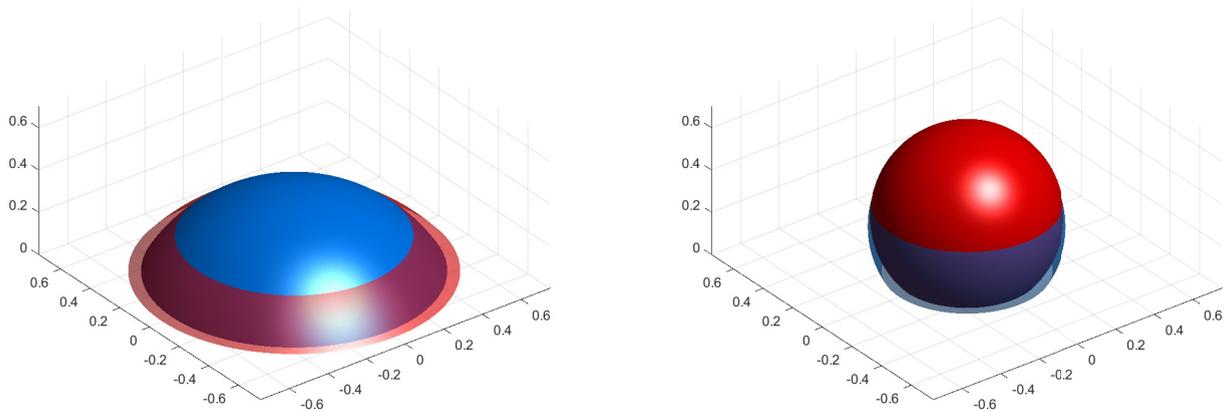}
    \caption{Comparison of equilibrium profiles of clean (blue sphere) and contaminated (red sphere) droplets on different surfaces: (left panel) hydrophilic surface with $\theta_s=60^\circ$; (right panel) hydrophobic surface with $\theta_s=120^\circ$. For contaminated droplets, the initial surfactant concentrations are random in $(0.02,0.03)$. In equilibrium, the contact angles of clean and contaminated droplets are respectively $59.3^\circ$ and $47.1^\circ$ for hydrophilic case, while they are respectively $120^\circ$ and $132.2^\circ$ for hydrophobic case.}
    \label{3D_final}
\end{figure}

For both hydrophilic ($\theta_s=60^\circ$) and hydrophobic ($\theta_s=120^\circ$) surfaces, two droplets (clean and contaminated) are selected for comparison to demonstrate the effect of surfactants on the dynamics of MCL.
To see how the droplets and the surfactant evolve, the initial and steady states of the droplet profiles and surfactant distributions are shown in Fig.~\ref{2D}, which are computed by using the second-order scheme \eqref{scheme_psi2}--\eqref{scheme_pressure2} until the steady state ($t=15$). In Fig.~\ref{2D_final}, the surfactant concentration around the interface is obviously higher than that in other regions in equilibrium.
The clean and contaminated droplets constantly spread or recoil until they reach their equilibrium states (shown in Fig.~\ref{3D_final}). Obviously, the presence of surfactant enhances hydrophilicity of a wetting droplet with its equilibrium contact angle decreasing to $47.1^\circ$, while it also helps droplet dewetting on hydrophobic surface with its equilibrium contact angle increasing to $132.2^\circ$.
In a word, surfactants affect wetting properties (such as equilibrium contact angle, contact line velocity) by altering surface tensions. This leads to many real applications, for example, removing grease from clothing using detergents.

\subsection{Droplet impact on a solid wall}\label{sec_impact}
In this section, we apply our second-order energy stable scheme \eqref{scheme_psi2}--\eqref{scheme_pressure2} to simulate a droplet impacting on a solid surface.
To demonstrate the influence of surfactants on the impact dynamics clearly, we compare the impact behaviors for clean and contaminated (with surfactants) droplets on hydrophilic and hydrophobic surfaces. Motivated by \cite{Zhang2016}, we study the processes by varying Reynolds number, Weber number and wettability of the solid surface. Different choices of these parameters will result in different phenomena, including adherence, partial bouncing, bouncing and splashing. Moreover, we control diffusion and adsorption ability of surfactants by varying $\mathrm{Pe}_\psi$. We show both three-dimensional plots of droplet profiles and their two-dimensional cross-sectional plots in radial direction. To illustrate the dynamics of fluids and surfactants, we also show the velocity fields and the concentration distribution of surfactant in cross-sectional plots. 

In all of the examples, a circular droplet with radius $r=0.25$ initially located at $(0, 0.25)$ is considered in two cases: one clean and the other contaminated with soluble surfactant accumulated on the interface of droplet. Here, this contaminated droplet is generated by independently evolving \eqref{scheme_psi2} from an initial uniform surfactant distribution (with $\psi=0.03$ everywhere) until steady state is reached, where $\phi$ is given in \eqref{initial_impact_phi} indicating a circular shape and $\mathbf{u}=\mathbf{0}$ is fixed. At the beginning, we assign the droplets with a downward initial velocity of absolute value 1. In summary, the initial conditions are mathematically expressed as
\begin{align}
&\phi(r, z, t=0)=\tanh \big(50\sqrt{2}\times(\sqrt{r^{2}+(z-0.25)^{2}}-0.25)\big),\label{initial_impact_phi}\\
&u_{r}(r, z, t=0)=0, \quad \quad u_{z}(r, z, t=0)=-0.5\times\Big(1-\tanh \big(50\sqrt{2}\times(\sqrt{r^{2}+(z-0.25)^{2}}-0.25)\big)\Big),\label{initial_impact_u}
\end{align}
and their profiles are shown in Fig.~\ref{impact_initial}. Initially, the surfactant concentrates on the interface and is uniformly distributed along its perimeter. Moreover, we fix the parameter $\mathrm{L}_{s}=0.0025$ in Examples 1--5, while $\mathrm{L}_{s}$ is chosen accordingly in Examples 6 and 7 in order to fit the experimental phenomena. The values of $\mathrm{Re}$, $\mathrm{We}$, $\theta_s$ and $\mathrm{Pe}_\psi$ will be specified below.
\begin{figure}[t!]
\center
\begin{overpic}[trim=0cm 0cm 0cm 0cm, clip,scale=0.36]{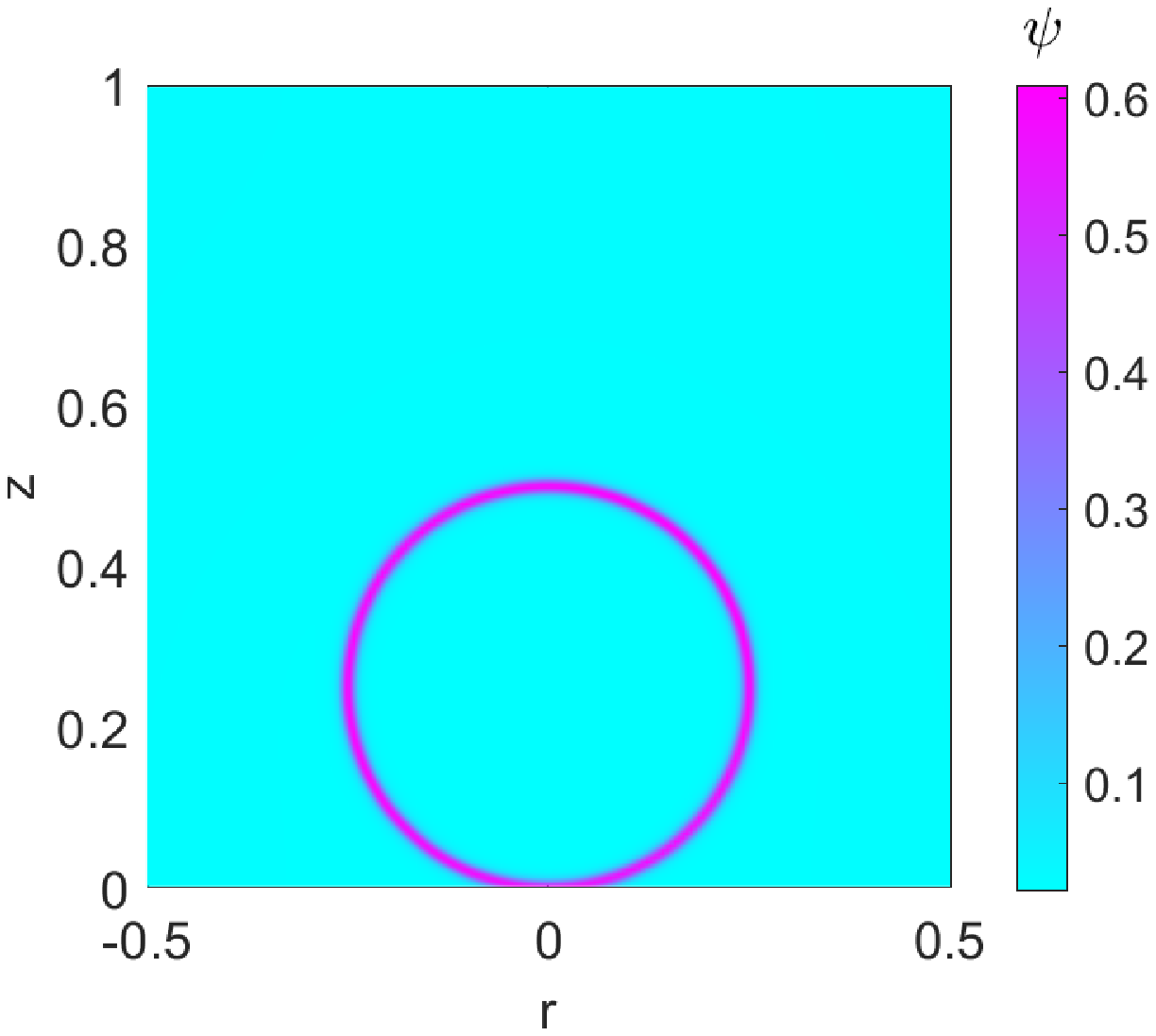}
\end{overpic}
\begin{overpic}[trim=0cm 0cm 0cm 0cm, clip,scale=0.36]{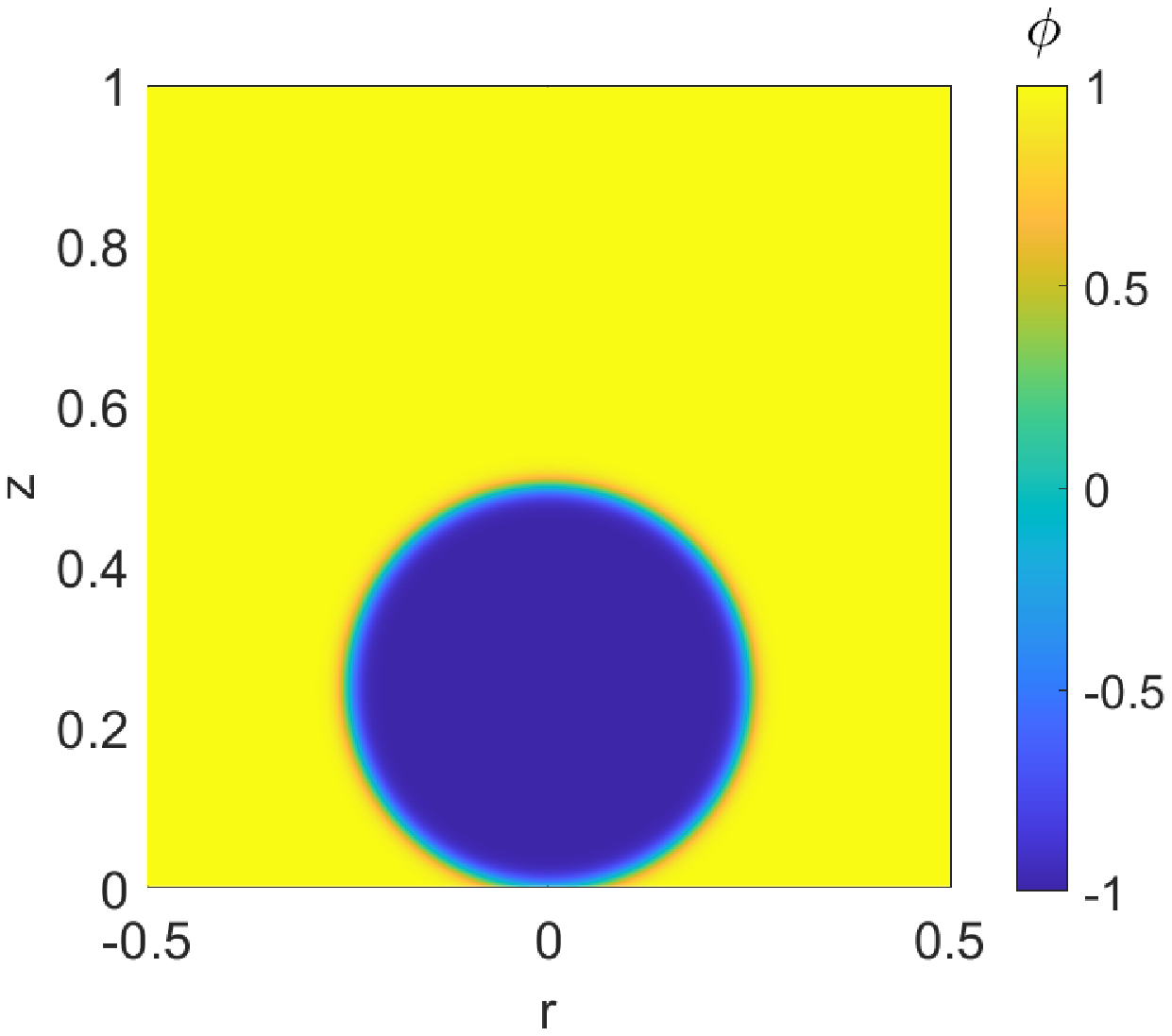}
\end{overpic}
\begin{overpic}[trim=0cm 0cm 0cm 0cm, clip,scale=0.36]{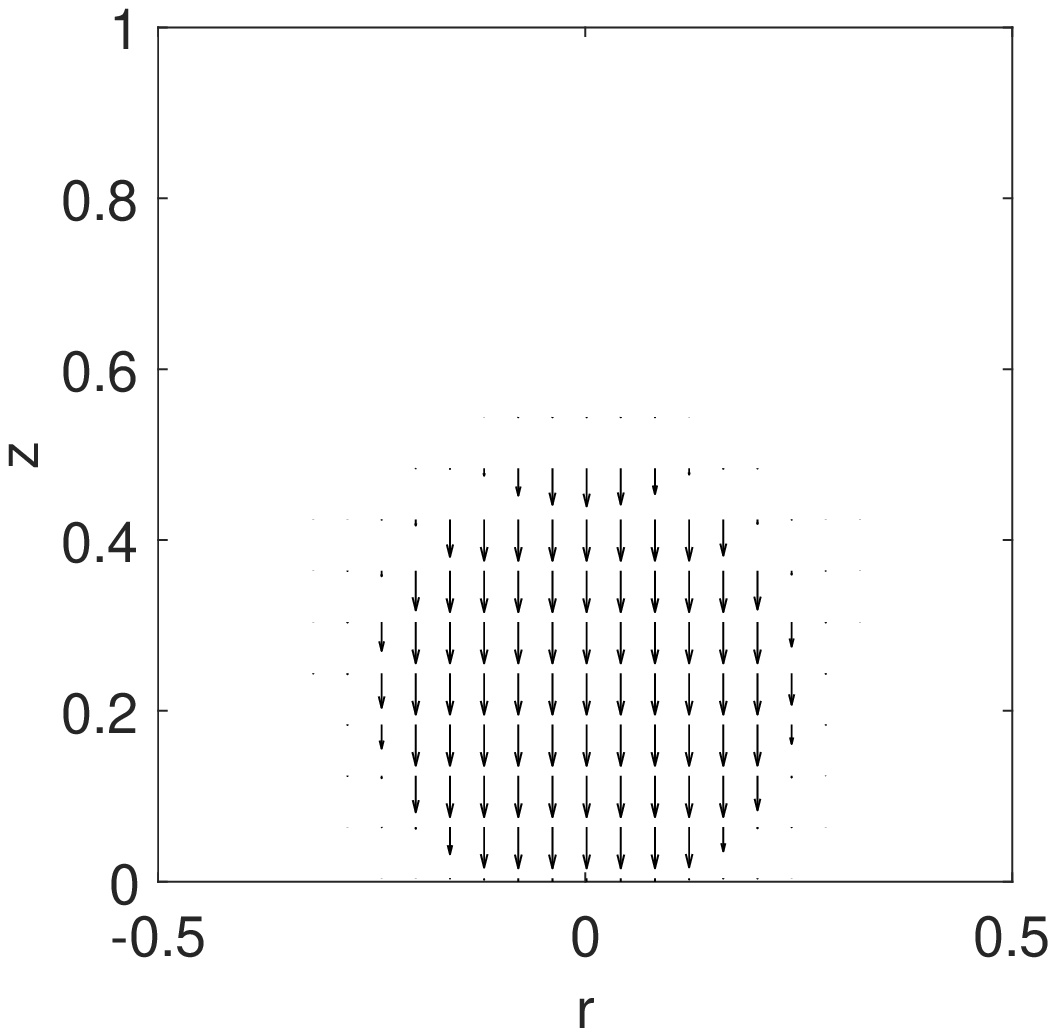}
\end{overpic}
\caption{From left to right, the initial profiles of $\psi$, $\phi$ and $\mathbf{u}$ in Examples 1--7 are shown.}\label{impact_initial}
\end{figure}

\subsubsection{Adherence}
In this section, we consider the first two examples for droplet adherence, which will happen when the initial kinetic energy is not large and the dissipation is large. In second example shown below, when surfactants are present, the outcome of the impact transits from adherence to complete bouncing.

\paragraph{Example 1} In this example, we choose the following parameters for the simulation:
\begin{equation*}
\mathrm{Re}=800, \quad\quad  \mathrm{We}=245.1, \quad\quad \theta_s=90^\circ,\quad\quad \mathrm{Pe}_\psi=100.
\end{equation*}
The evolutional profiles of droplet and surfactant concentration are depicted in Fig.~\ref{Adherence_2}.

As shown in Fig.~\ref{Adherence_2}a-c, both the clean and contaminated droplets behave in similar manners at the beginning ($t=1.6$). Due to contact line dynamics and the flow nearby the interface, the surfactant on the contaminated droplet concentrates at the moving front, leading to a lower surface tension. This lower surface tension makes the moving front of interface deformed more easily, which results in a smaller contact angle than that of the clean droplet. By uncompensated Young stress, smaller contact angle (and thus larger deviation from Young's angle) gives rise to larger contact line speed. Therefore, the contact line dynamics is accelerated and the front moves faster in the contaminated droplet. Moreover, the non-uniform distribution in surfactant concentration induces a gradient in surface tension and generates Marangoni force, which also contributes to the spreading dynamics.

When the droplets attain their maximal radii, they lose most of their kinetic energies while they store a large amount of surface energies. Then they start to release these surface energies through recoiling towards their centers. However, since surface tension is reduced due to the addition of surfactant, the contaminated droplet is more easily deformed and its center part becomes thinner than that of the clean droplet, which gives rise to breakup tendency as shown in Fig.~\ref{Adherence_2}d-f. When the  contaminated droplet breaks up from the center, a new contact line forms and recedes outwards as a result of uncompensated Young stress (Fig.~\ref{Adherence_2}h). The droplet behaves as a doughnut-like torus with shrinking width. In comparison, the clean droplet continues to recoil towards a spherical shape (Fig.~\ref{Adherence_2}g) and the difference between its shape and that of the contaminated one at $t=4$ is shown in Fig.~\ref{Adherence_2}i. The torus shape of the contaminated droplet does not last long due to surface tension, it continues evolving until a spherical shape is achieved in the end.

In this example, both the clean and contaminated droplets adhere on the substrate. The contaminated droplet experiences topological changes while the clean one does not. The differences between the two processes can also be read from the plots of dissipations. As illustrated in Fig.~\ref{Adherence_2_dissi}, both viscous ($R_v$) and slip ($R_s$) dissipations decay after initial peaks during droplet impact. They dominate the other dissipations and gently vary in spreading dynamics. In the spreading stage, the diffusion ($R_d$) and relaxation ($R_r$) dissipations become important in the sense of encoding interfacial evolution and identifying topological changes, although they are not large in magnitude. For instance, when topological change happens around $t=3.08$ and a `contaminated' torus forms, an increasing in $R_d$ (and also in $R_r$) is observed (dashed pink curve in the inset plot of Fig.~\ref{Adherence_2_dissi_a}). Therefore, we can easily distinguish the contaminated case from the clean case by simply looking at these two dissipations. In addition, some particular interface profiles can also be reflected when these two dissipations experience rapid changes (See Fig.~\ref{Adherence_2_dissi_b} and \ref{Adherence_2_dissi_c}).

\begin{figure}[t!]
\center\begin{overpic}[trim=0cm 0cm 0cm 0cm, clip,scale=0.3]{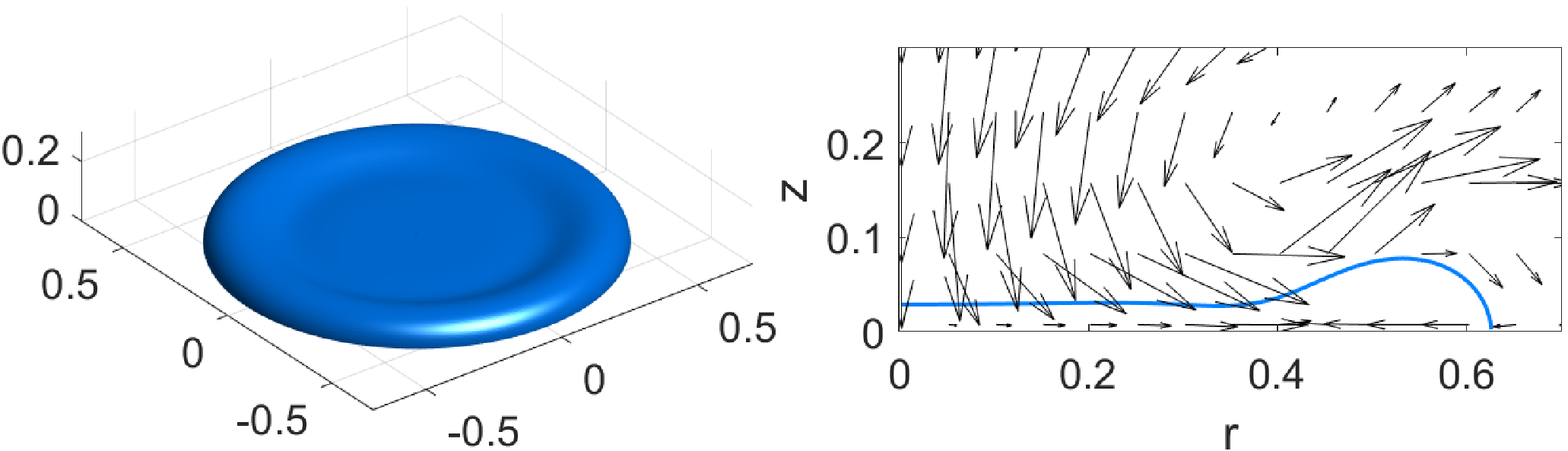}
\put(2,26){\scriptsize{(a)}}
\end{overpic}
\hspace{-0.6cm}
\begin{overpic}[trim=0cm 0cm 0cm 0cm, clip,scale=0.3]{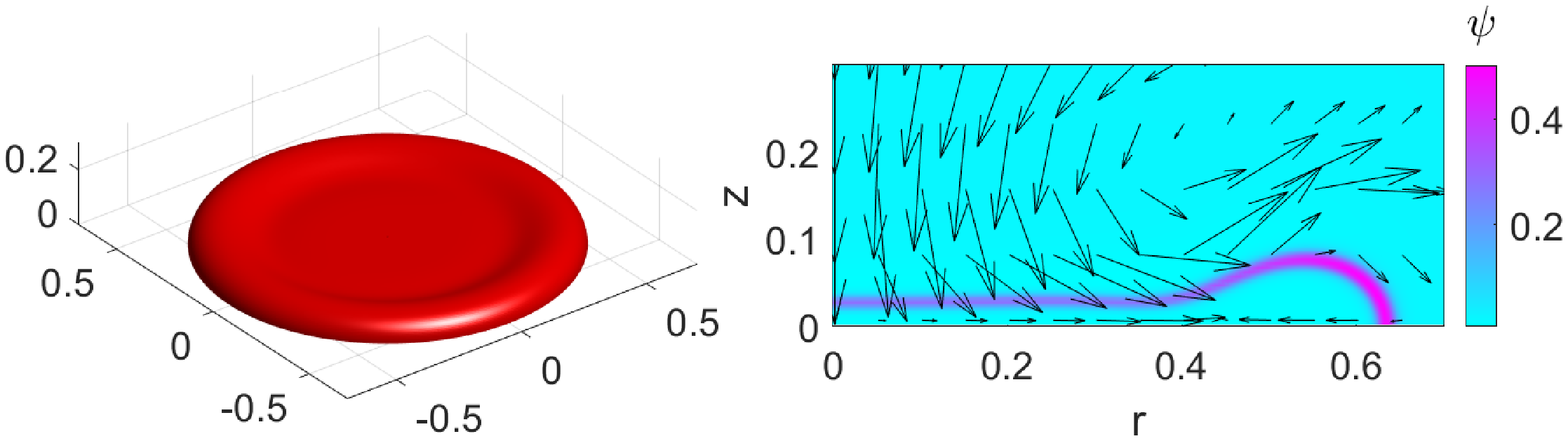}
\put(2,26){\scriptsize{(b)}}
\end{overpic}
\hspace{-0.1cm}
\begin{overpic}[trim=0cm 0cm 0cm 0cm, clip,scale=0.3]{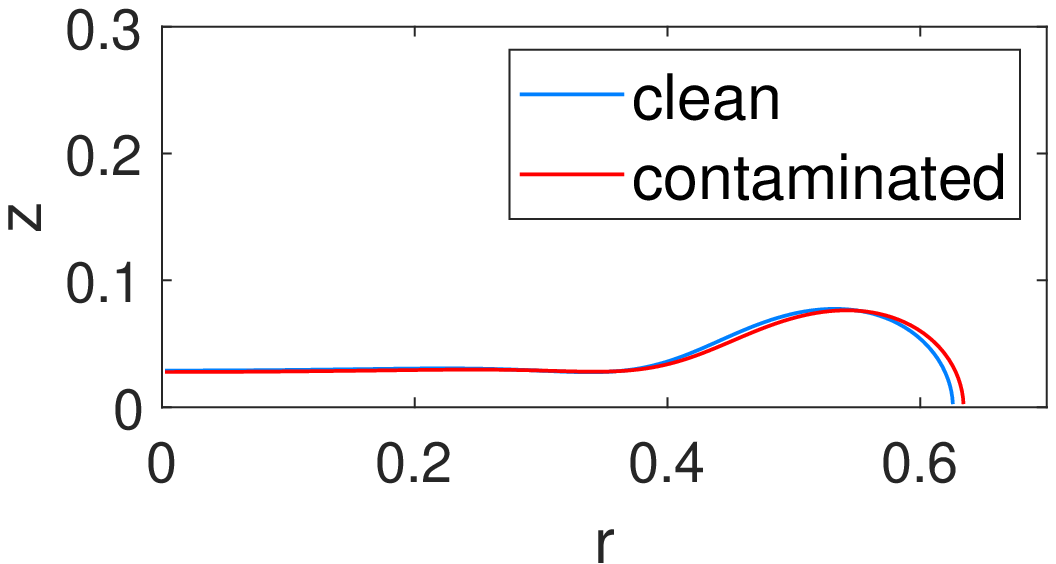}
\put(0,56){\scriptsize{(c)}}
\end{overpic}

\begin{overpic}[trim=0cm 0cm 0cm 0cm, clip,scale=0.3]{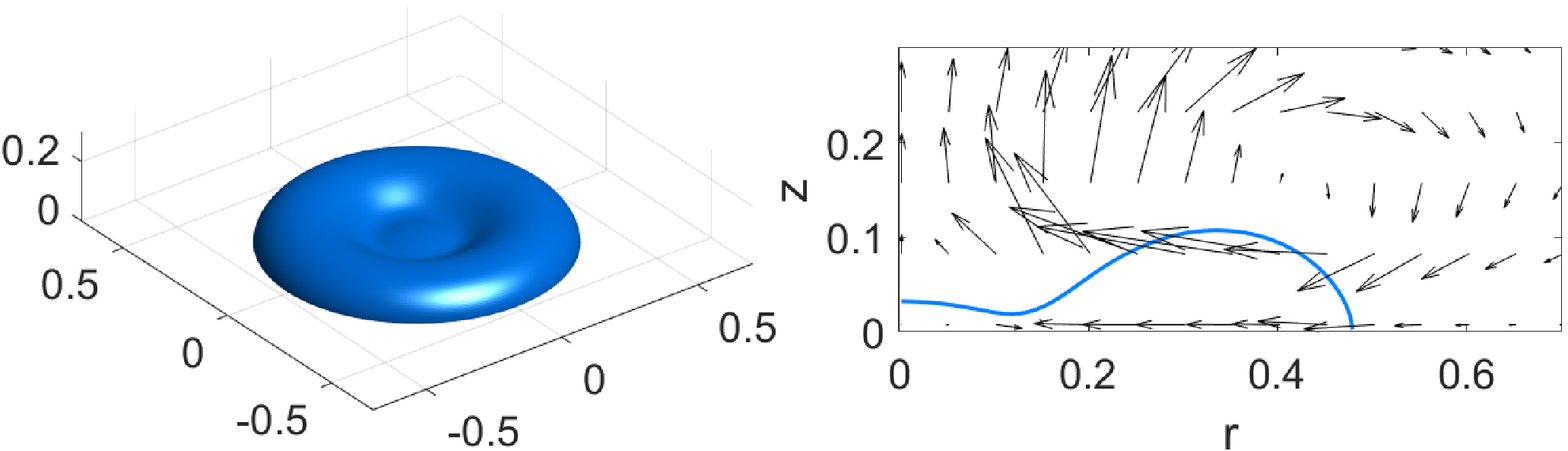}
\put(2,26){\scriptsize{(d)}}
\end{overpic}
\hspace{-0.6cm}
\begin{overpic}[trim=0cm 0cm 0cm 0cm, clip,scale=0.3]{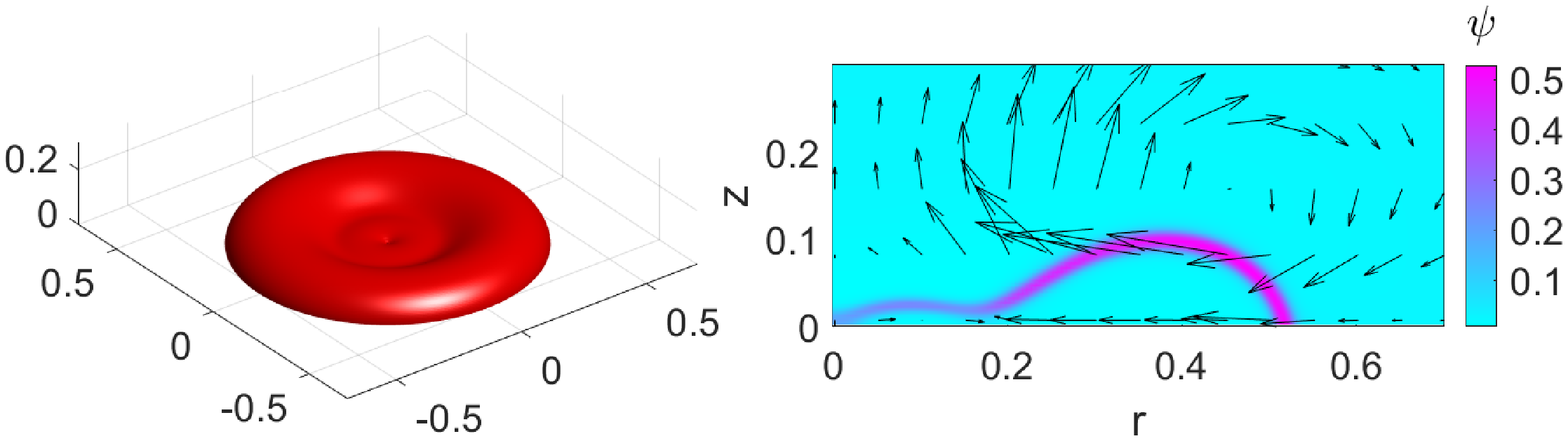}
\put(2,26){\scriptsize{(e)}}
\end{overpic}
\hspace{-0.1cm}
\begin{overpic}[trim=0cm 0cm 0cm 0cm, clip,scale=0.3]{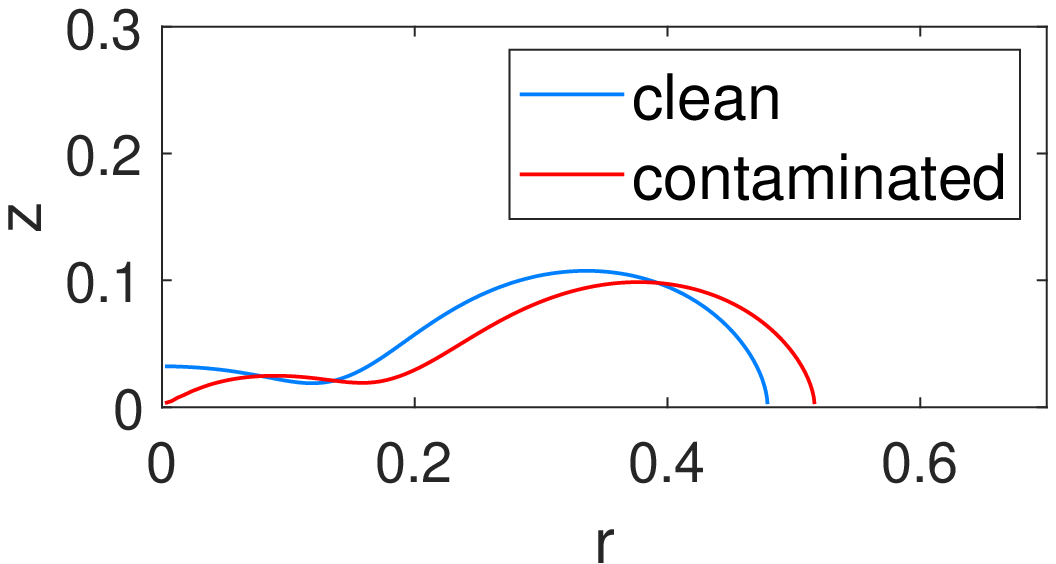}
\put(0,56){\scriptsize{(f)}}
\end{overpic}

\begin{overpic}[trim=0cm 0cm 0cm 0cm, clip,scale=0.3]{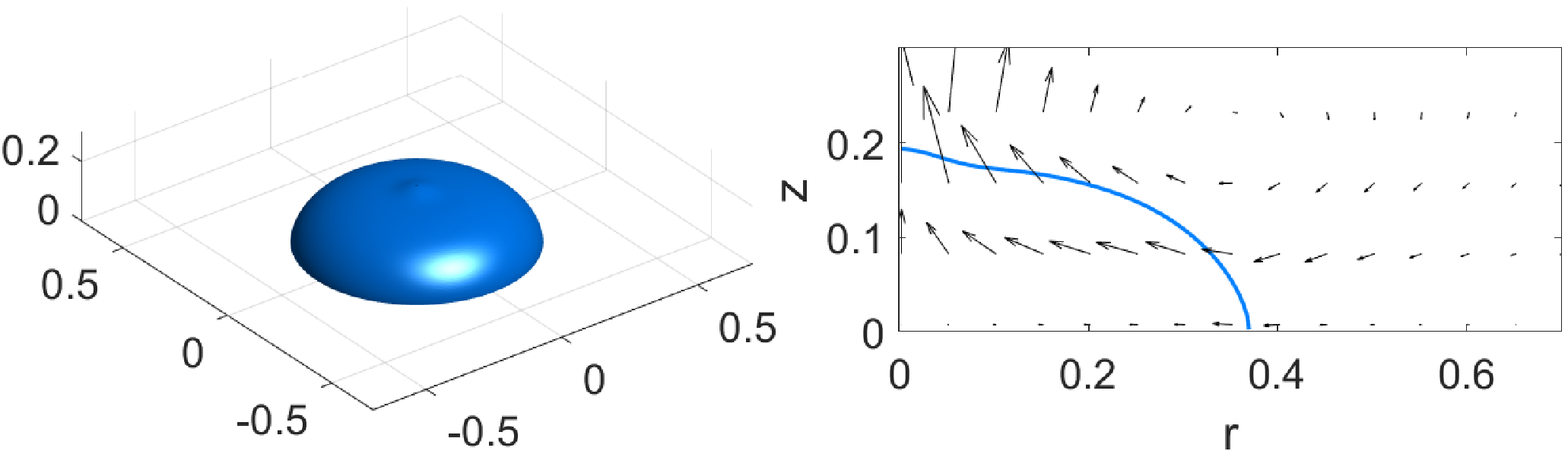}
\put(2,26){\scriptsize{(g)}}
\end{overpic}
\hspace{-0.6cm}
\begin{overpic}[trim=0cm 0cm 0cm 0cm, clip,scale=0.3]{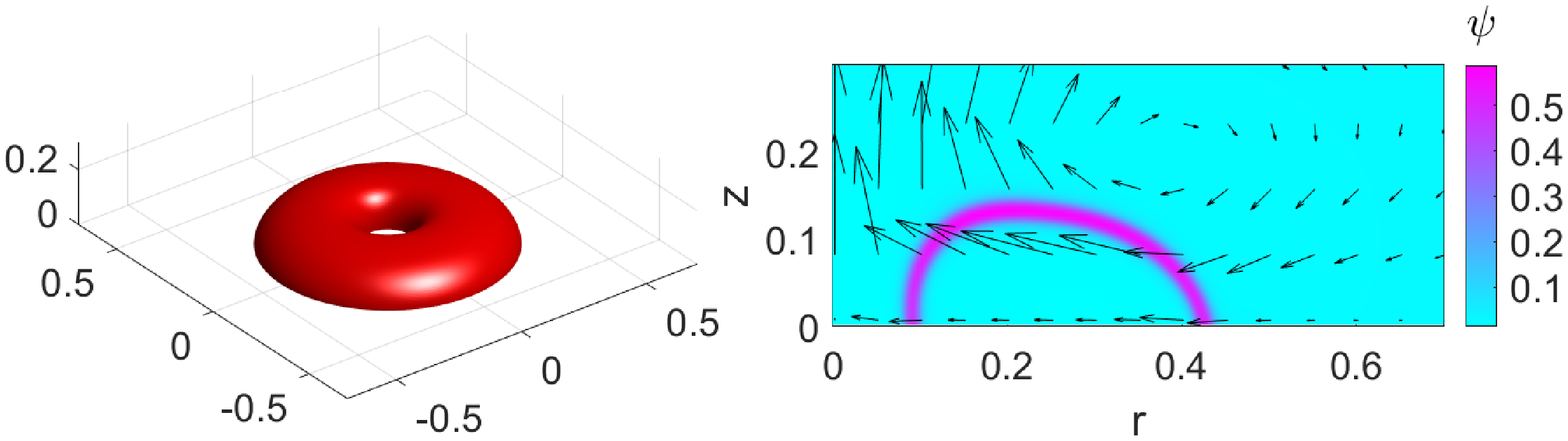}
\put(2,26){\scriptsize{(h)}}
\end{overpic}
\hspace{-0.1cm}
\begin{overpic}[trim=0cm 0cm 0cm 0cm, clip,scale=0.3]{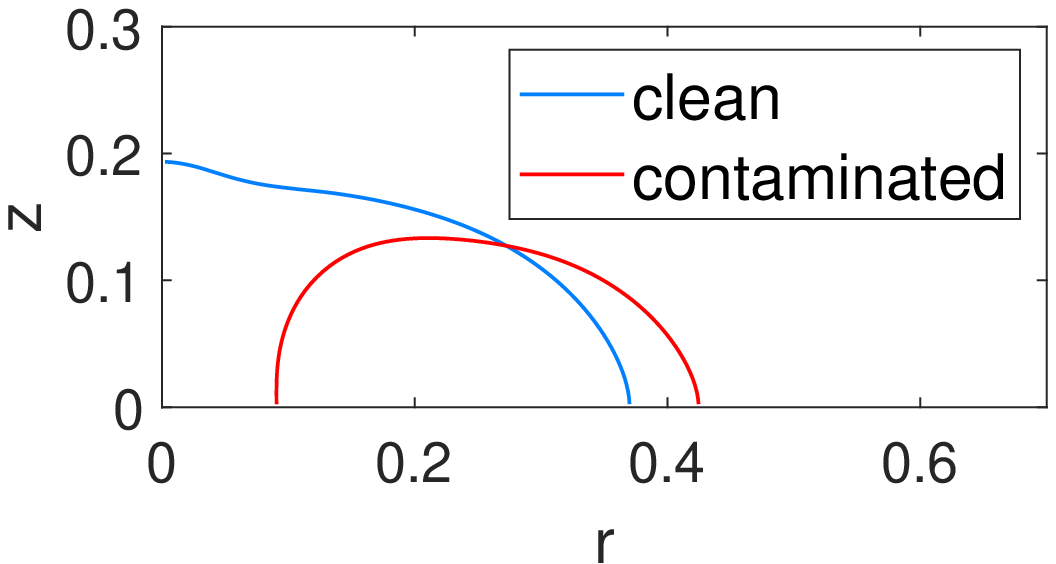}
\put(0,56){\scriptsize{(i)}}
\end{overpic}

\caption{Droplet adherence in Example 1 ($\mathrm{Re}=800$, $\mathrm{We}=245.1$, $\theta_s=90^\circ$, and $\mathrm{Pe}_\psi=100$). Profiles of clean and contaminated droplets are shown in both three-dimensional views (1st and 3rd columns) and two-dimensional radial plots ((2nd and 4th columns). Velocity fields and surfactant concentrations are also shown in the radial plots by quivers and colormaps respectively. Comparisons between the interface shapes of clean and contaminated droplets are given in the 5th column. Snapshots are captured at time $t = 1.6$ (a-c), $t = 3.08$ (d-f), and $t = 4$ (g-i).}
\label{Adherence_2}
\end{figure}
\begin{figure}[ht!]
\centering
\begin{subfigure}{1\linewidth}
\centering
\includegraphics[scale=0.7]{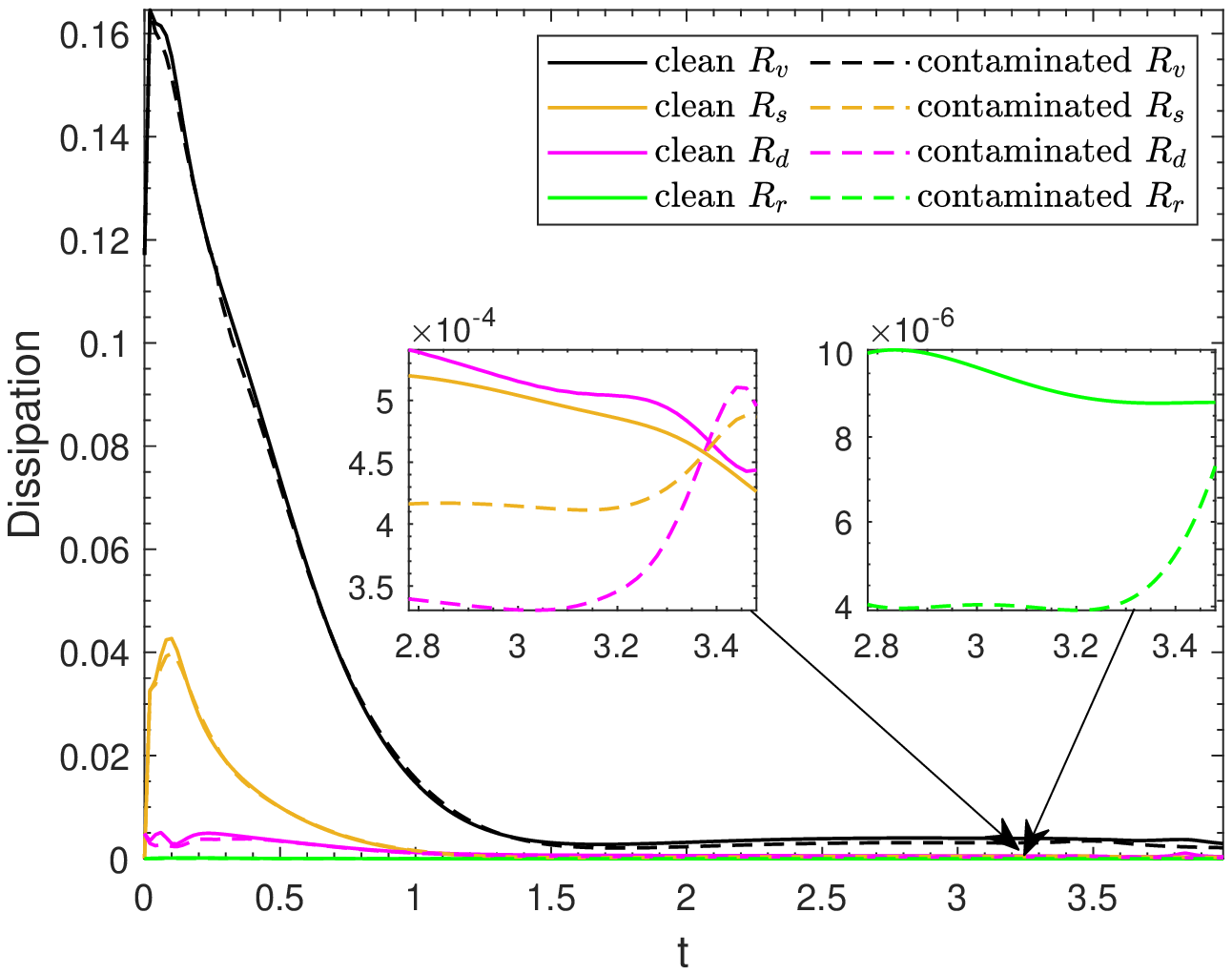}
\caption{}
\label{Adherence_2_dissi_a}
\end{subfigure}

\begin{subfigure}{0.49\linewidth}
\centering
\includegraphics[scale=0.58]{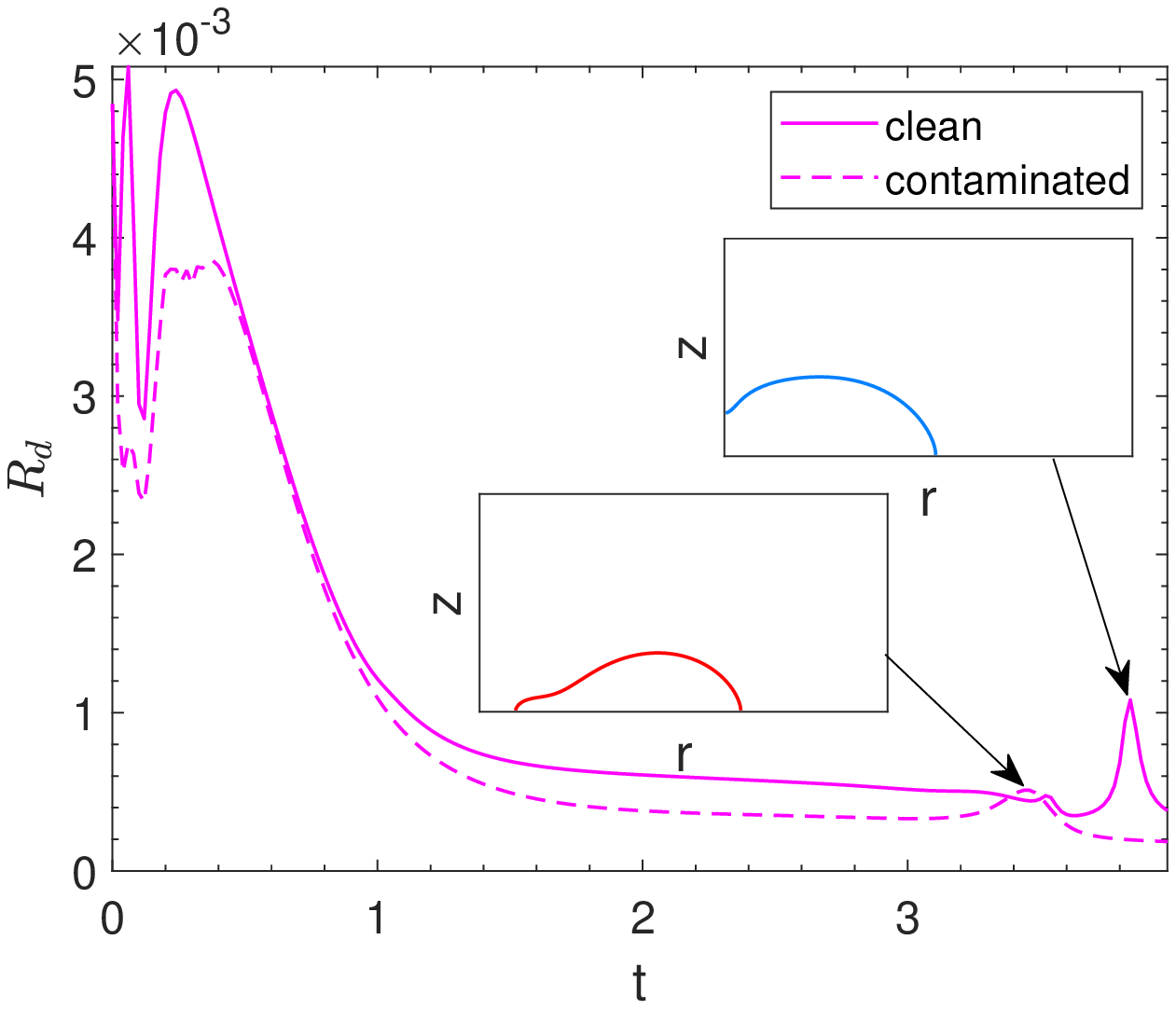}
\caption{}
\label{Adherence_2_dissi_b}
\end{subfigure}
\begin{subfigure}{0.49\linewidth}
\centering
\includegraphics[scale=0.58]{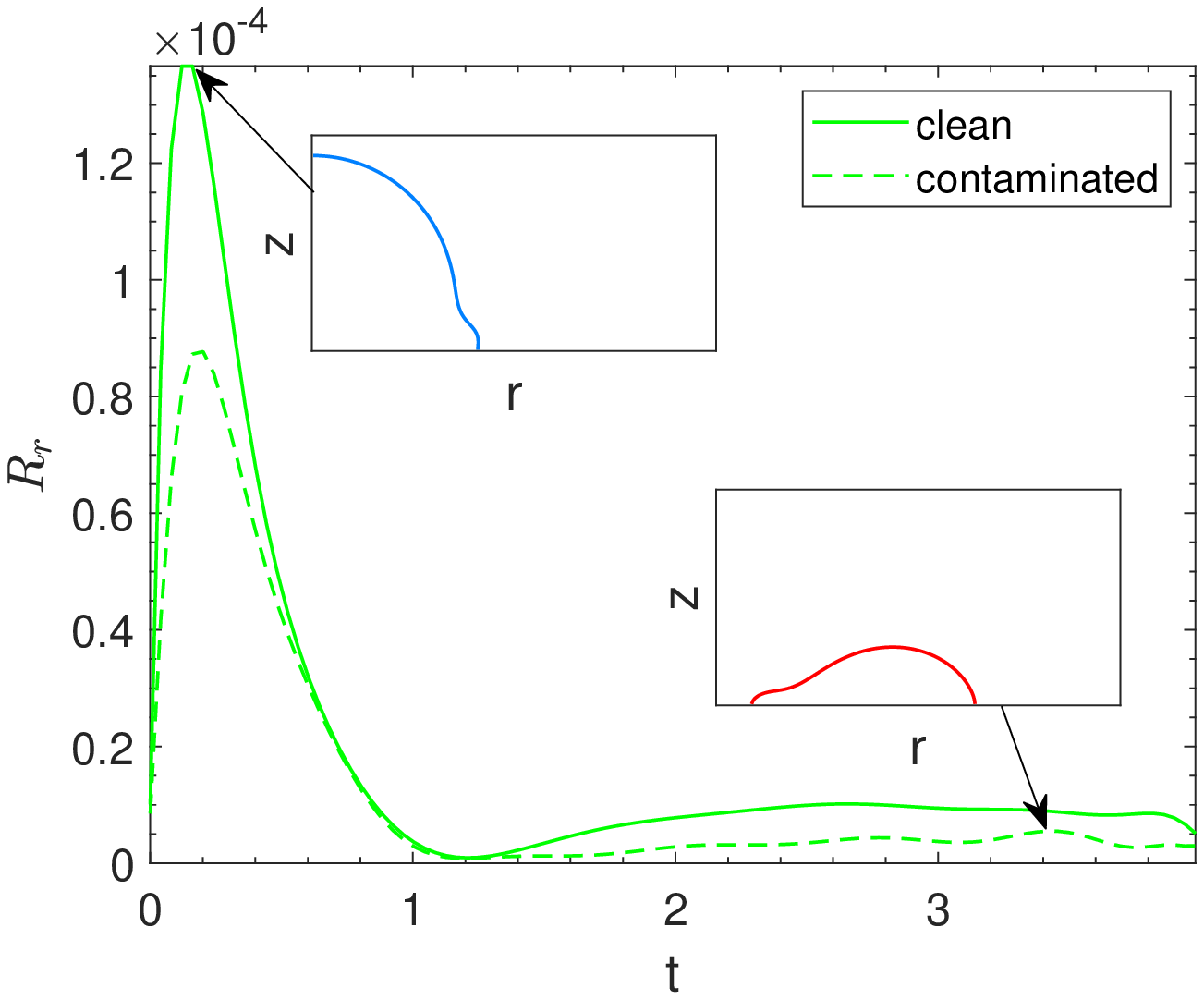}
\caption{}
\label{Adherence_2_dissi_c}
\end{subfigure}
    \caption{Dissipations in Example 1. In (a), four dissipations are shown for both clean (solid curves) and contaminated (dashed curves) cases, including viscous dissipation $R_v$, slip dissipation $R_s$, diffusion dissipation $R_d$, and relaxation dissipation $R_r$. The inset plots illustrate differences in dissipations for the two cases when topological change occurs. In particular, $R_d$ and $R_r$ are shown in (b) and (c), where their rapid changes indicate some particular interface geometries in the inset plots. }
\label{Adherence_2_dissi}
\end{figure}

\paragraph{Example 2} In this example, we strengthen both inertial and capillary effects by increasing Reynold number and decreasing Weber number as follows:
\begin{equation*}
\mathrm{Re}=1600, \quad\quad  \mathrm{We}=18.9, \quad\quad \theta_s=90^\circ,\quad\quad \mathrm{Pe}_\psi=10,
\end{equation*}
where P\'{e}cl\'{e}t number is taken to be smaller for the purpose of improving surfactant effect by enhancing adsorption.

Similar as in Example 1, after the initial impact, both clean and contaminated droplets spread until they achieve their maximal radii. This is followed by droplet recoiling  during which topological changes occur in both cases and a hole generates in the center (Fig.~\ref{Adherence_3}a-c). There is a slight difference that in the contaminated case a new island forms in the center. This may be attributed to the lower surface tension caused by surfactant, which makes the interface more easily deformed and the contaminated droplet `softer'.

Due to large inertial effect, the droplet brings a lot of kinetic energy (more than that in Example 1) and would bounce. However, the large capillary force makes the droplet stiffer and less deformed. As a result, the clean droplet is `dragged' and does not bounce. The balance between capillary and inertial effects leads to an oscillatory
motion with amplitude decreasing in time (Fig.~\ref{Adherence_3}d,g). In contrast, in the contaminated case, the presence of surfactant weakens capillary force and makes the droplet `softer'. As a result of droplet deformation, the upward velocity field generates nearby the interface, which provides a lifting flow driving the droplet off entirely from the solid surface at time $t = 3.3$ (Fig.~\ref{Adherence_3}e,h). A detailed comparison in the interface shapes is also given in Fig.~\ref{Adherence_3}f,i.

Again, rapid changes (and peaks) in dissipations indicate topological changes, such as formation of torus and island, and droplet bouncing away from the substrate (Fig.~\ref{Adherence_3_dissi}). Moreover, it is also observed that each dissipation in the contaminated case is overall smaller than the corresponding one in the clean case. This explains why the contaminated droplet can bounce while the clean one cannot.
\begin{figure}[t!]
\center
\begin{overpic}[trim=0cm 0cm 0cm 0cm, clip,scale=0.36]{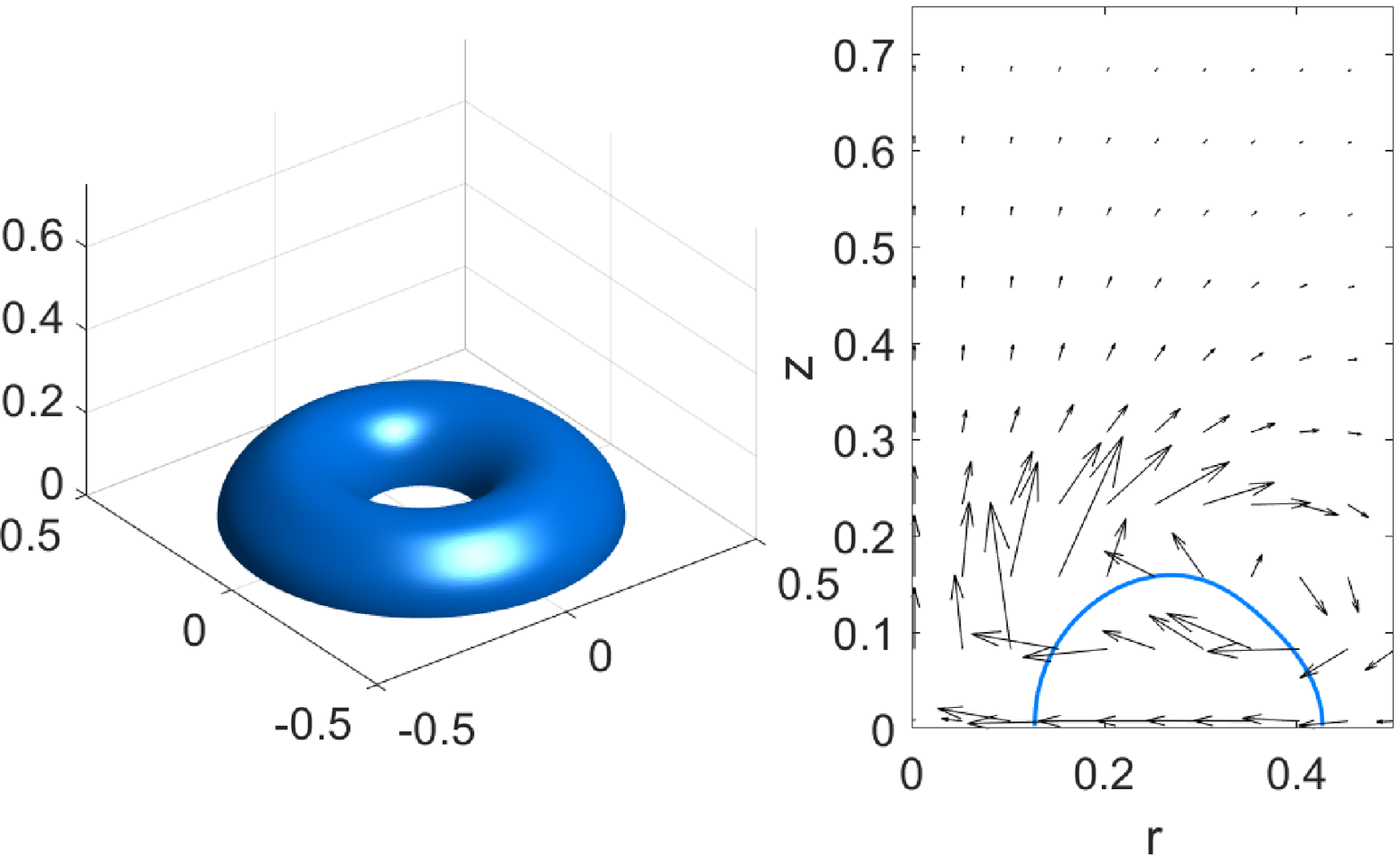}
\put(2,56){\scriptsize{(a)}}
\end{overpic}
\hspace{-0.4cm}
\begin{overpic}[trim=0cm 0cm 0cm 0cm, clip,scale=0.36]{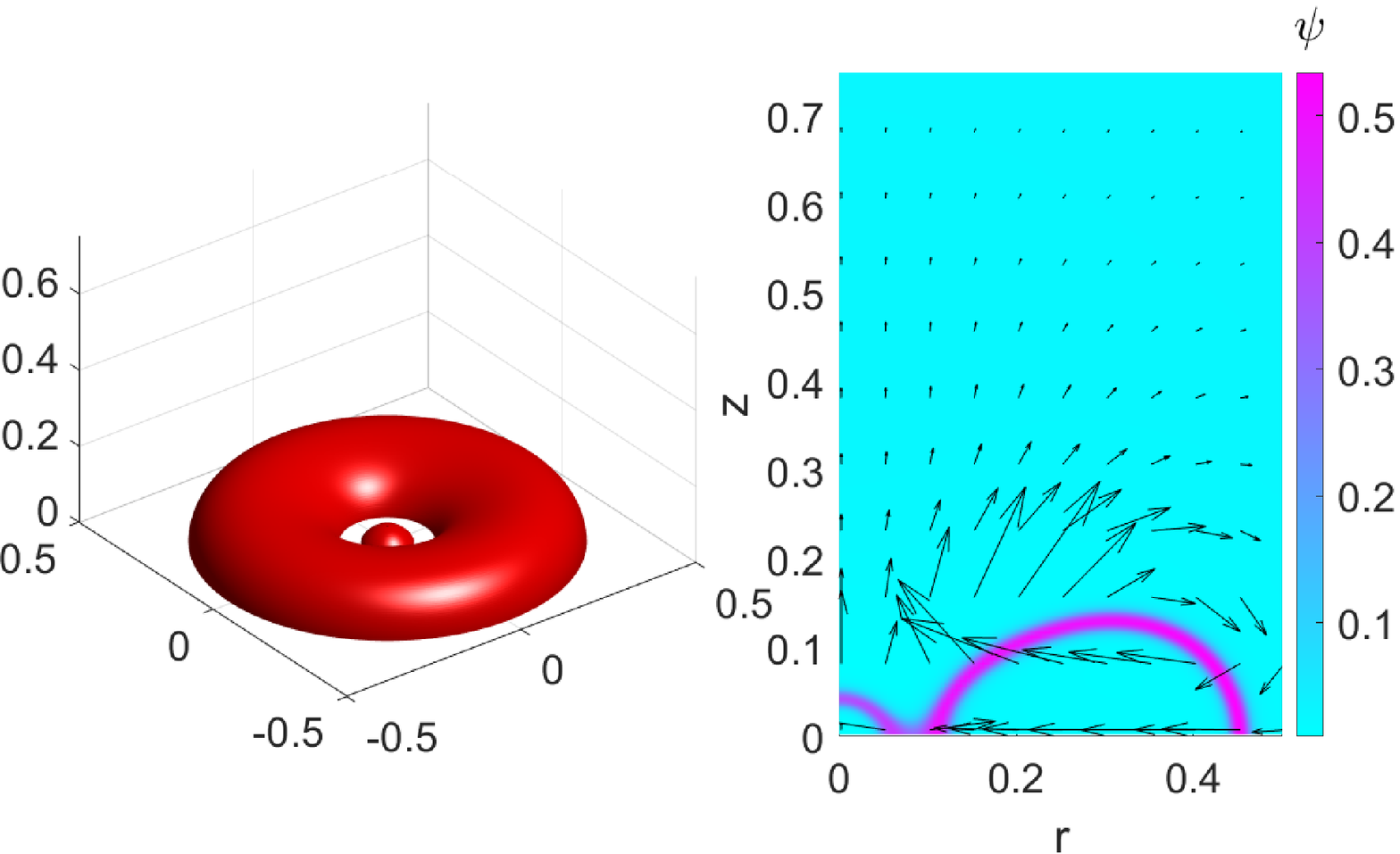}
\put(2,56){\scriptsize{(b)}}
\end{overpic}
\hspace{0.1cm}
\begin{overpic}[trim=0cm 0cm 0cm 0cm, clip,scale=0.36]{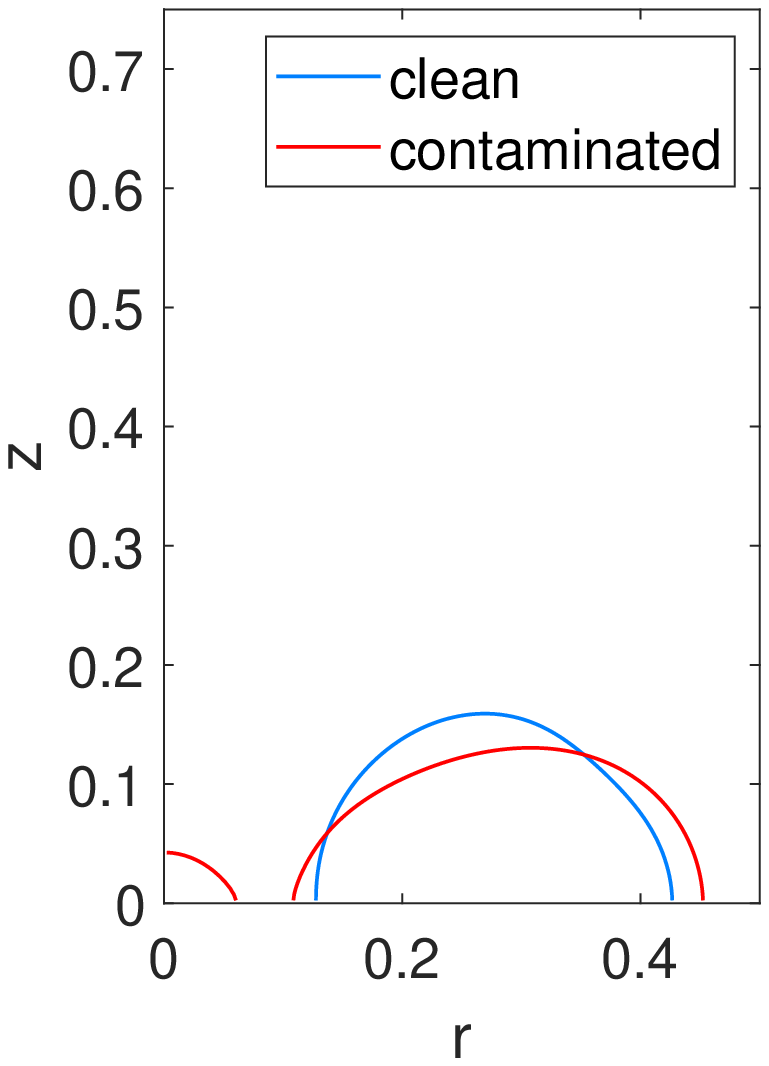}
\put(0,100){\scriptsize{(c)}}
\end{overpic}

\begin{overpic}[trim=0cm 0cm 0cm 0cm, clip,scale=0.36]{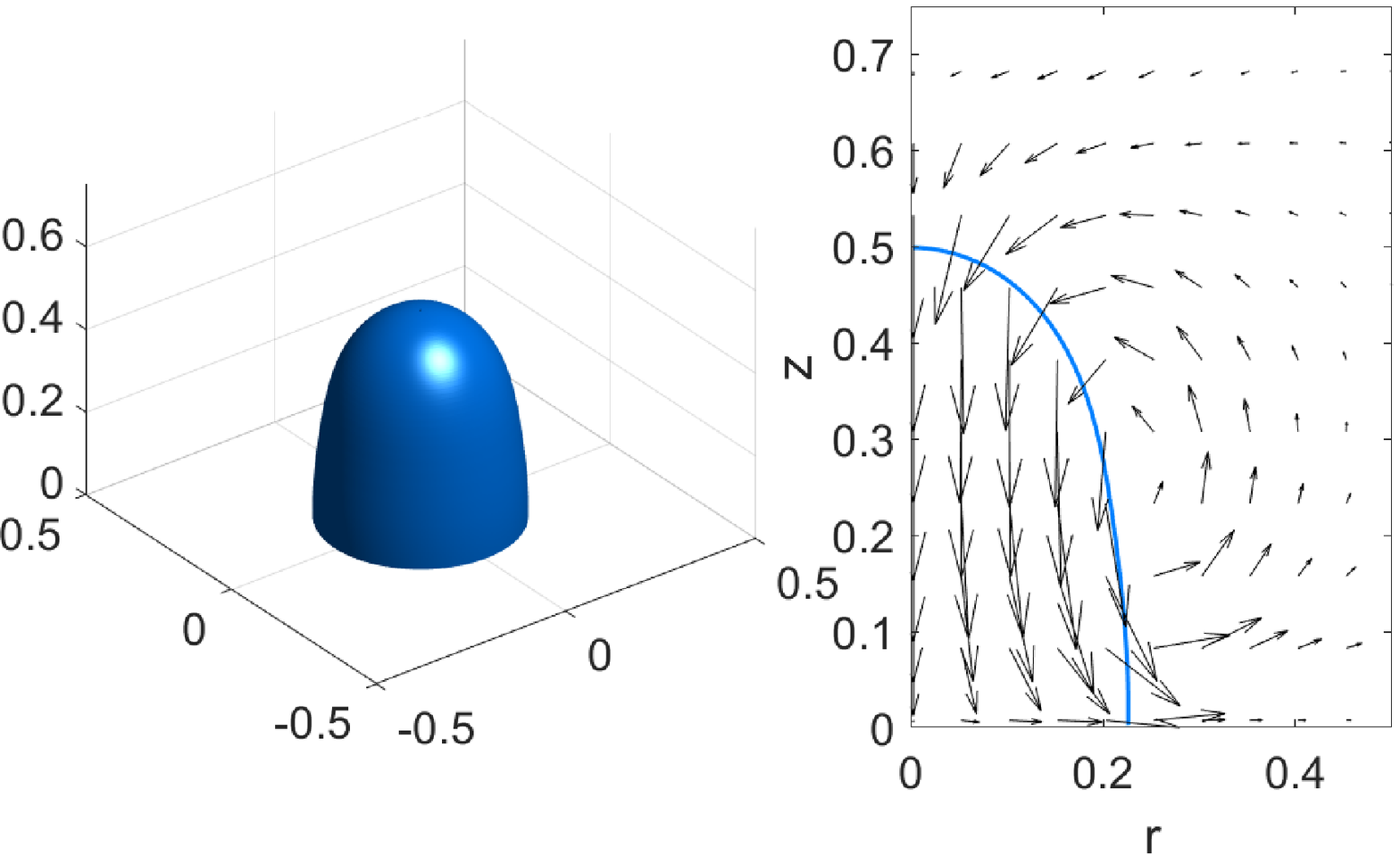}
\put(2,56){\scriptsize{(d)}}
\end{overpic}
\hspace{-0.4cm}
\begin{overpic}[trim=0cm 0cm 0cm 0cm, clip,scale=0.36]{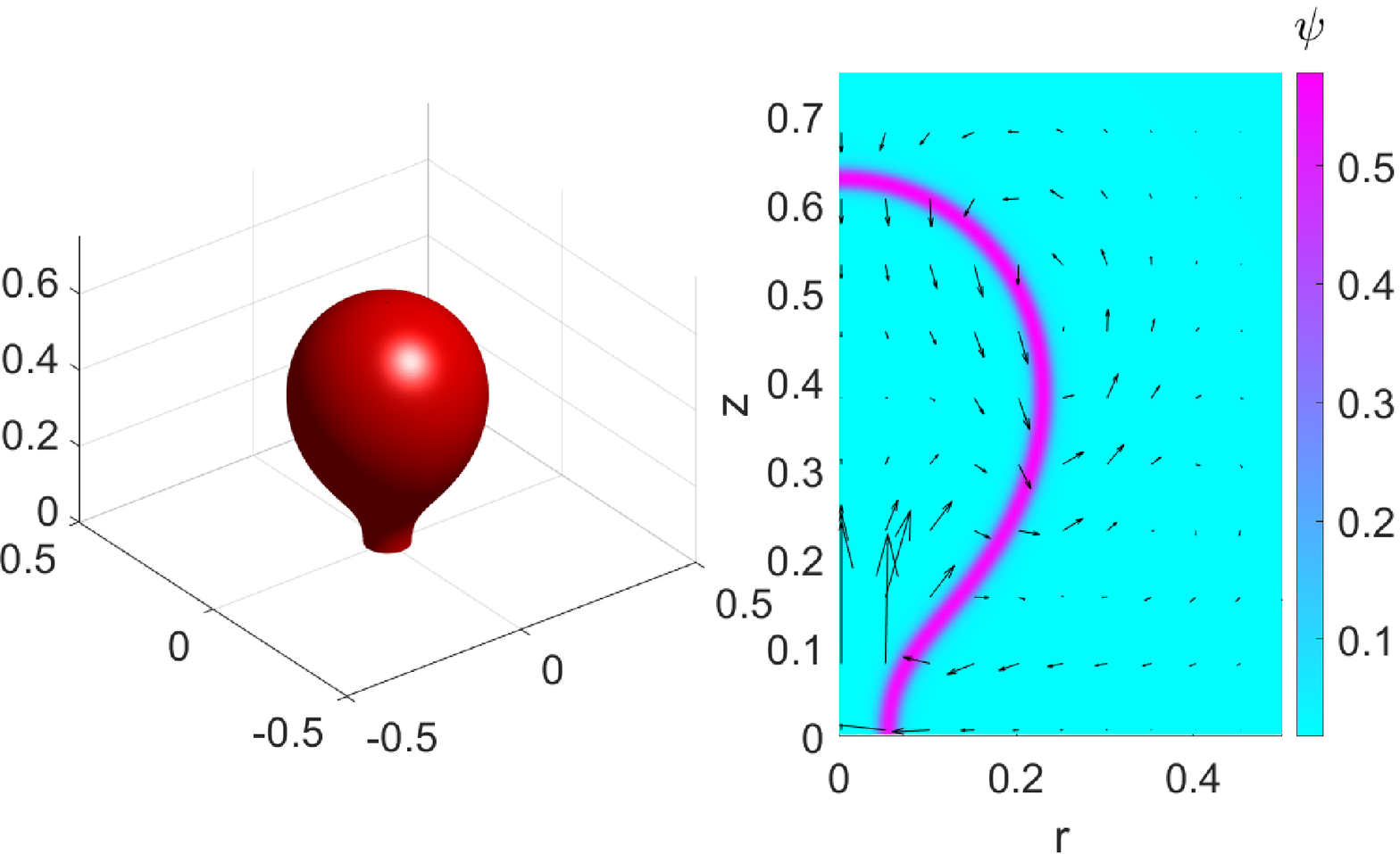}
\put(2,56){\scriptsize{(e)}}
\end{overpic}
\hspace{0.1cm}
\begin{overpic}[trim=0cm 0cm 0cm 0cm, clip,scale=0.36]{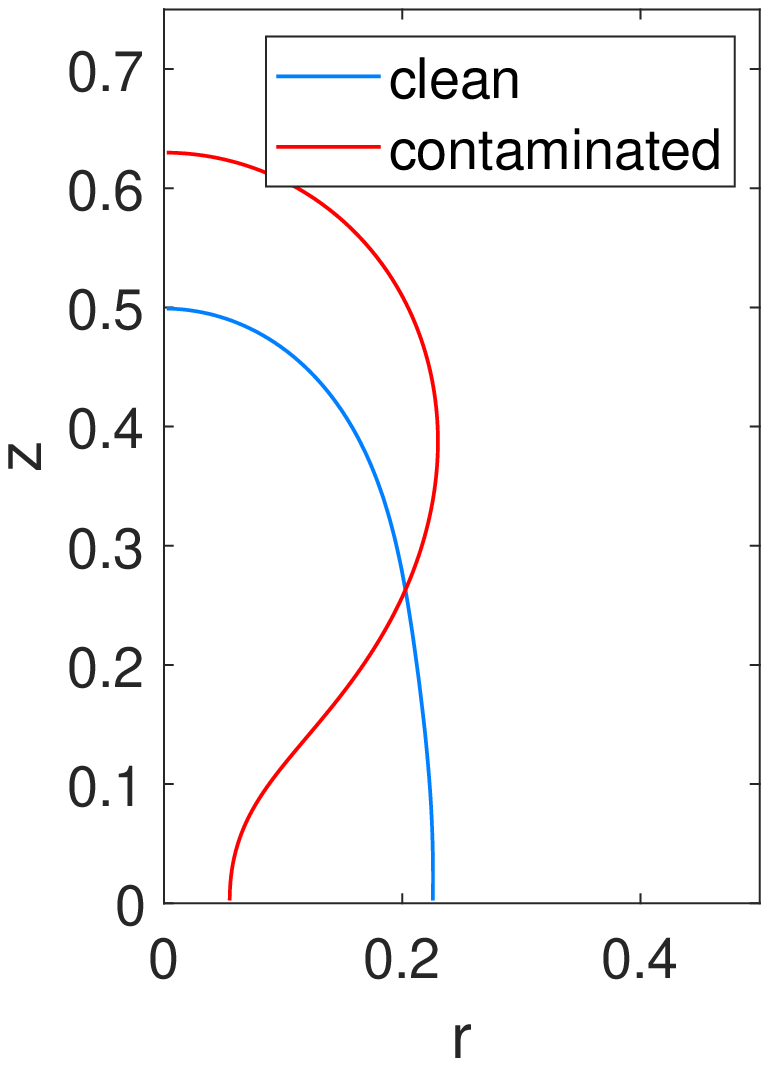}
\put(0,100){\scriptsize{(f)}}
\end{overpic}

\begin{overpic}[trim=0cm 0cm 0cm 0cm, clip,scale=0.36]{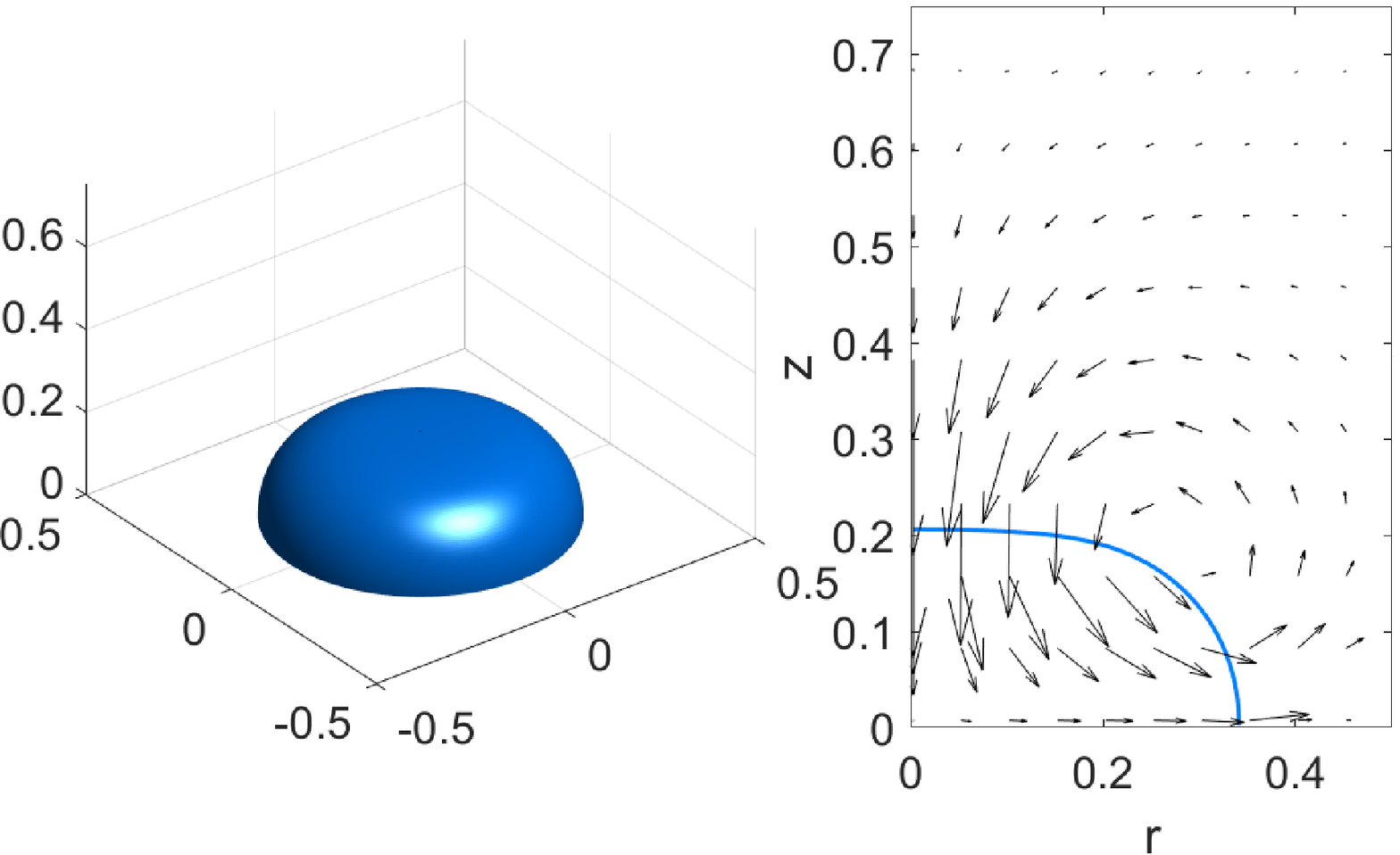}
\put(2,56){\scriptsize{(g)}}
\end{overpic}
\hspace{-0.4cm}
\begin{overpic}[trim=0cm 0cm 0cm 0cm, clip,scale=0.36]{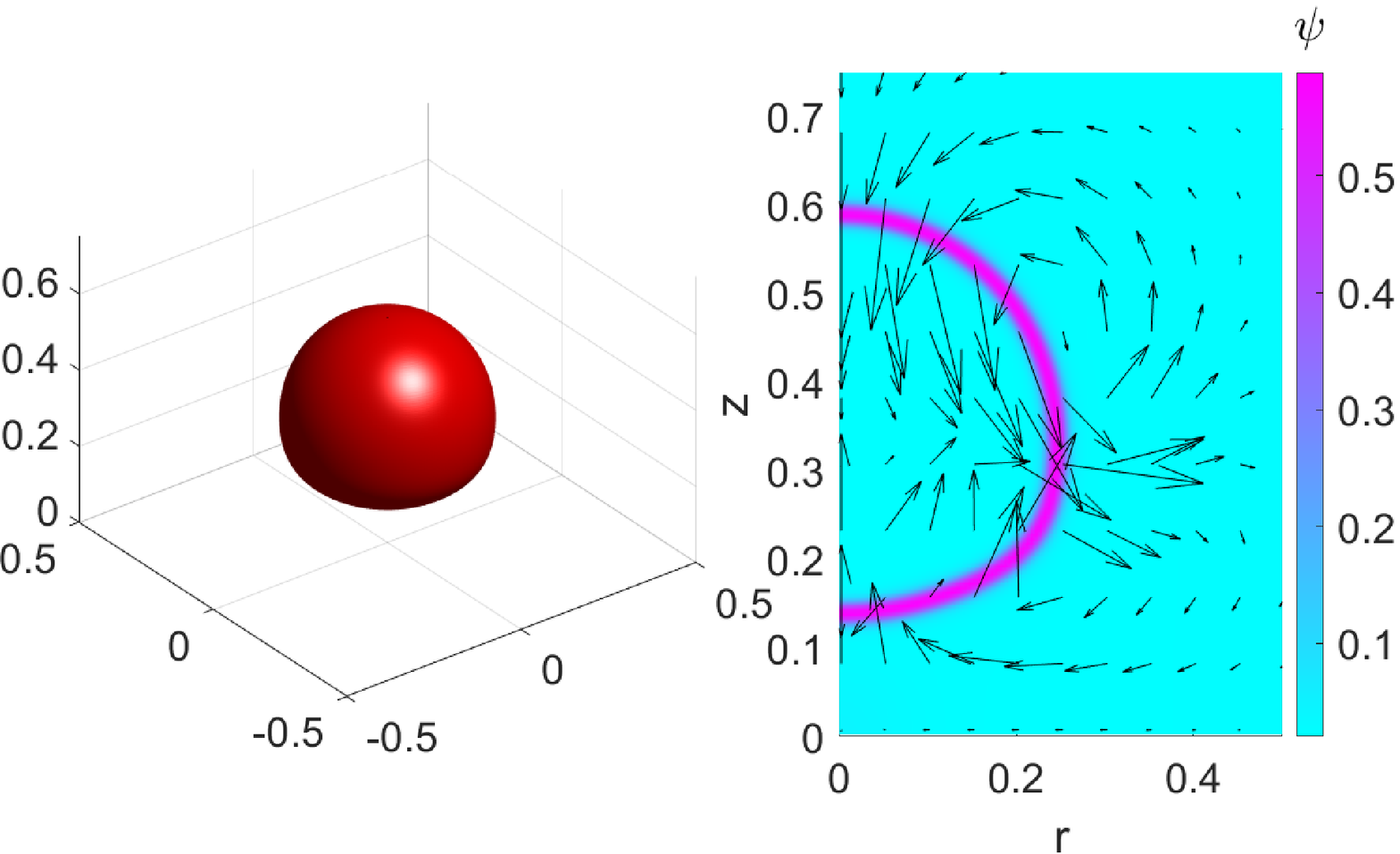}
\put(2,56){\scriptsize{(h)}}
\end{overpic}
\hspace{0.1cm}
\begin{overpic}[trim=0cm 0cm 0cm 0cm, clip,scale=0.36]{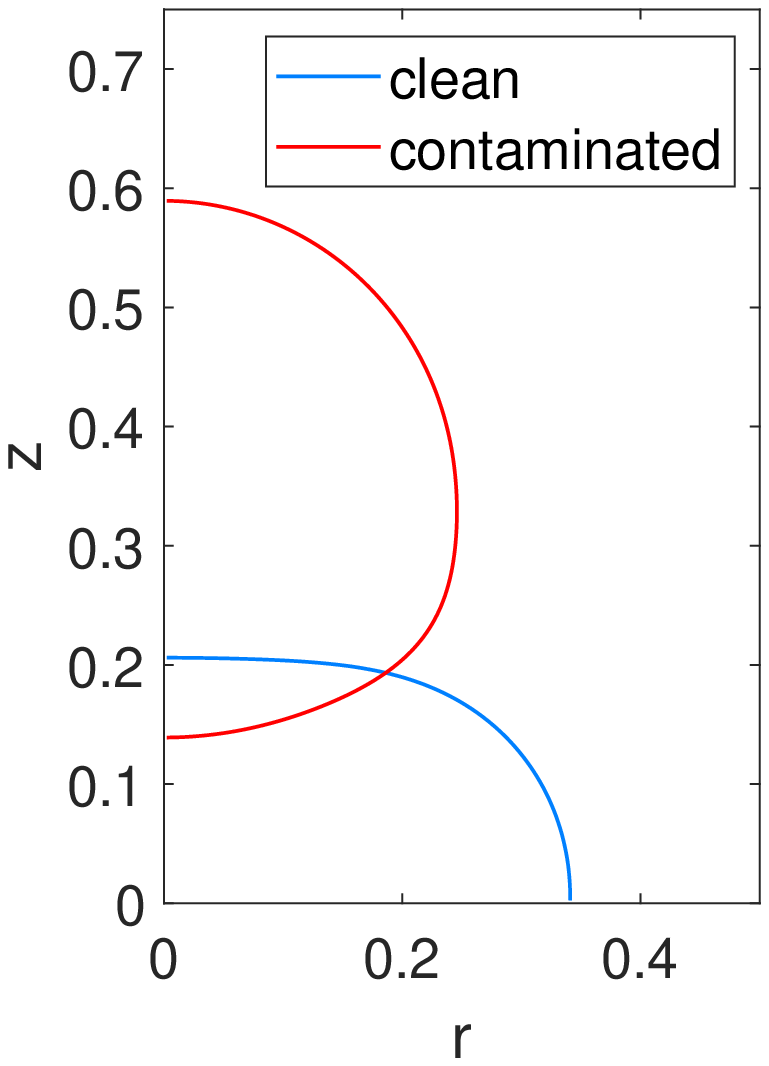}
\put(0,100){\scriptsize{(i)}}
\end{overpic}

\caption{Droplet adherence and bouncing in Example 2 ($\mathrm{Re}=1600$, $\mathrm{We}=18.9$, $\theta_s=90^\circ$, and $\mathrm{Pe}_\psi=10$). Profiles of clean and contaminated droplets are shown in both three-dimensional views (1st and 3rd columns) and two-dimensional radial plots (2nd and 4th columns). Velocity fields and surfactant concentrations are also shown in the radial plots by quivers and colormaps respectively. Comparisons between the interface shapes of clean and contaminated droplets are given in the 5th column. Snapshots are captured at time $t = 0.9$ (a-c), $t = 2.88$ (d-f), and $t = 3.3$ (g-i). }
\label{Adherence_3}
\end{figure}
\begin{figure}[ht!]
\centering
\begin{subfigure}{1\linewidth}
\centering
\includegraphics[scale=0.7]{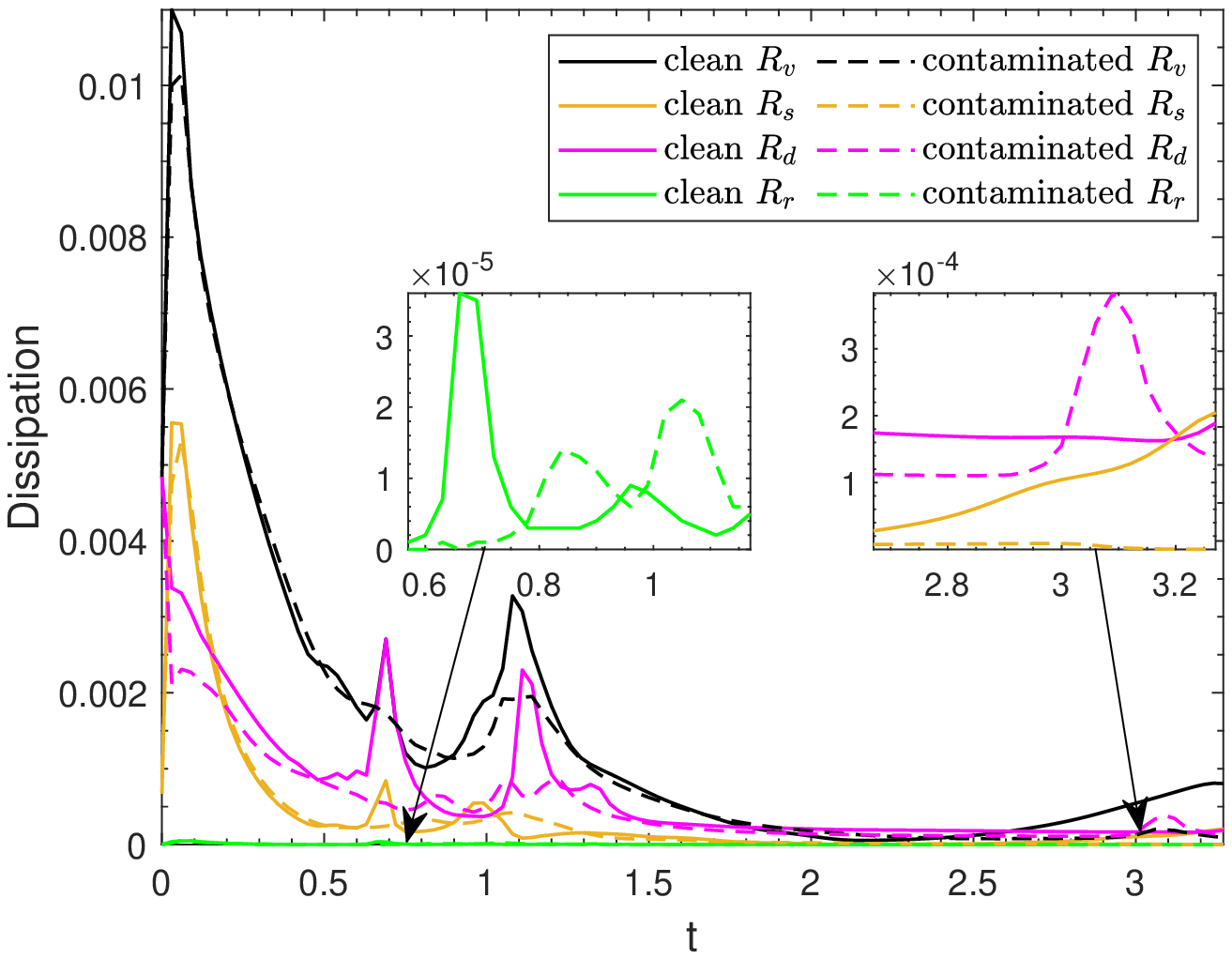}
\caption{}
\end{subfigure}

\begin{subfigure}{0.49\linewidth}
\centering
\includegraphics[scale=0.58]{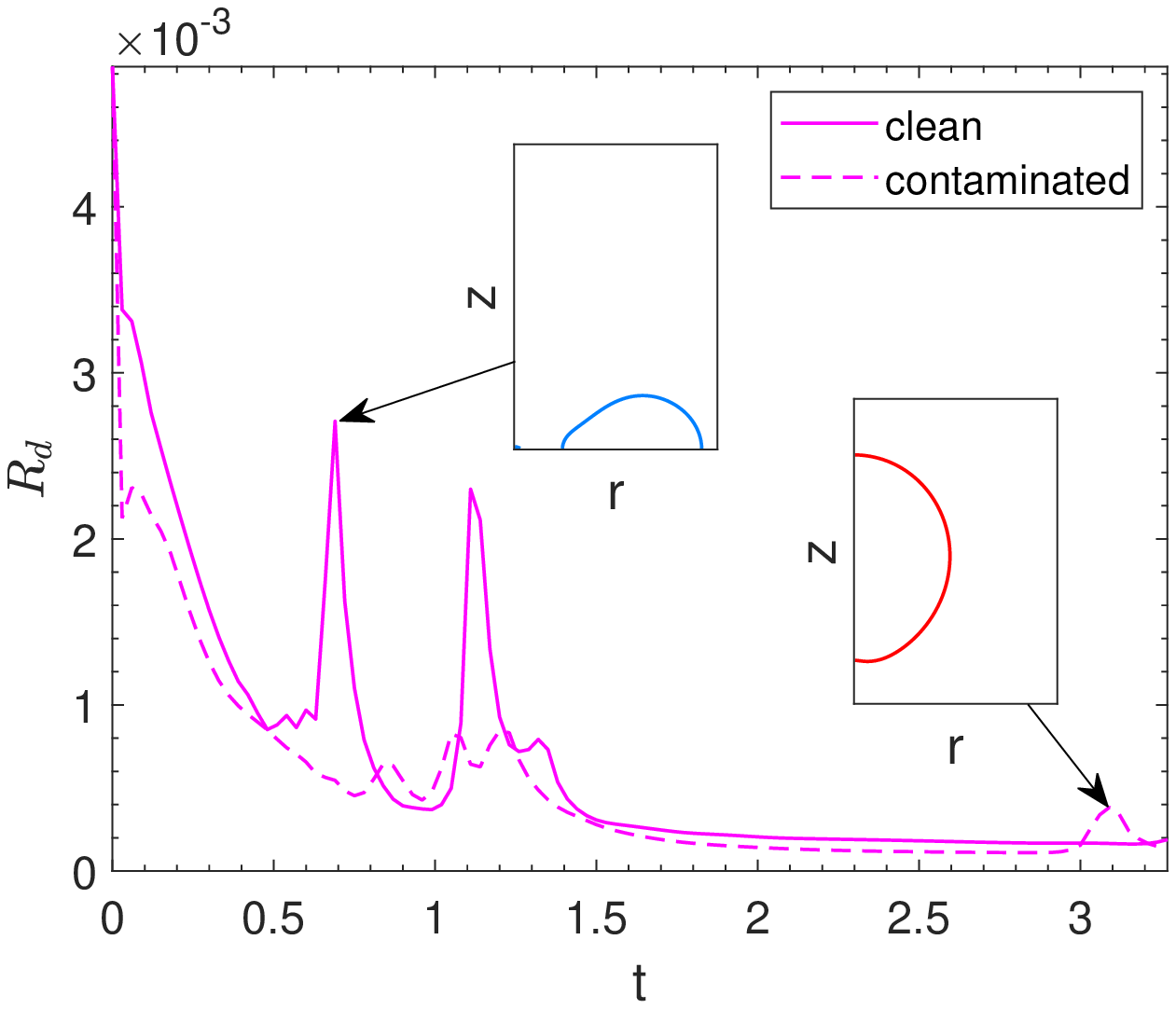}
\caption{}
\end{subfigure}
\begin{subfigure}{0.49\linewidth}
\centering
\includegraphics[scale=0.58]{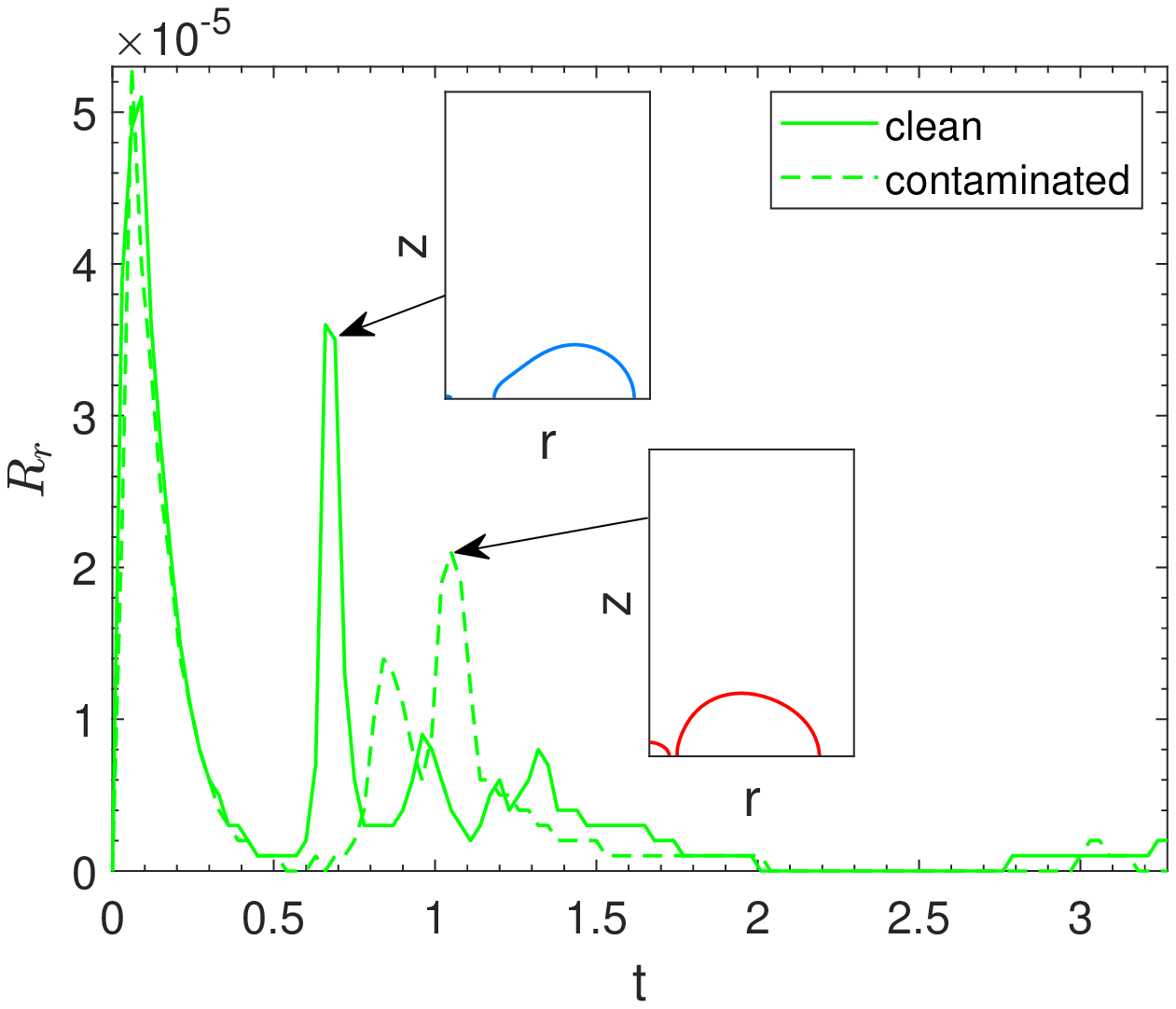}
\caption{}
    \end{subfigure}
    \caption{Dissipations in Example 2. In (a), four dissipations are shown for both clean (solid curves) and contaminated (dashed curves) cases. The inset plots illustrate differences in dissipations for the two cases nearby topological changes (formation of torus and island, and droplet bouncing). In particular, $R_d$ and $R_r$ are shown in (b) and (c), where their peaks reflect some particular interface geometries in the inset plots.}
\label{Adherence_3_dissi}
\end{figure}

\subsubsection{Bouncing}
In this section, we numerically investigate bouncing phenomena.

\paragraph{Example 3} In comparison to Example 2, we enlarge the Weber number so that the capillary effect is weakened. The following parameters are taken:
\begin{equation*}
\mathrm{Re}=1600, \quad\quad \mathrm{We}=37.7, \quad\quad \theta_s=90^\circ, \quad\quad \mathrm{Pe_\psi}=10.
\end{equation*}
As shown in Fig.~\ref{Bouncing_1}a-c, the spreading dynamics in both cases are similar as in Example 2. However, due to the weakened capillary effect, both the clean and contaminated droplets are able to bounce (Fig.~\ref{Bouncing_1}d-i). The contaminated droplet is `softer' so that in bouncing process it breaks up into two parts: one continues to rise, while the other stays on the substrate. This leads to partial bouncing (Fig.~\ref{Bouncing_1}h). In contrast, the clean droplet experiences complete bouncing (Fig.~\ref{Bouncing_1}g).

In Fig.~\ref{Bouncing_1_dissi}, more peaks in dissipations are observed, indicating complex topological changes happen. In particular, breakup in the bouncing contaminated droplet occurs nearby the peak in $R_d$ around $t=4.71$, meanwhile, the clean droplet has completely left the substrate since the peak in $R_d$ emerges around $t=3.87$.
\begin{figure}[t!]
\center
\begin{overpic}[trim=0cm 0cm 0cm 0cm, clip,scale=0.36]{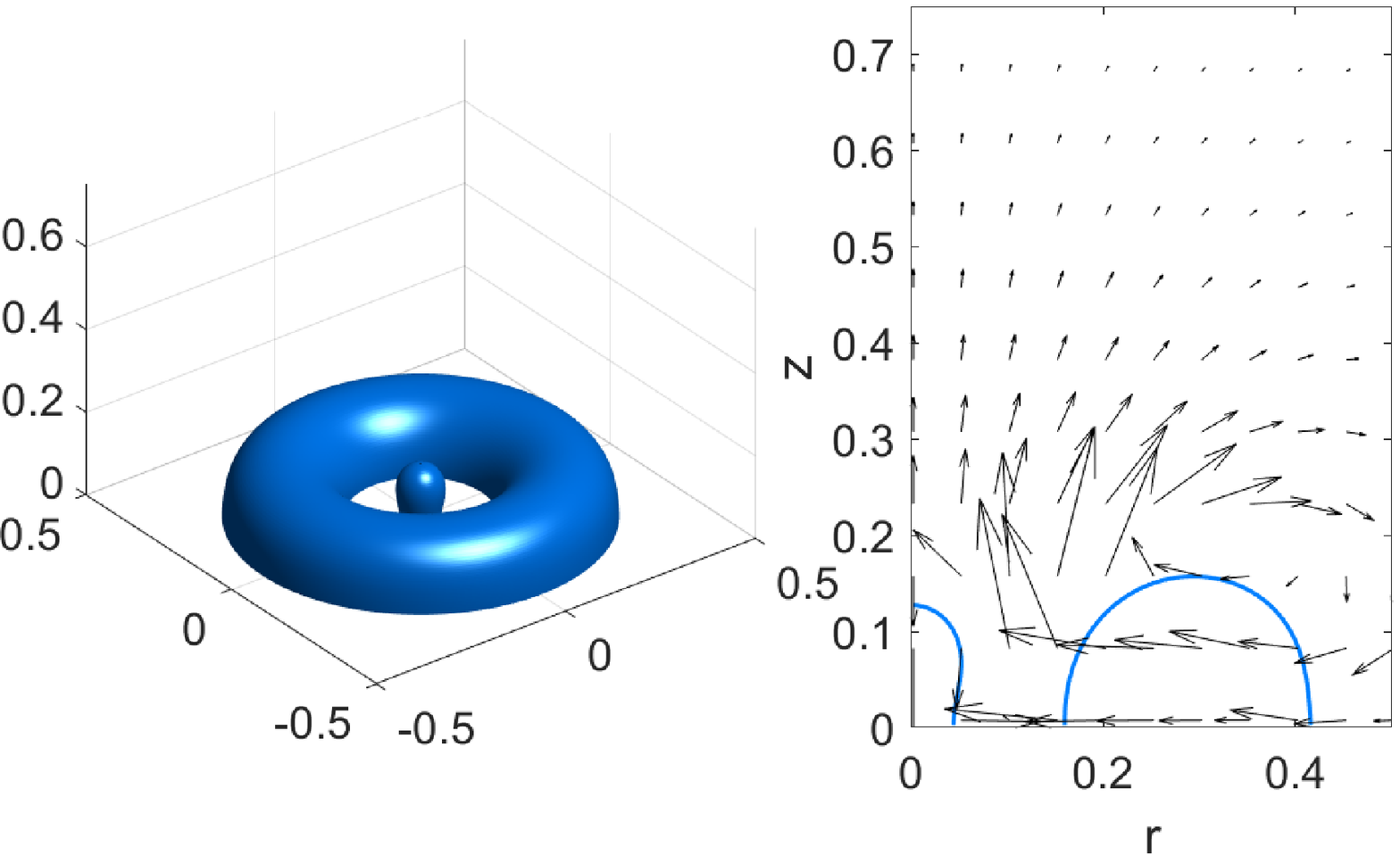}
\put(2,56){\scriptsize{(a)}}
\end{overpic}
\hspace{-0.4cm}
\begin{overpic}[trim=0cm 0cm 0cm 0cm, clip,scale=0.36]{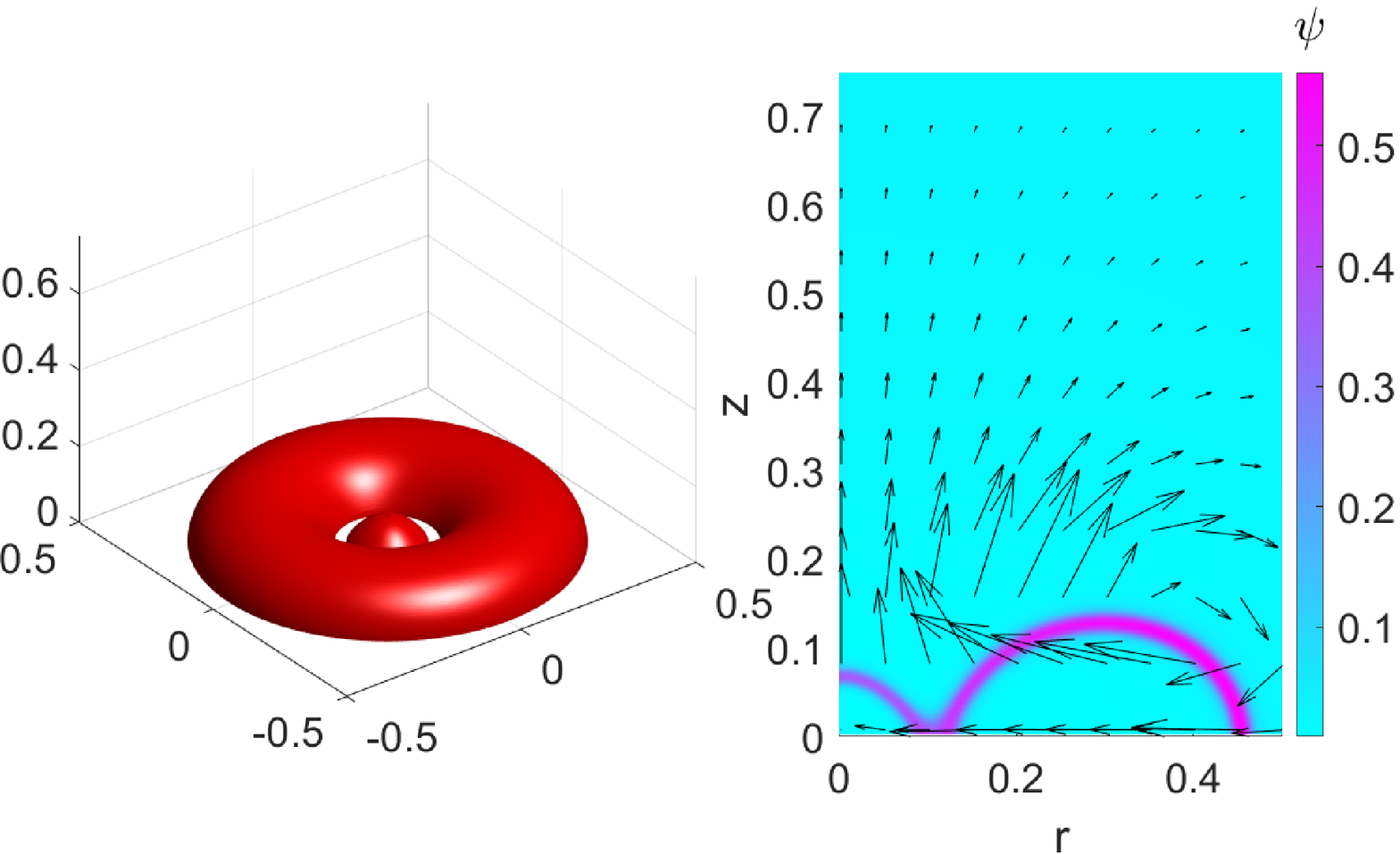}
\put(2,56){\scriptsize{(b)}}
\end{overpic}
\hspace{0.1cm}
\begin{overpic}[trim=0cm 0cm 0cm 0cm, clip,scale=0.36]{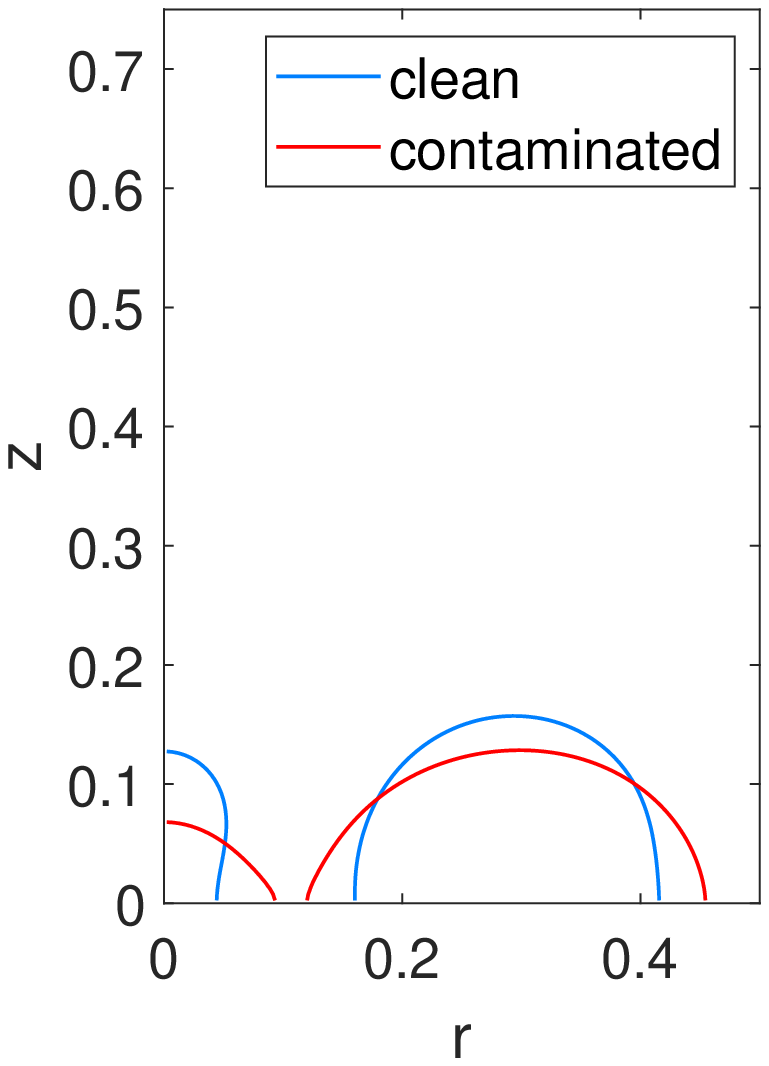}
\put(0,100){\scriptsize{(c)}}
\end{overpic}

\begin{overpic}[trim=0cm 0cm 0cm 0cm, clip,scale=0.36]{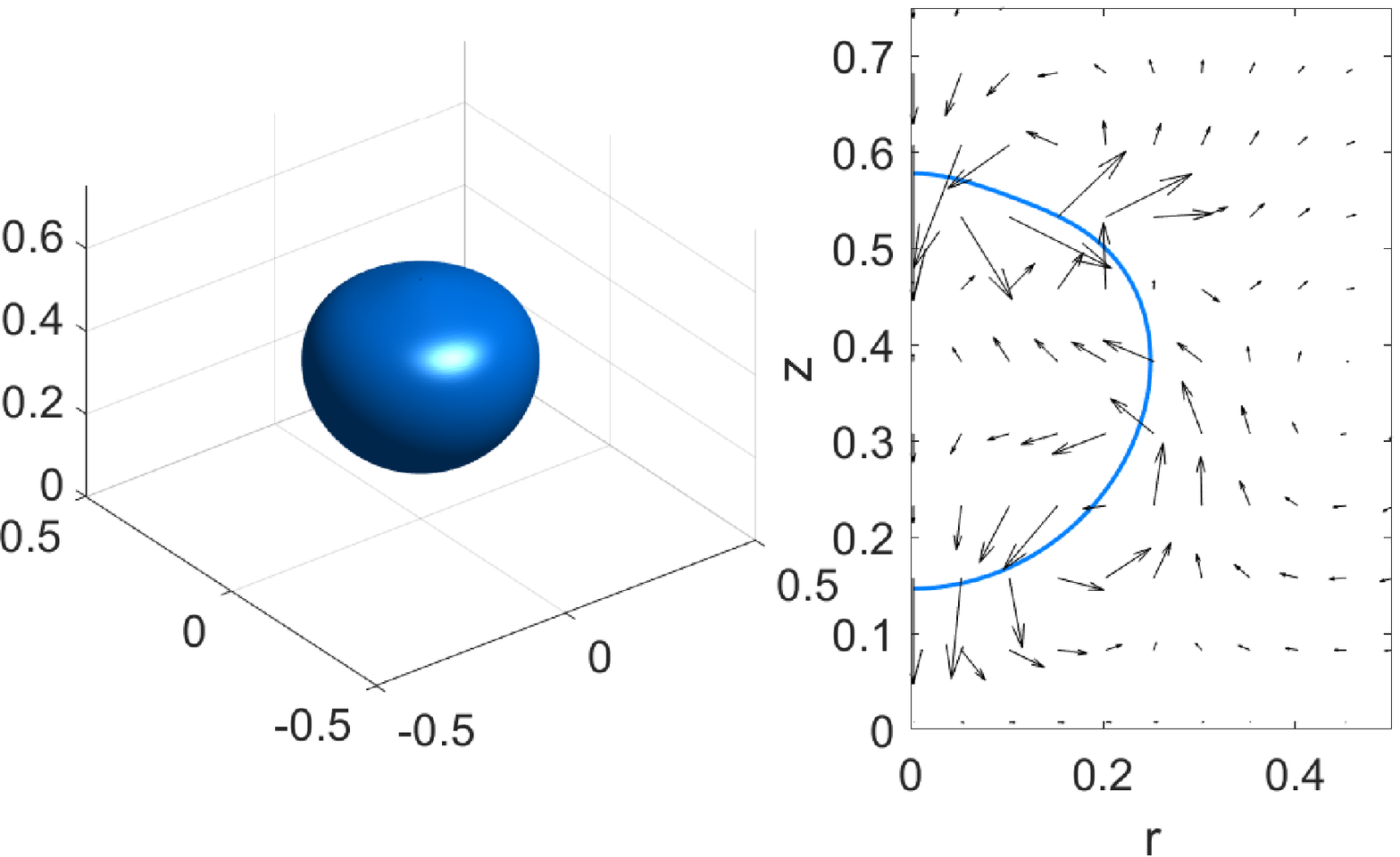}
\put(2,56){\scriptsize{(d)}}
\end{overpic}
\hspace{-0.4cm}
\begin{overpic}[trim=0cm 0cm 0cm 0cm, clip,scale=0.36]{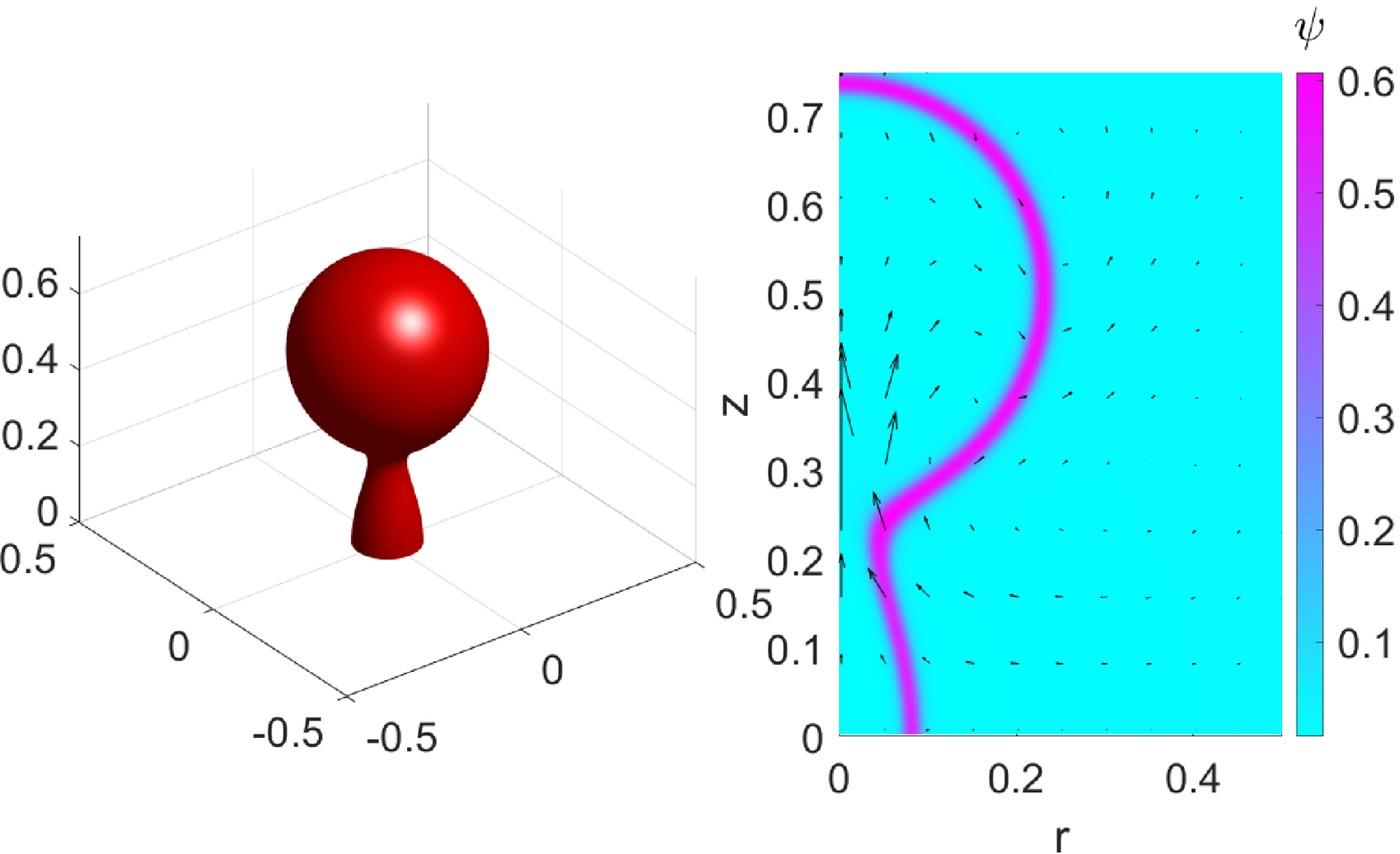}
\put(2,56){\scriptsize{(e)}}
\end{overpic}
\hspace{0.1cm}
\begin{overpic}[trim=0cm 0cm 0cm 0cm, clip,scale=0.36]{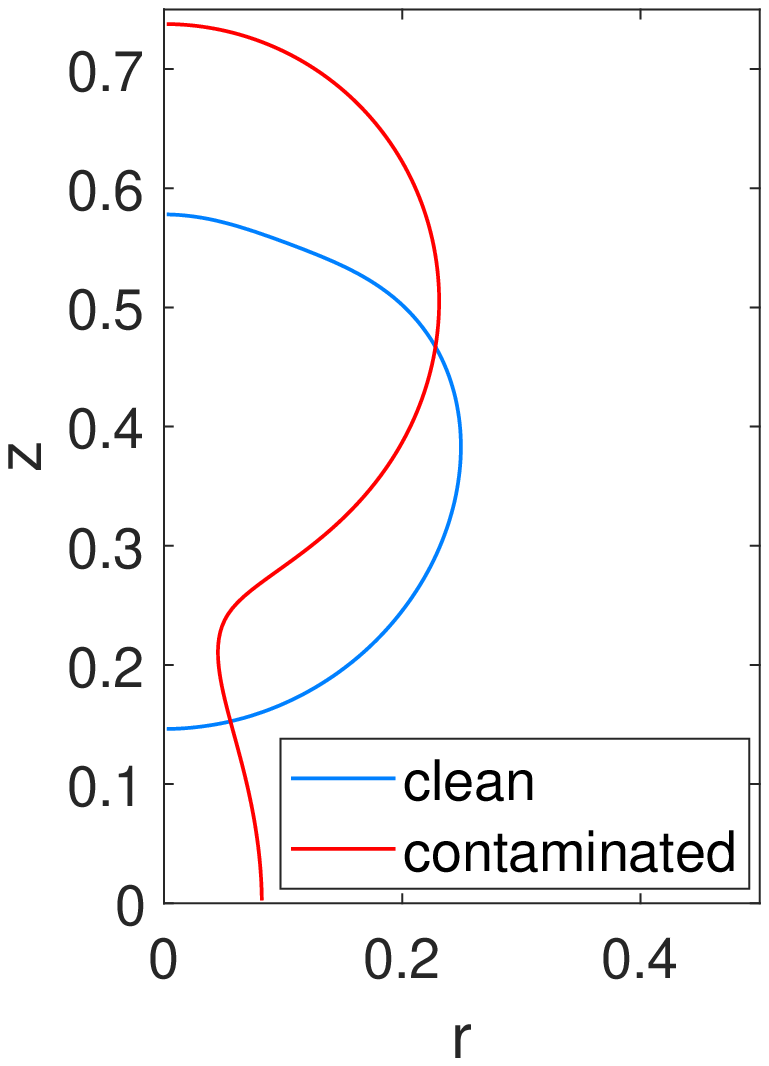}
\put(0,100){\scriptsize{(f)}}
\end{overpic}

\begin{overpic}[trim=0cm 0cm 0cm 0cm, clip,scale=0.36]{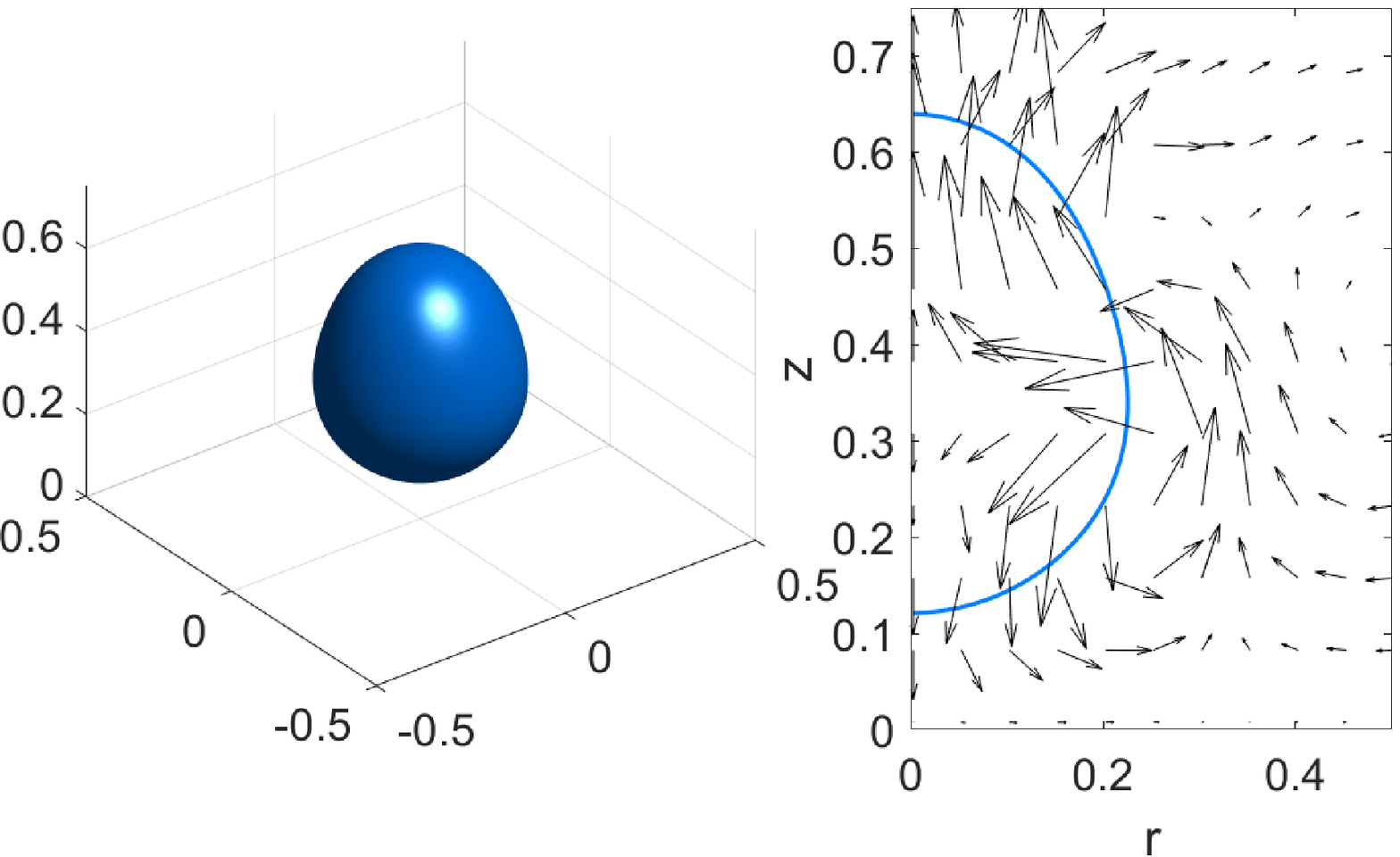}
\put(2,56){\scriptsize{(g)}}
\end{overpic}
\hspace{-0.4cm}
\begin{overpic}[trim=0cm 0cm 0cm 0cm, clip,scale=0.36]{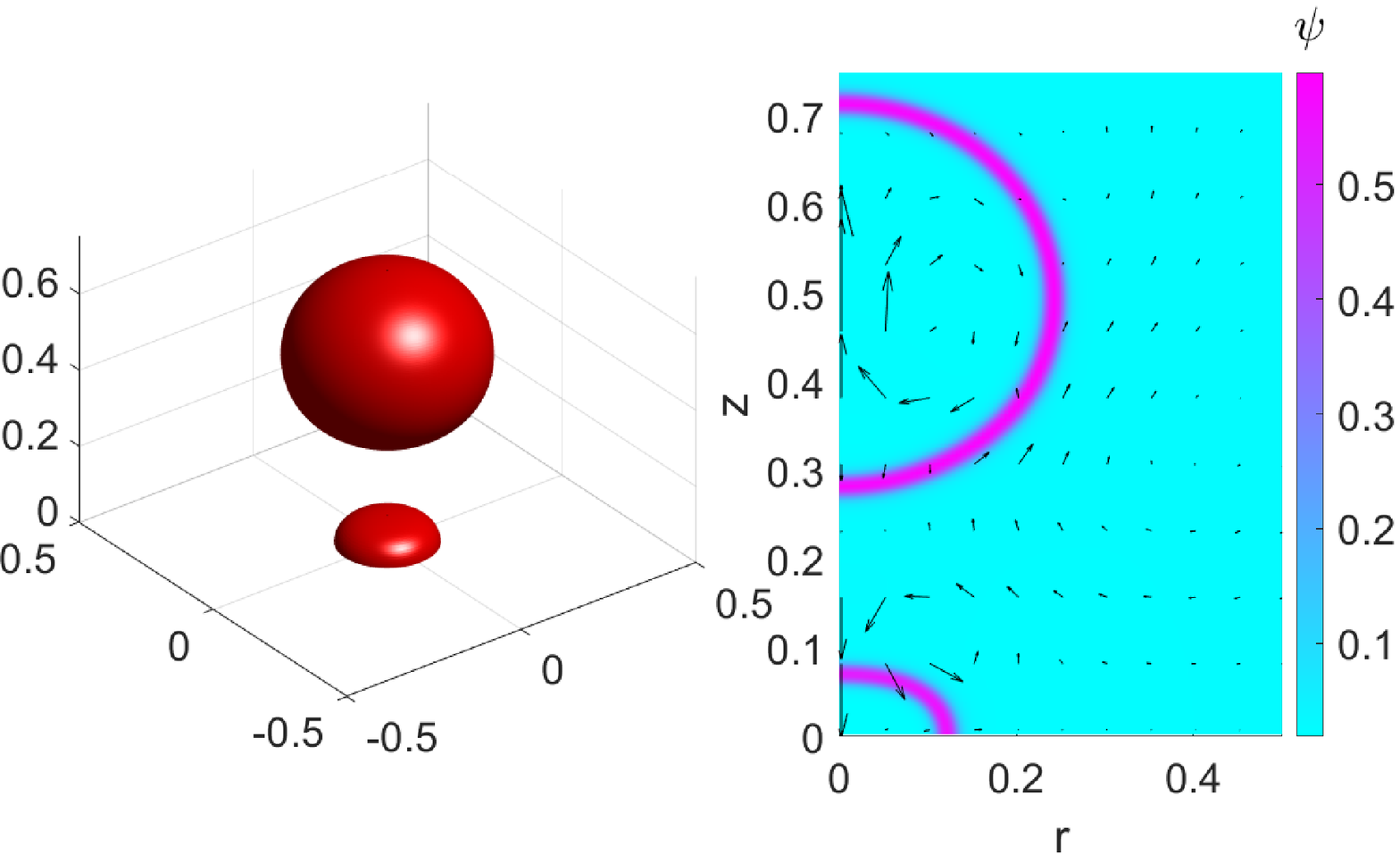}
\put(2,56){\scriptsize{(h)}}
\end{overpic}
\hspace{0.1cm}
\begin{overpic}[trim=0cm 0cm 0cm 0cm, clip,scale=0.36]{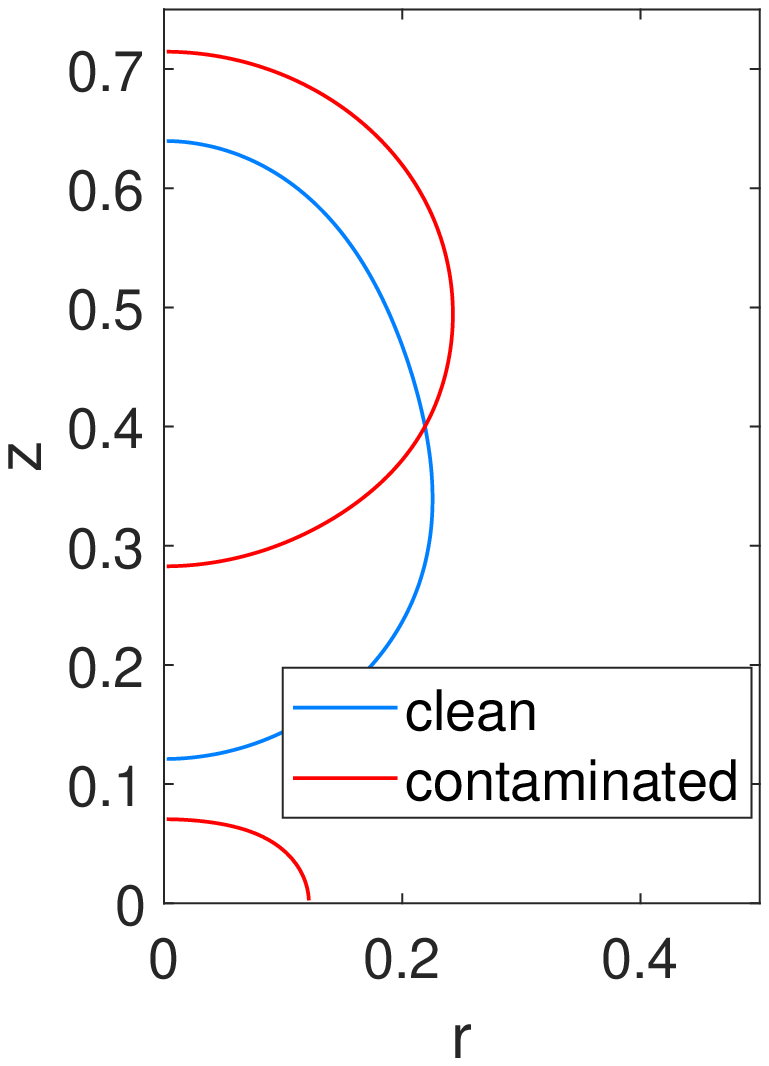}
\put(0,100){\scriptsize{(i)}}
\end{overpic}

\caption{Droplet bouncing and partial bouncing in Example 3 ($\mathrm{Re}=1600$, $\mathrm{We}=37.7$, $\theta_s=90^\circ$, and $\mathrm{Pe}_\psi=10$). Profiles of clean and contaminated droplets are shown in both three-dimensional views ((1st and 3rd columns) and two-dimensional radial plots ((2nd and 4th columns). Velocity fields and surfactant concentrations are also shown in the radial plots by quivers and colormaps respectively. Comparisons between the interface shapes of clean and contaminated droplets are given in the 5th column. Snapshots are captured at time $t = 1.38$ (a-c), $t = 4.5$ (d-f), and $t = 5.1$ (g-i).}
\label{Bouncing_1}
\end{figure}
\begin{figure}[t!]
\centering
\includegraphics[scale=0.7]{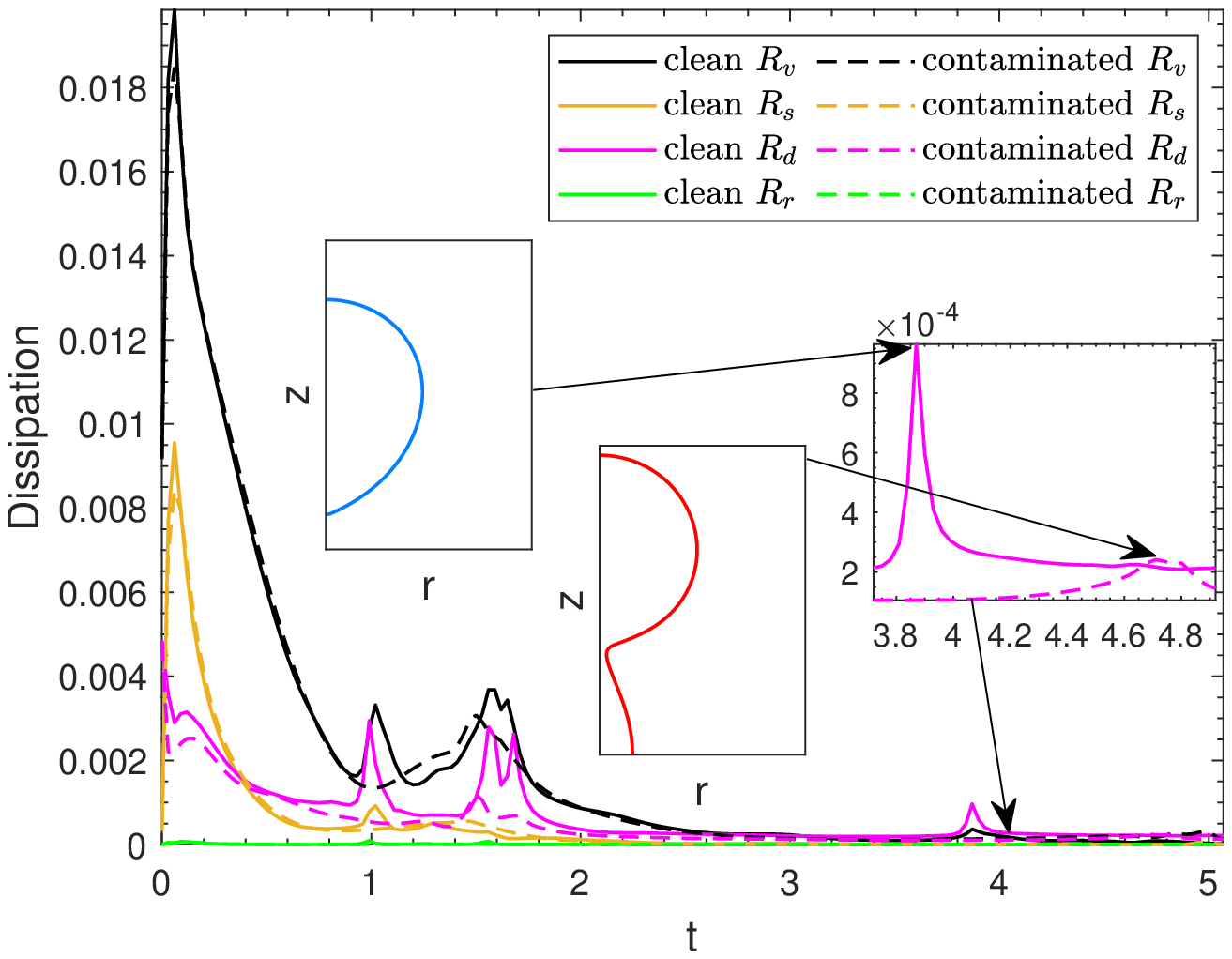}
    \caption{Dissipations in Example 3. In (a), four dissipations are shown for both clean (solid curves) and contaminated (dashed curves) cases. The right inset plot illustrates differences in $R_d$ for the two cases when topological change occurs, and some particular interface geometries corresponding to peaks of $R_d$ are shown in other inset plots.}
\label{Bouncing_1_dissi}
\end{figure}

\paragraph{Example 4} Droplet bouncing is more easily to happen if we further enlarge the Weber number while making the substrate hydrophobic. In particular, we choose the following parameters:
\begin{equation*}
\mathrm{Re}=2000, \quad\quad  \mathrm{We}=282.8, \quad\quad \theta_s=100^\circ,\quad\quad \mathrm{Pe}_\psi=100.
\end{equation*}
In this case, both clean and contaminated droplets undergo complete bouncing. As shown in Fig.~\ref{Bouncing_2}a-c, breakups are more easily to happen in the spreading dynamics of the contaminated droplet. In recoiling process, the central component of droplet first bounces, while it takes some time for the peripheric component (torus part) to agglomerate and then bounce. It can be observed in Fig.~\ref{Bouncing_2}d-i that in both clean and contaminated cases, two small drops rise with the lower one chasing the upper one until they merge. Due the presence of surfactant, the two contaminated drops are more likely to be deformed and keep in contact without merging. It takes longer time for them to merge compared to the clean case.

In the plots of dissipations (Fig.~\ref{Bouncing_2_dissi}), one can infer more topological changes by examining peaks in $R_d$ and $R_r$. In addition, merging of two drops can also be inferred from the rapid change of $R_d$ around $t=8$ and $t=10.8$.
\begin{figure}[t]
\center
\begin{overpic}[trim=0cm 0cm 0cm 0cm, clip,scale=0.33]{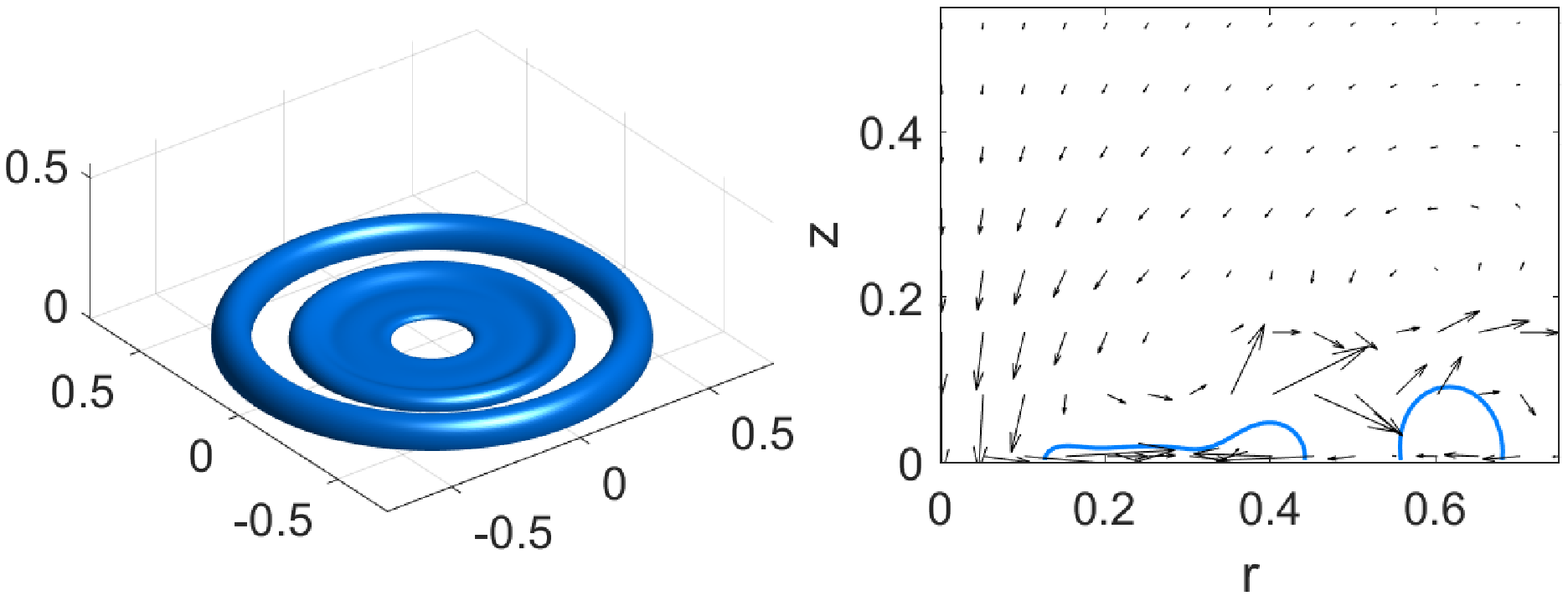}
\put(2,35){\scriptsize{(a)}}
\end{overpic}
\hspace{-0.5cm}
\begin{overpic}[trim=0cm 0cm 0cm 0cm, clip,scale=0.33]{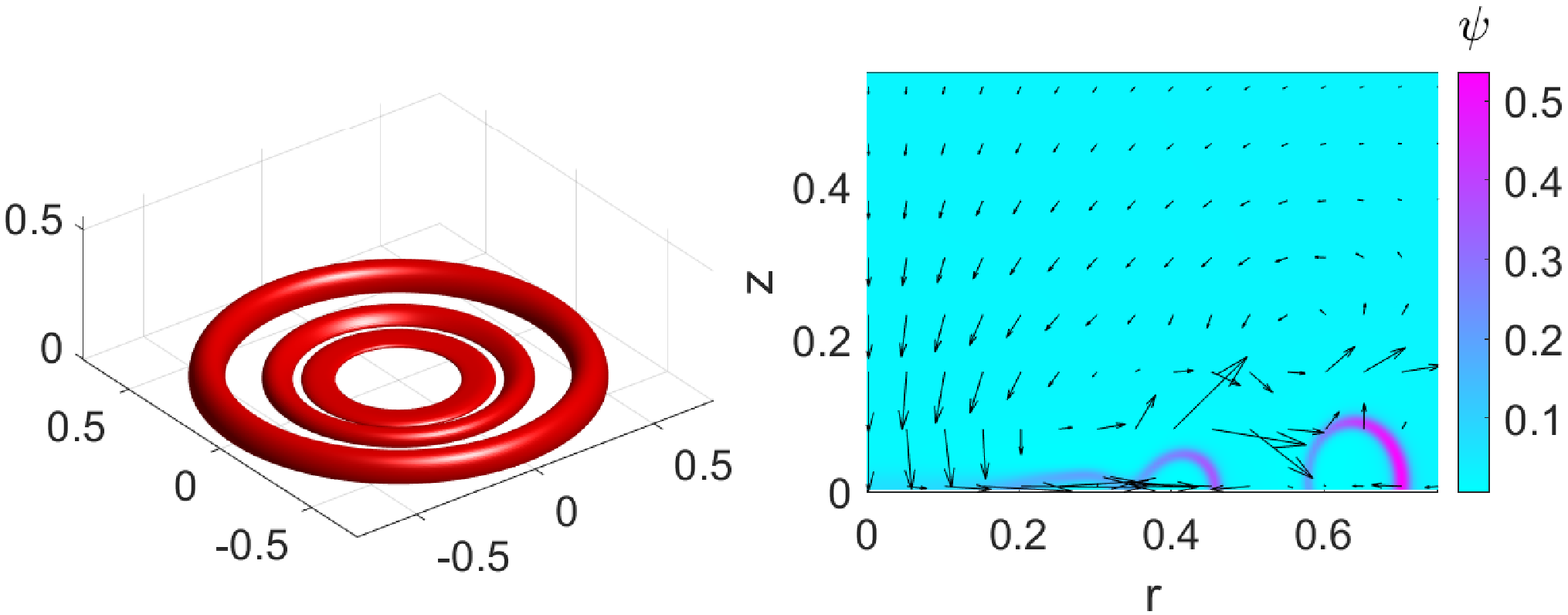}
\put(2,35){\scriptsize{(b)}}
\end{overpic}
\begin{overpic}[trim=0cm 0cm 0cm 0cm, clip,scale=0.33]{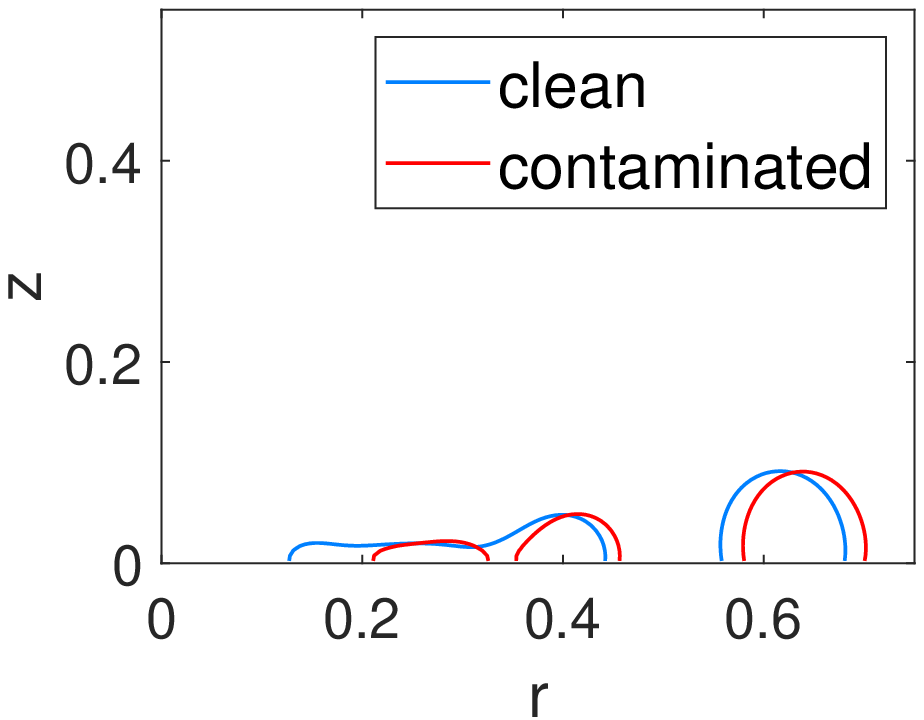}
\put(3,78){\scriptsize{(c)}}
\end{overpic}

\begin{overpic}[trim=0cm 0cm 0cm 0cm, clip,scale=0.33]{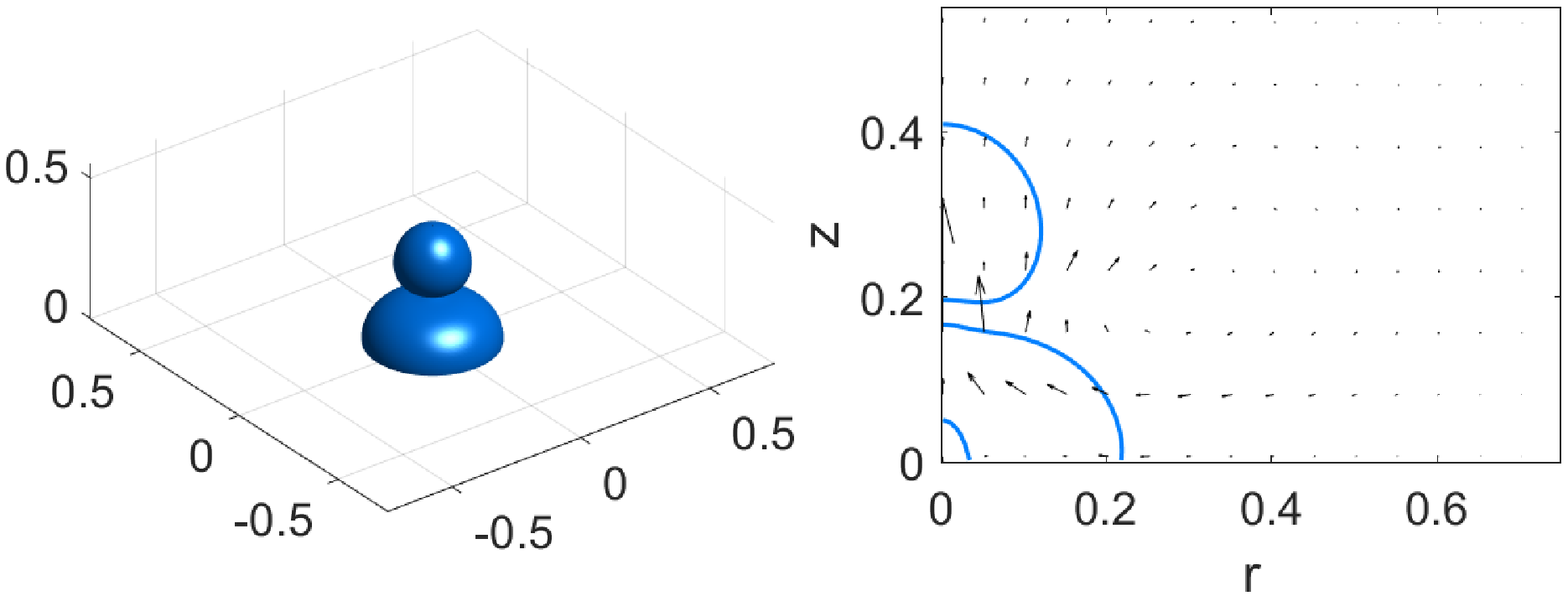}
\put(2,35){\scriptsize{(d)}}
\end{overpic}
\hspace{-0.5cm}
\begin{overpic}[trim=0cm 0cm 0cm 0cm, clip,scale=0.33]{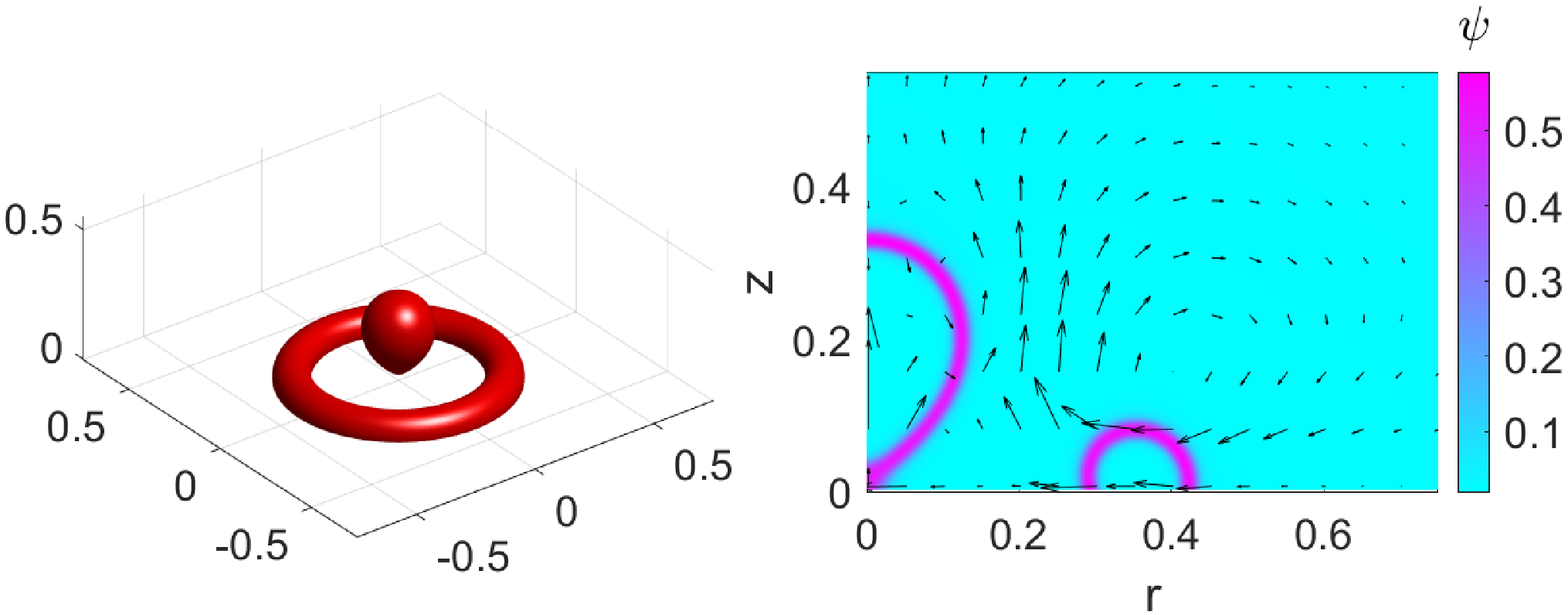}
\put(2,35){\scriptsize{(e)}}
\end{overpic}
\begin{overpic}[trim=0cm 0cm 0cm 0cm, clip,scale=0.33]{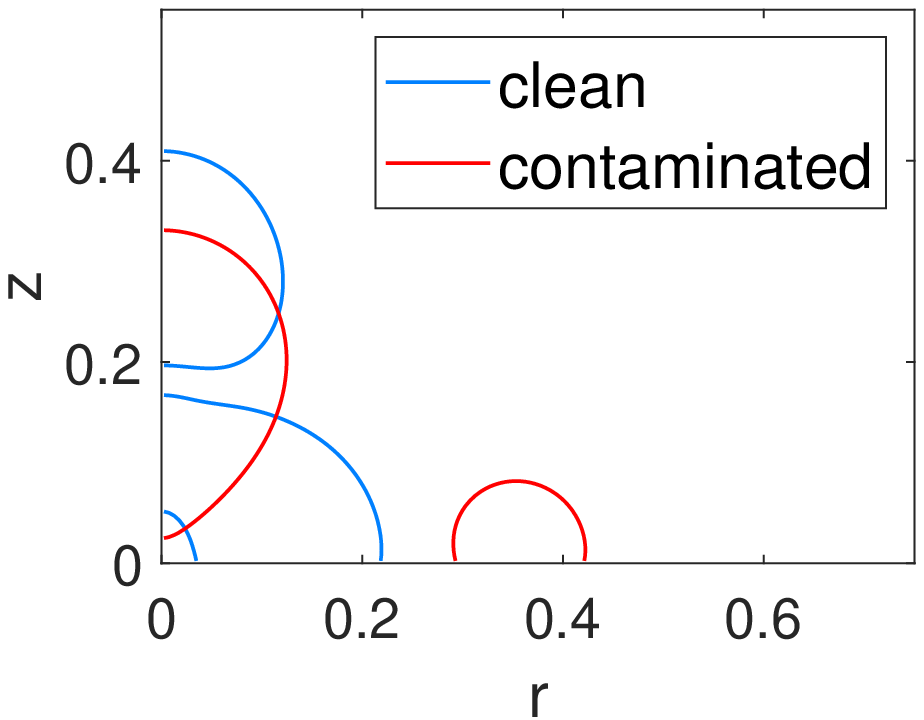}
\put(3,78){\scriptsize{(f)}}
\end{overpic}

\begin{overpic}[trim=0cm 0cm 0cm 0cm, clip,scale=0.33]{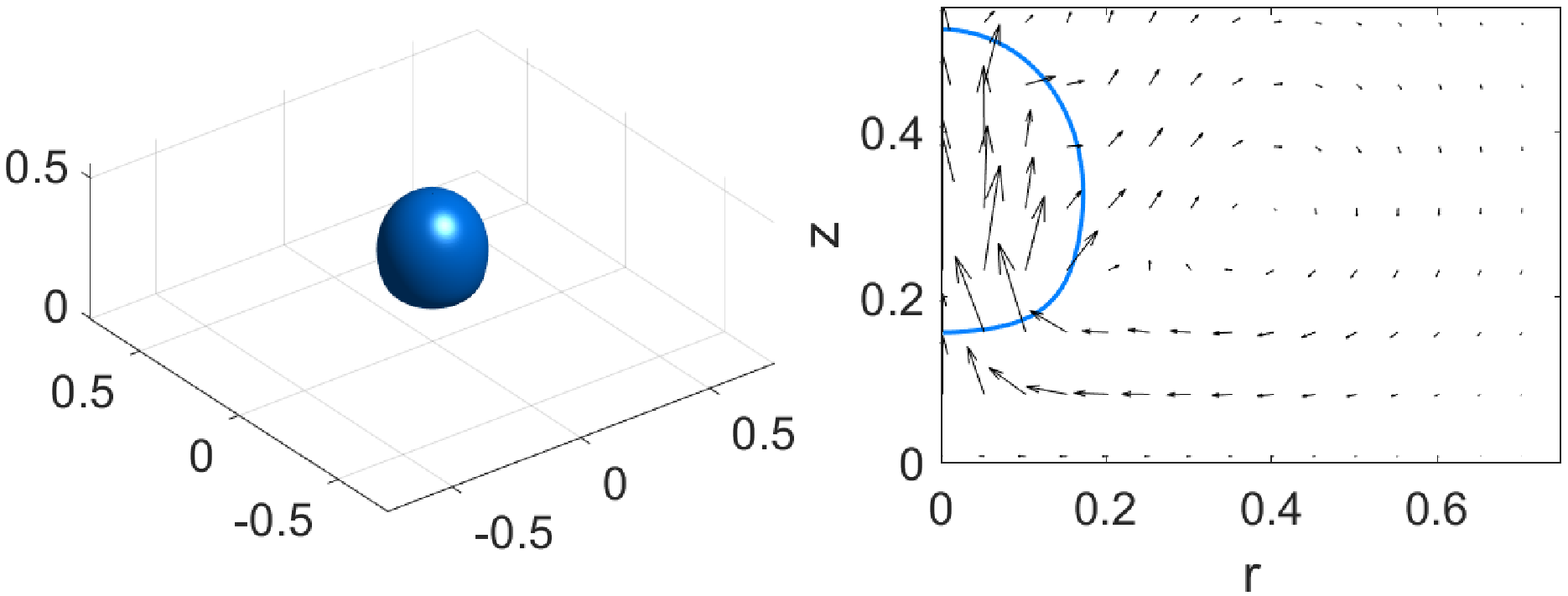}
\put(2,35){\scriptsize{(g)}}
\end{overpic}
\hspace{-0.5cm}
\begin{overpic}[trim=0cm 0cm 0cm 0cm, clip,scale=0.33]{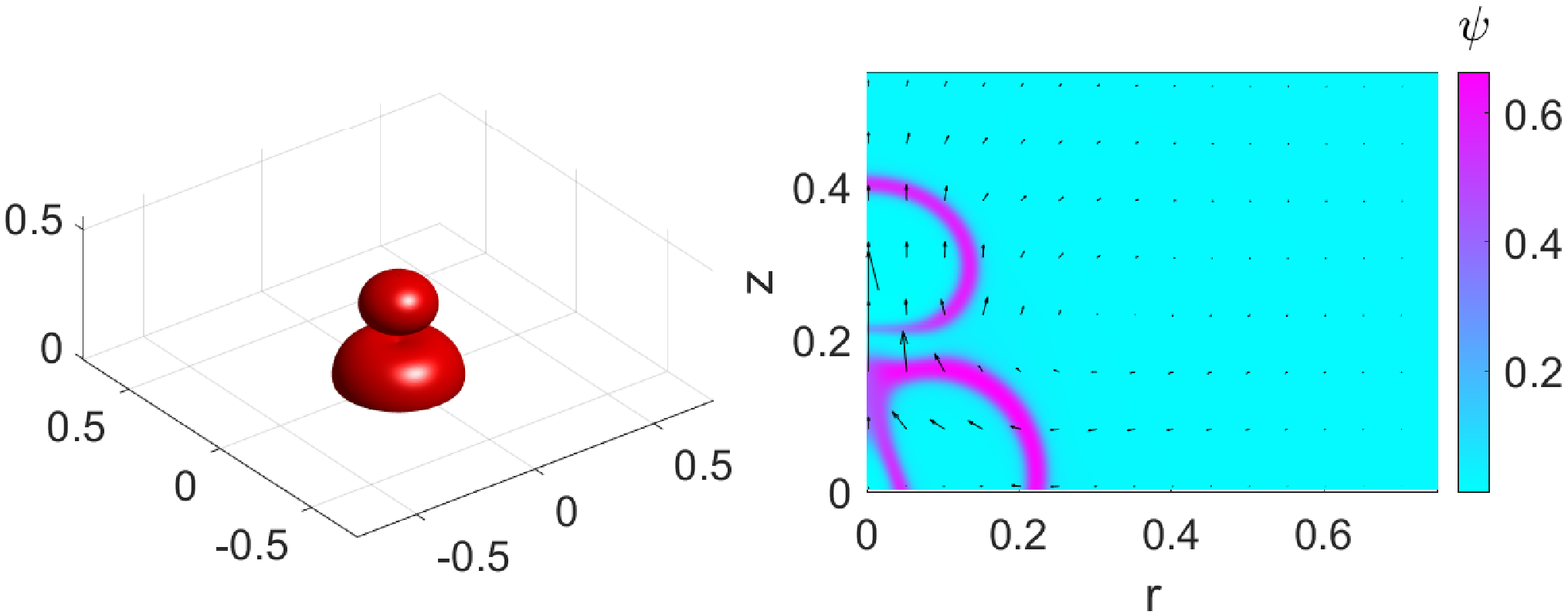}
\put(2,35){\scriptsize{(h)}}
\end{overpic}
\begin{overpic}[trim=0cm 0cm 0cm 0cm, clip,scale=0.33]{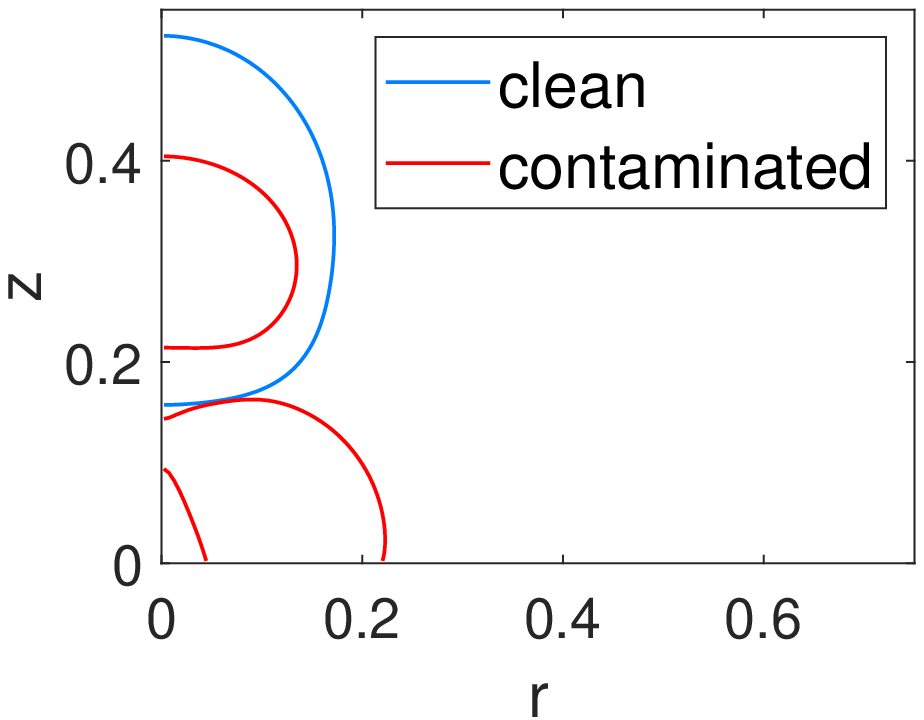}
\put(3,78){\scriptsize{(i)}}
\end{overpic}

\caption{Droplet bouncing in Example 4 ($\mathrm{Re}=2000$, $\mathrm{We}=282.8$, $\theta_s=100^\circ$, and $\mathrm{Pe}_\psi=100$). Profiles of clean and contaminated droplets are shown in both three-dimensional views (1st and 3rd columns) and two-dimensional radial plots ((2nd and 4th columns). Velocity fields and surfactant concentrations are also shown in the radial plots by quivers and colormaps respectively. Comparisons between the interface shapes of clean and contaminated droplets are given in the 5th column. Snapshots are captured at time $t = 2.04$ (a-c), $t = 7.92$ (d-f), and $t = 10.64$ (g-i).}
\label{Bouncing_2}
\end{figure}
\begin{figure}[t!]
\centering
\includegraphics[scale=0.7]{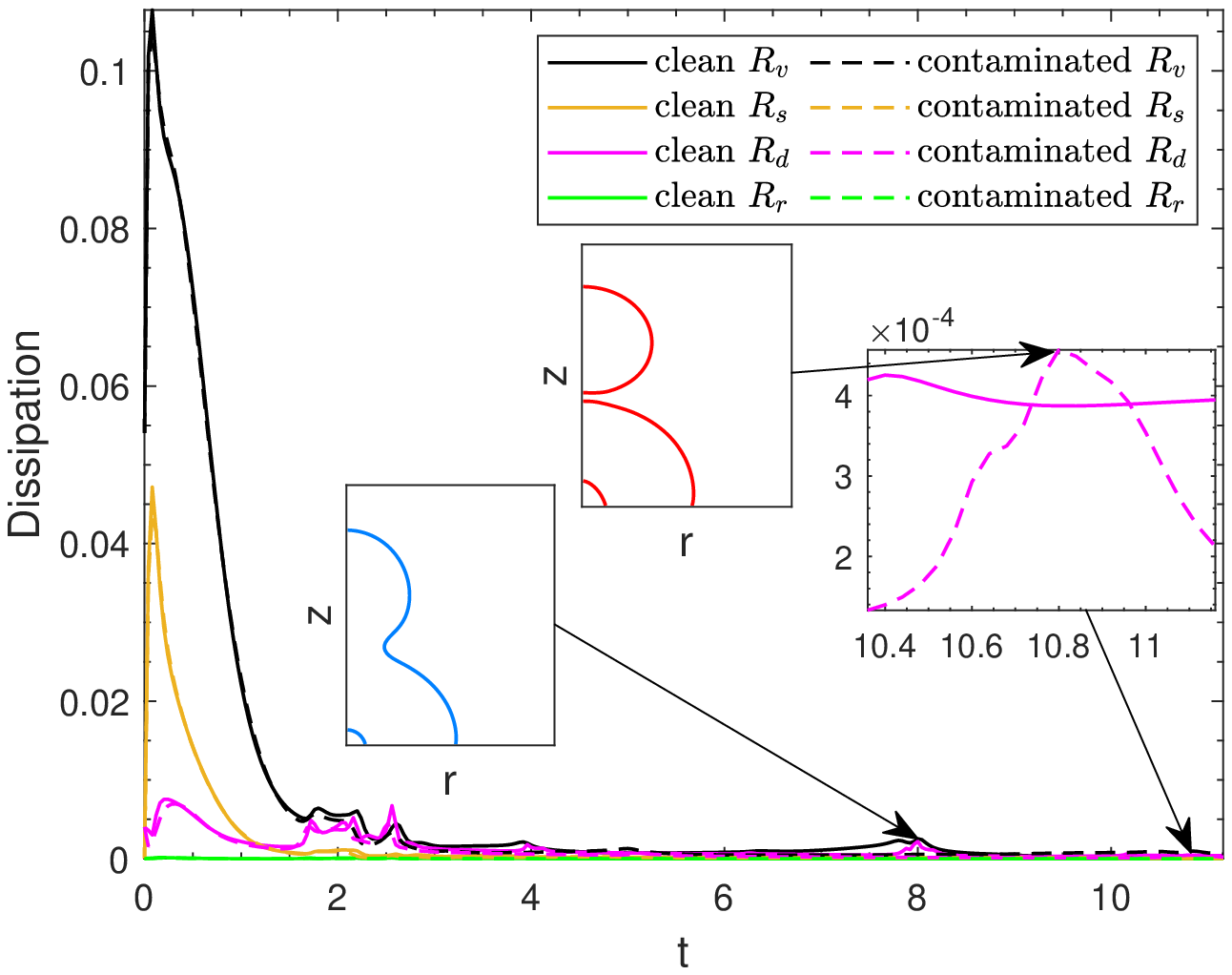}
    \caption{Dissipations in Example 4. In (a), four dissipations are shown for both clean (solid curves) and contaminated (dashed curves) cases. The right inset plot illustrates differences in $R_d$ for the two cases when topological change occurs, and some particular interface geometries corresponding to some of peaks are shown in other inset plots.}
\label{Bouncing_2_dissi}
\end{figure}

\subsubsection{Splashing}\label{sec_Splashing}

When the substrate is more hydrophobic so that it becomes less stick, droplet splashing will happen.

\paragraph{Example 5} We investigate splashing using the following parameters:
\begin{equation*}
\mathrm{Re}=1600, \quad\quad  \mathrm{We}=282.8, \quad\quad \theta_s=160^\circ,\quad\quad \mathrm{Pe}_\psi=100.
\end{equation*}
After droplet impact, the spreading dynamics follows. However, since the substrate is less sticky, the droplet does not recoil in a whole. Instead, the moving front of droplet upwarps its head, like a spindrift (Fig.~\ref{Splashing}a-c). As inertial effect is more important than capillary effect \citep{Richard2002} at the moving front, breakup occurs and the head of the moving front takes off and leaves the substrate (Fig.~\ref{Splashing}d-f). In other words, splashing happens. The presence of surfactant makes the interface softer and break up into more flying drops (Fig.~\ref{Splashing}g-i). This is consistent with our understanding on surfactants, which make the droplet more “active” to the forcing movement and thus help dewetting.

As in previous examples, the dissipations in the contaminated case are in general smaller than that in the clean case (see Fig.~\ref{Splashing_dissi}), which implies that more free energy is carried by the contaminated droplet. This additional free energy is either transformed into kinetic energy or stored on complex interfaces with larger total area. As a consequence, breakups of droplet and bouncing (or splashing) are more likely to happen.


\begin{figure}[t!]
\center
\begin{overpic}[trim=0cm 0cm 0cm 0cm, clip,scale=0.3]{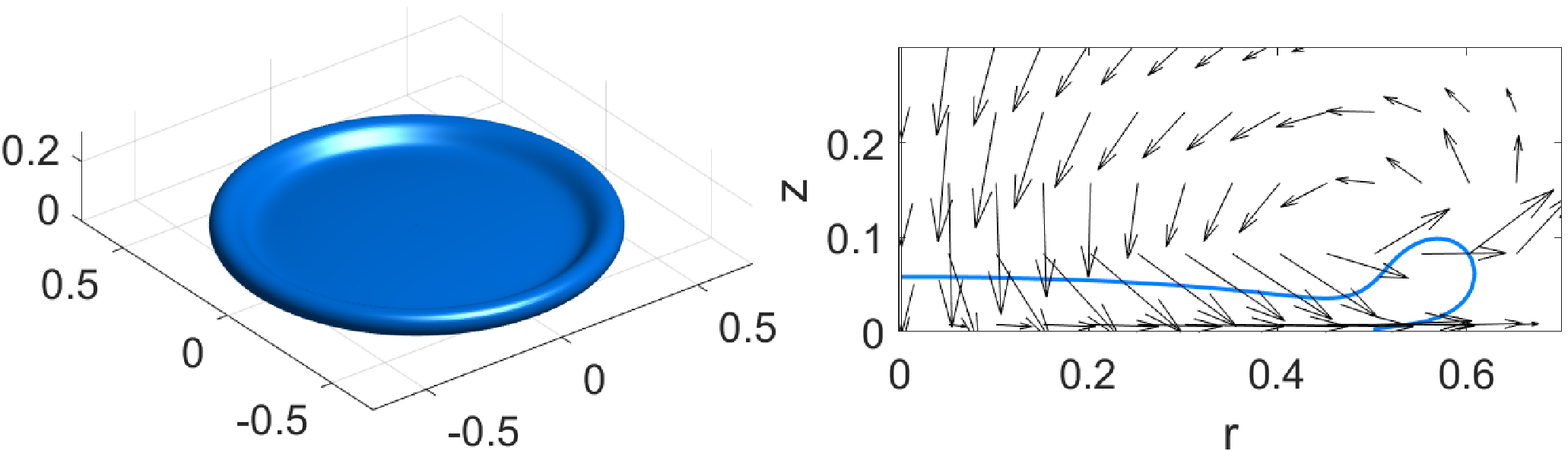}
\put(2,26){\scriptsize{(a)}}
\end{overpic}
\hspace{-0.6cm}
\begin{overpic}[trim=0cm 0cm 0cm 0cm, clip,scale=0.3]{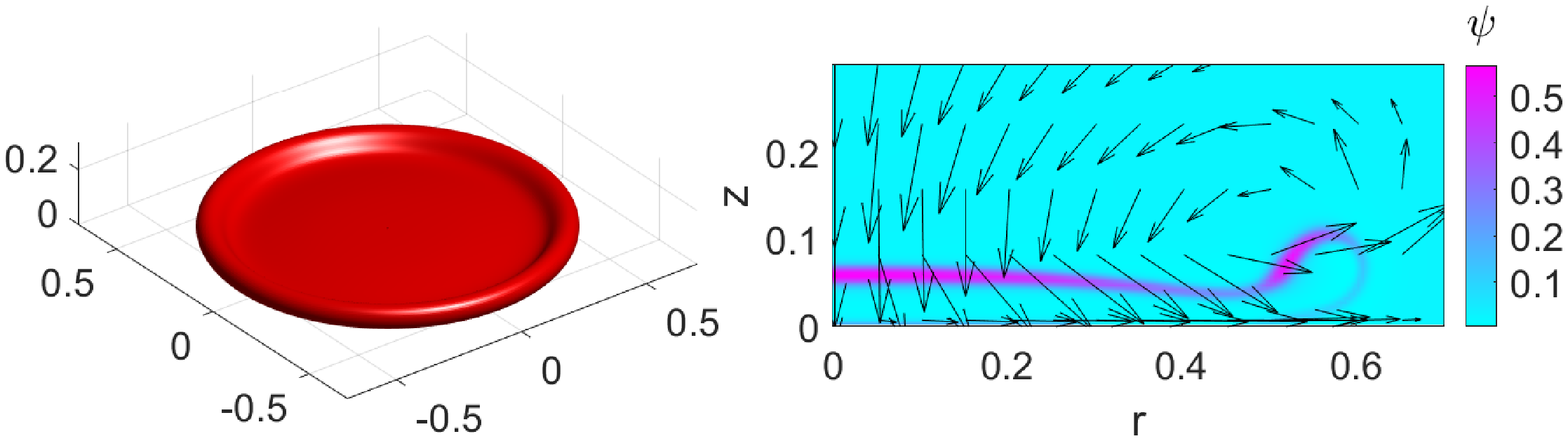}
\put(2,26){\scriptsize{(b)}}
\end{overpic}
\hspace{-0.1cm}
\begin{overpic}[trim=0cm 0cm 0cm 0cm, clip,scale=0.3]{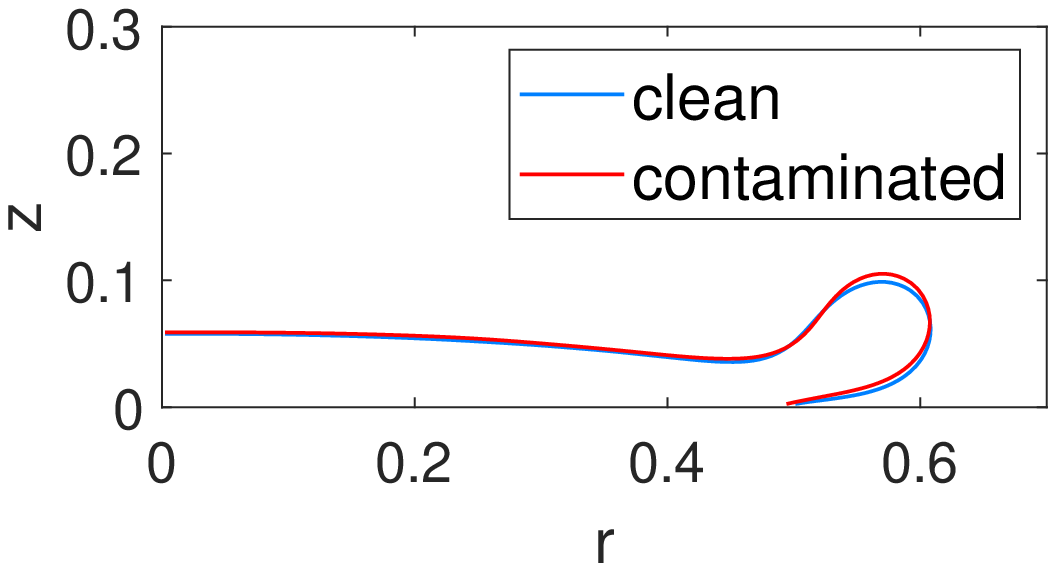}
\put(0,56){\scriptsize{(c)}}
\end{overpic}

\begin{overpic}[trim=0cm 0cm 0cm 0cm, clip,scale=0.3]{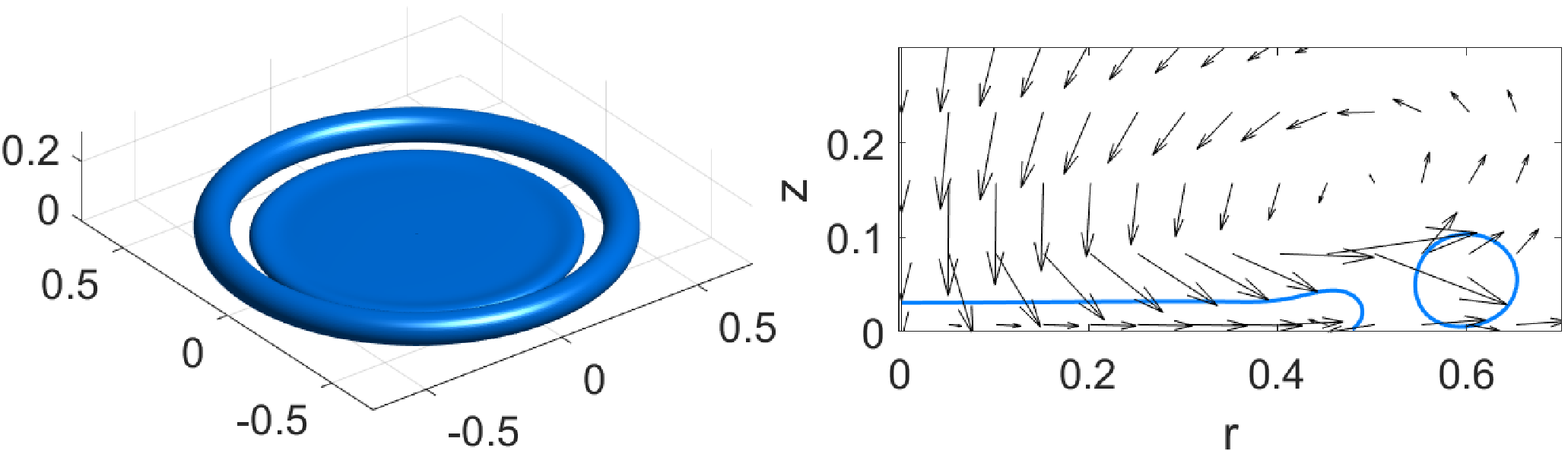}
\put(2,26){\scriptsize{(d)}}
\end{overpic}
\hspace{-0.6cm}
\begin{overpic}[trim=0cm 0cm 0cm 0cm, clip,scale=0.3]{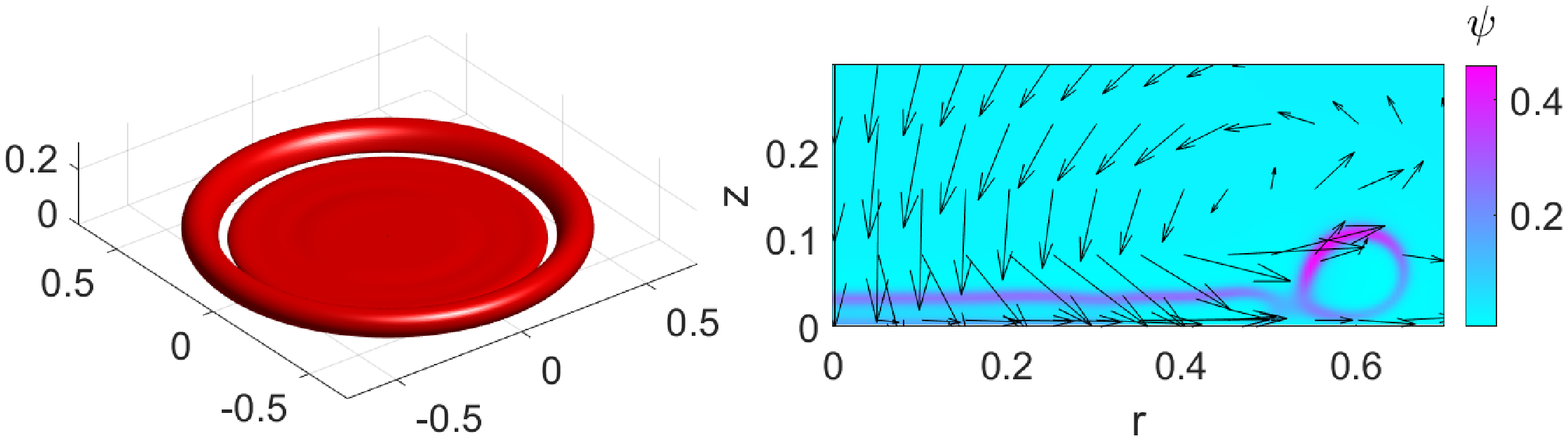}
\put(2,26){\scriptsize{(e)}}
\end{overpic}
\hspace{-0.1cm}
\begin{overpic}[trim=0cm 0cm 0cm 0cm, clip,scale=0.3]{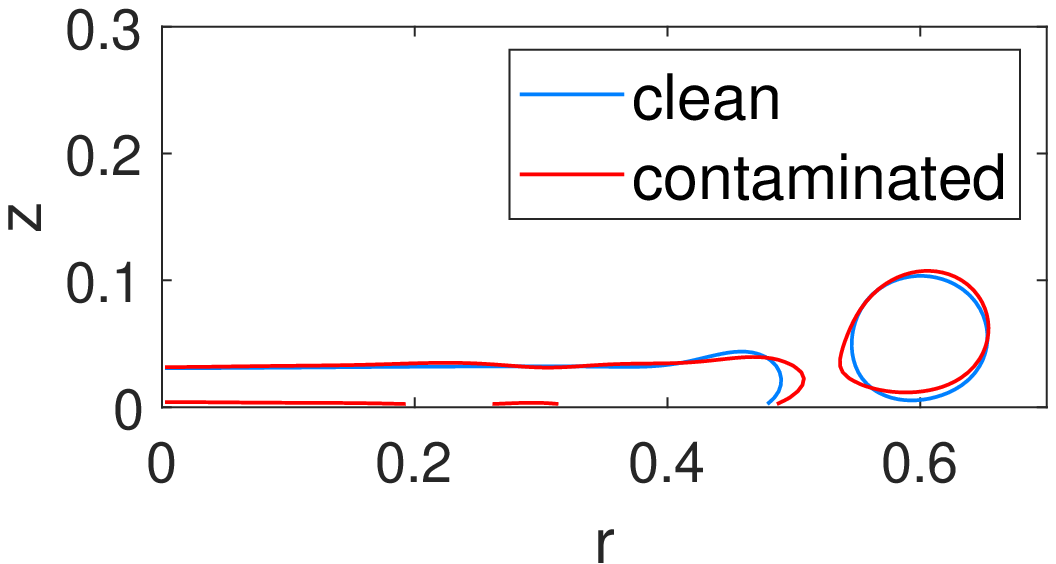}
\put(0,56){\scriptsize{(f)}}
\end{overpic}

\begin{overpic}[trim=0cm 0cm 0cm 0cm, clip,scale=0.3]{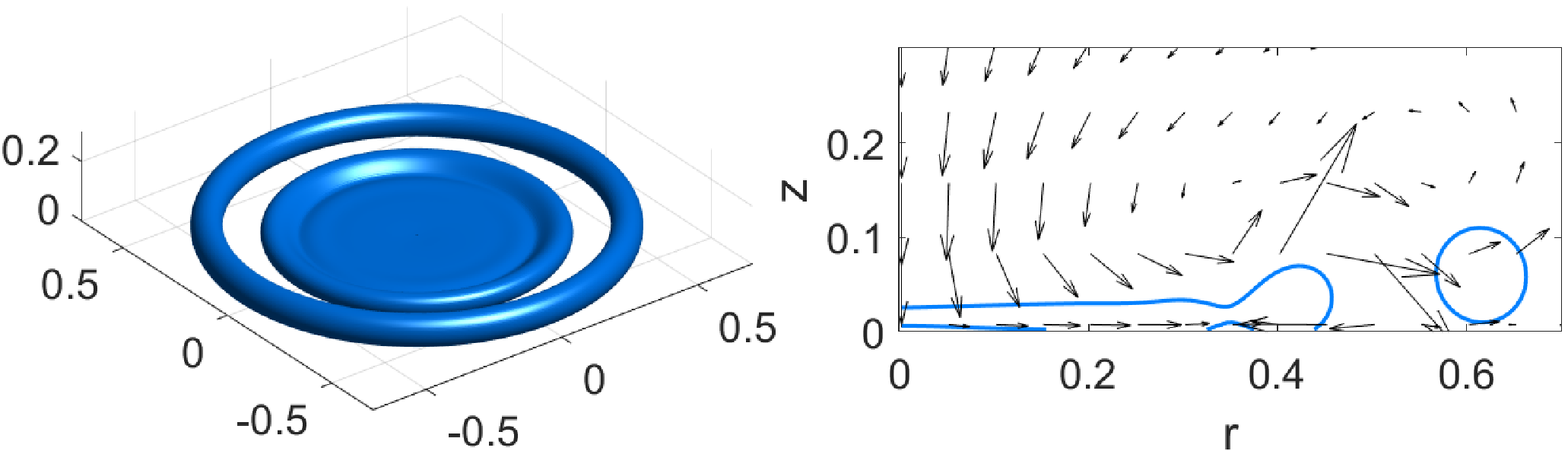}
\put(2,26){\scriptsize{(g)}}
\end{overpic}
\hspace{-0.6cm}
\begin{overpic}[trim=0cm 0cm 0cm 0cm, clip,scale=0.3]{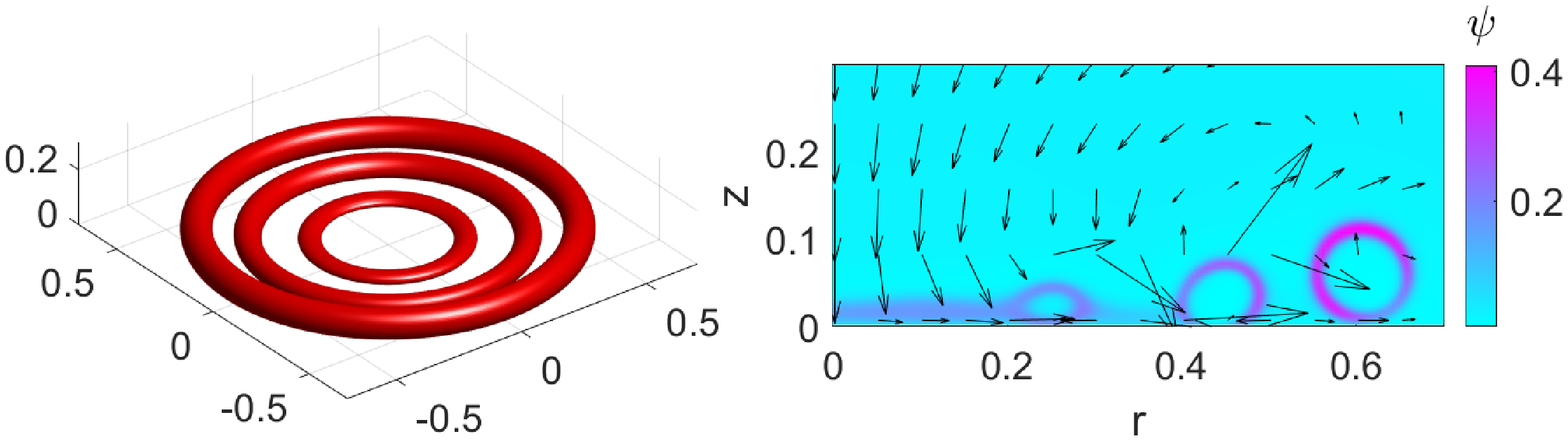}
\put(2,26){\scriptsize{(h)}}
\end{overpic}
\hspace{-0.1cm}
\begin{overpic}[trim=0cm 0cm 0cm 0cm, clip,scale=0.3]{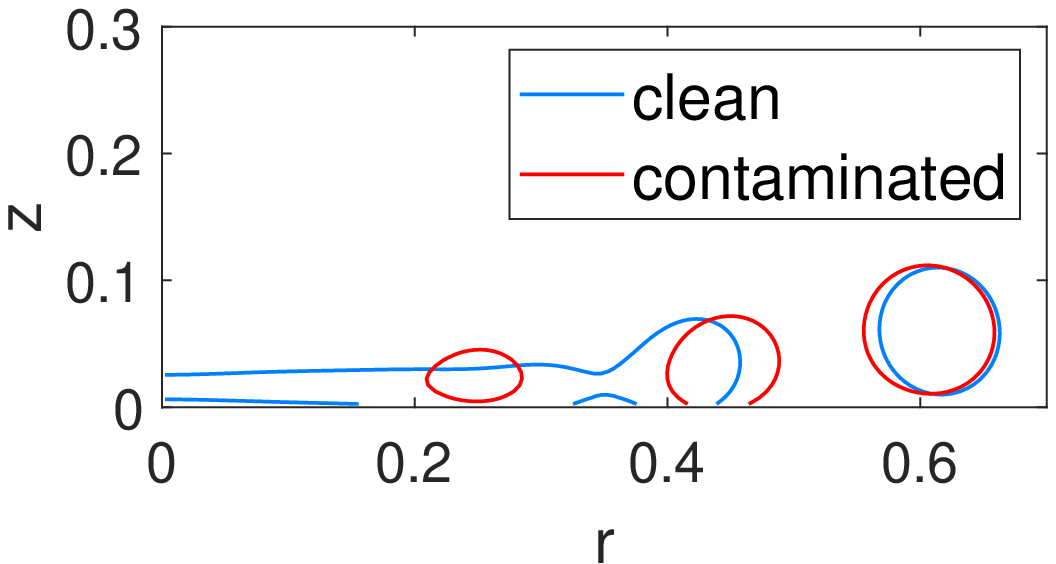}
\put(0,56){\scriptsize{(i)}}
\end{overpic}

\caption{Droplet splashing in Example 5 ($\mathrm{Re}=1600$, $\mathrm{We}=282.8$, $\theta_s=160^\circ$, and $\mathrm{Pe}_\psi=100$). Profiles of clean and contaminated droplets are shown in both three-dimensional views (1st and 3rd columns) and two-dimensional radial plots ((2nd and 4th columns). Velocity fields and surfactant concentrations are also shown in the radial plots by quivers and colormaps respectively. Comparisons between the interface shapes of clean and contaminated droplets are given in the 5th column. Snapshots are captured at time $t = 0.8$ (a-c), $t = 1.22$ (d-f), and $t = 1.48$ (g-i).}
\label{Splashing}
\end{figure}
\begin{figure}[t!]
\centering
\includegraphics[scale=0.7]{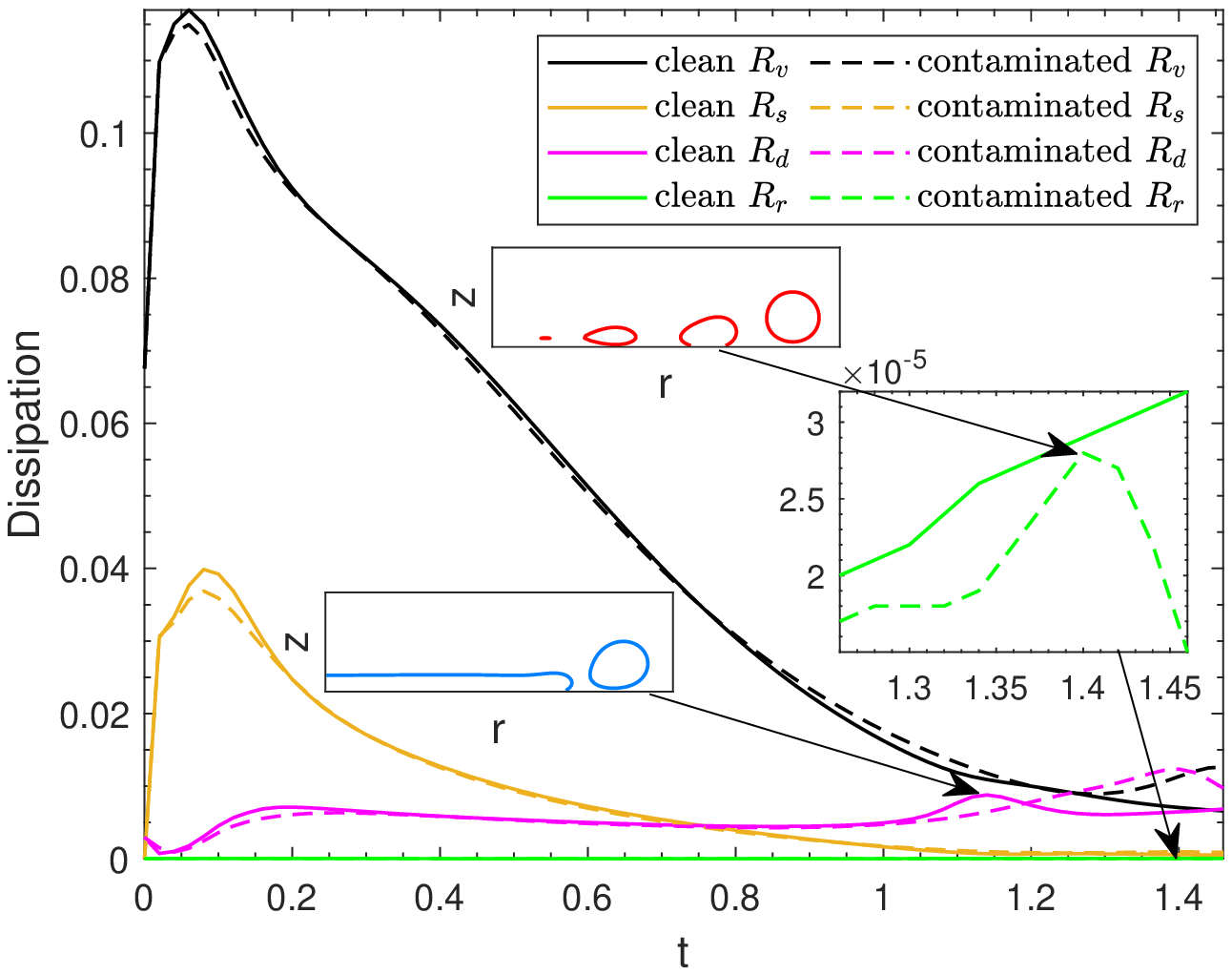}
    \caption{Dissipations in Example 5. In (a), four dissipations are shown for both clean (solid curves) and contaminated (dashed curves) cases. The right inset plot illustrates differences in $R_r$ for the two cases when topological change occurs, and some particular interface geometries corresponding to some of peaks are shown in other inset plots.}
\label{Splashing_dissi}
\end{figure}

\subsubsection{Comparison with experimental results}\label{sec_Reality}
Experiments on droplet impact have been carried out in many literatures. For instance,  \cite{Rioboo2002} systematically studied adherence, bouncing, and partial bouncing of clean drops in their experiments. In this section, we apply our model and numerical methods to reproduce some of their experimental results. For comparison, we also make numerical predictions of impact behaviors for contaminated droplets under the same physical parameters.

\paragraph{Example 6} We study bouncing dynamics in this example. The same parameters as in the experiment are used, which after scaling are equivalent to
\begin{equation*}
\mathrm{Re}=6490,\quad\quad \mathrm{We}=98,\quad\quad \mathrm{L}_{s}=0.0025,\quad\quad \theta_s = 95^\circ, \quad\quad \lambda_\rho=\frac{1}{830},\quad\quad\lambda_\eta=\frac{1}{66.2}.
\end{equation*}
In addition, we choose $\mathrm{Pe}_\psi=100$ for contaminated droplet in order to see the surfactant transport during impact dynamics. The computational domain is fixed to be $\Omega =[0, 1.5] \times[0, 2]$.

Experimental and numerical results for clean droplets are compared in Fig.~\ref{Reality_1}. We see almost quantitative agreement between these two sets of results in spreading, recoiling and bouncing processes. In the recoiling stage (Fig.~\ref{Reality_1}h,i), capillary wave instability results in ruptures of the thin film, creating a
ring structure with dried out regions. The emergence of a ring structure with a dry-out region
for large enough $\mathrm{Re}$ and $\mathrm{We}$ was also reported experimentally and numerically in \cite{Renardy2003}. In Fig.~\ref{Reality_1}n,q, the rising droplet breaks up into two parts, with one part staying on the surface and the other flying in the air. This leads to the result of partial bouncing for the clean droplet.

When surfactant is present, it gradually migrates from the center part of interface towards moving front (Fig.~\ref{Reality_1}c,f) as a result of convection and diffusion (Marangoni effect also helps). When recoiling, the contaminated droplet breaks up into three tori in comparison to the clean case (Fig.~\ref{Reality_1}i). The hole in the center continues to exist even in the bouncing process (Fig.~\ref{Reality_1}l,o). Breakup also occurs in the rising droplet, but with both child drops flying in the air. This phenomenon can also be attributed to the lower dissipations in impact dynamics of contaminated droplet, as we explained in Example 5.
\begin{figure}[htbp!]
\begin{overpic}[trim=0cm -4cm 0cm 0.6cm, clip,scale=0.13]{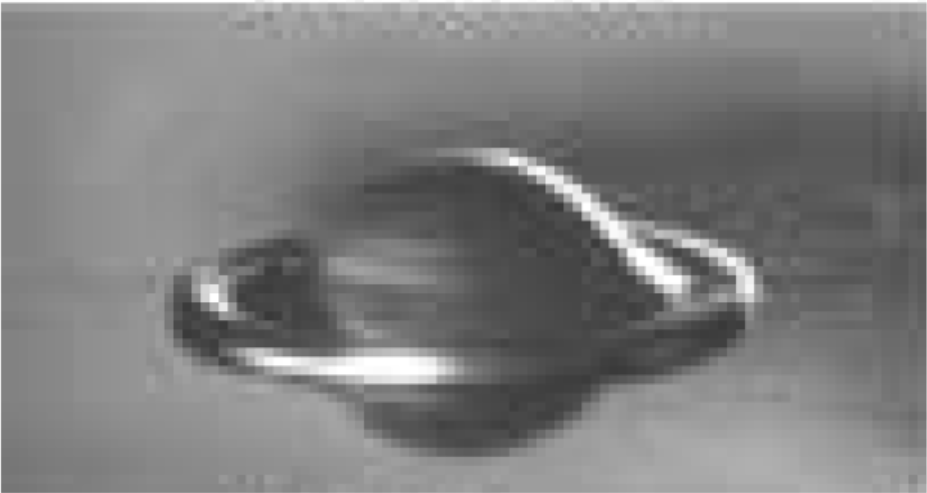}
\put(0,84){\scriptsize{(a)}}
\end{overpic}
\hspace{0.3cm}
\begin{overpic}[trim=0cm 0cm 0cm 0.6cm, clip,scale=0.36]{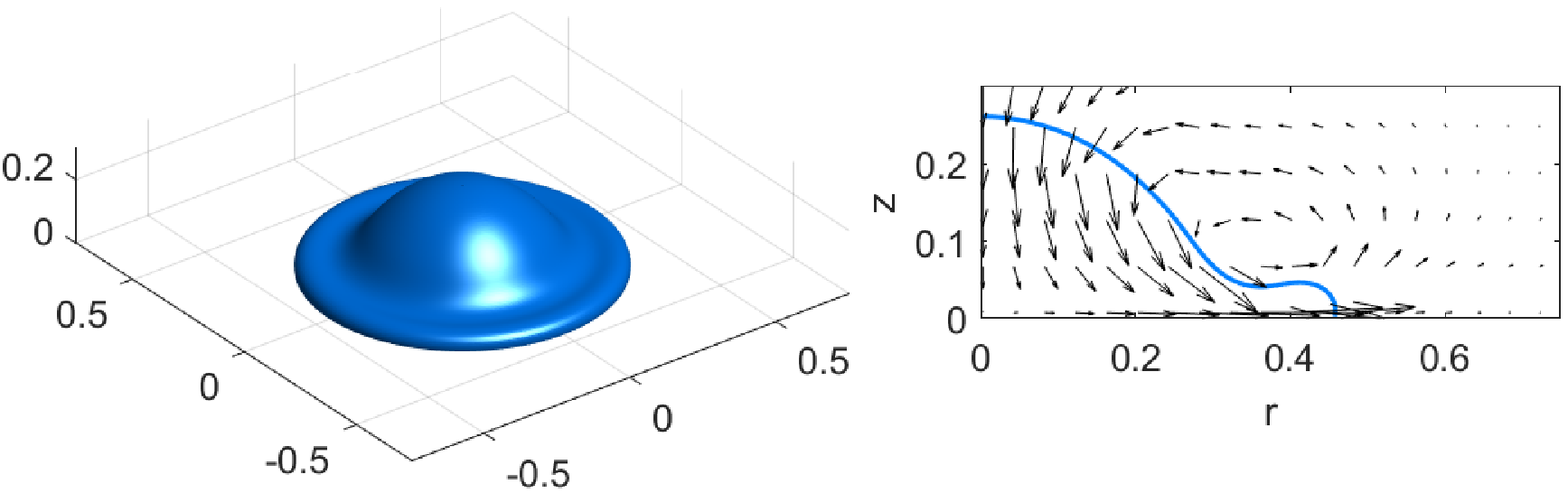}
\put(1,26){\scriptsize{(b)}}
\end{overpic}
\hspace{0.1cm}
\begin{overpic}[trim=0cm 0cm 0cm 0.6cm, clip,scale=0.36]{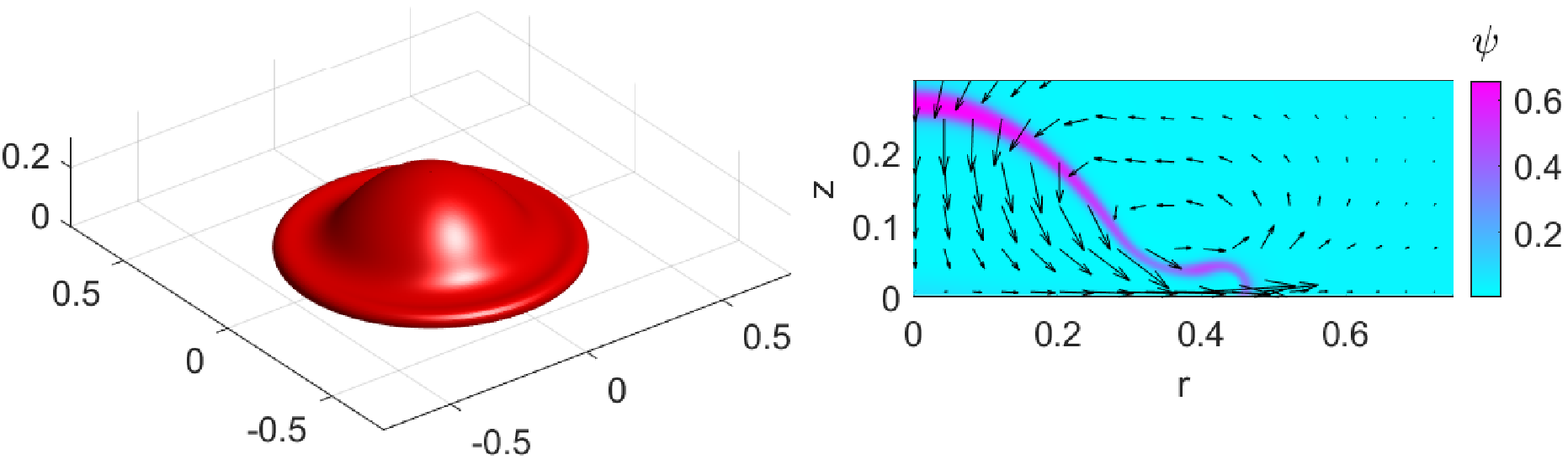}
\put(1,26){\scriptsize{(c)}}
\end{overpic}

\begin{overpic}[trim=0cm -4cm 0cm 0.6cm, clip,scale=0.13]{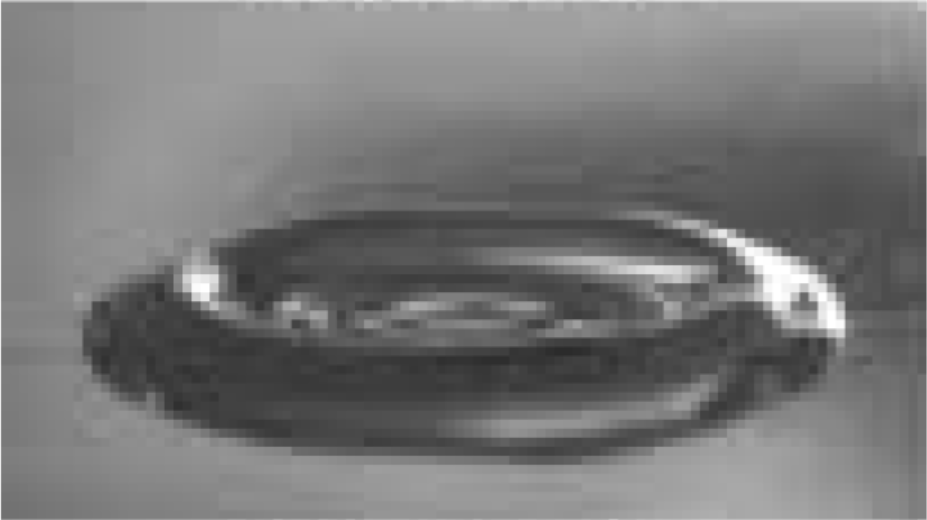}
\put(0,84){\scriptsize{(d)}}
\end{overpic}
\hspace{0.3cm}
\begin{overpic}[trim=0cm 0cm 0cm 0.6cm, clip,scale=0.36]{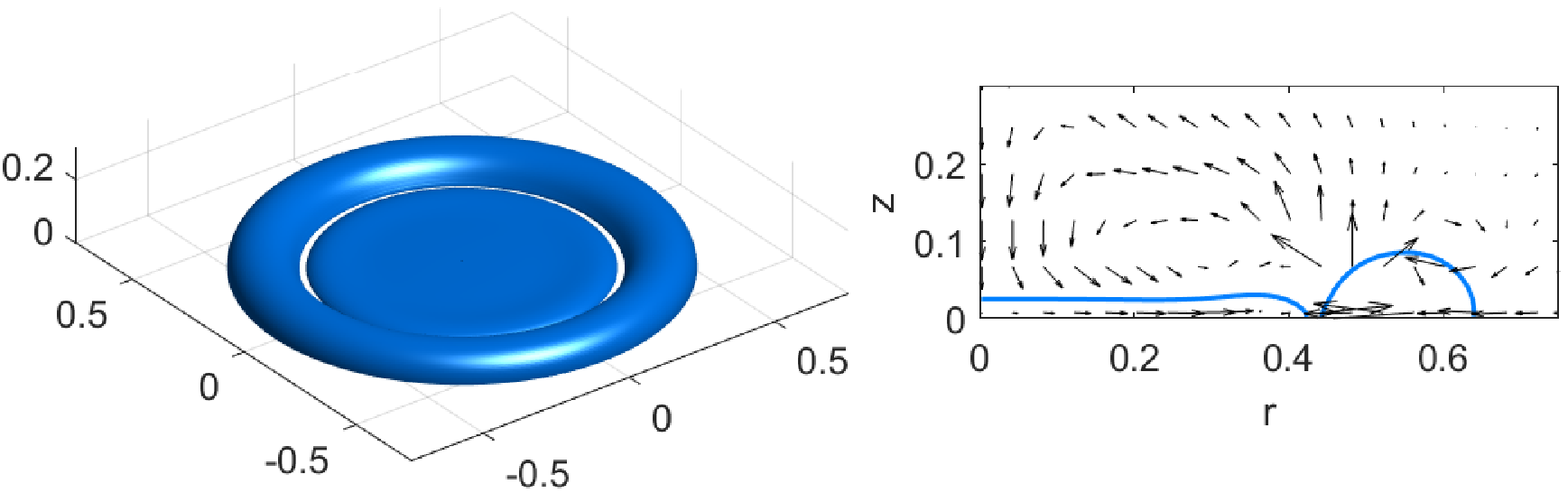}
\put(1,26){\scriptsize{(e)}}
\end{overpic}
\hspace{0.1cm}
\begin{overpic}[trim=0cm 0cm 0cm 0.6cm, clip,scale=0.36]{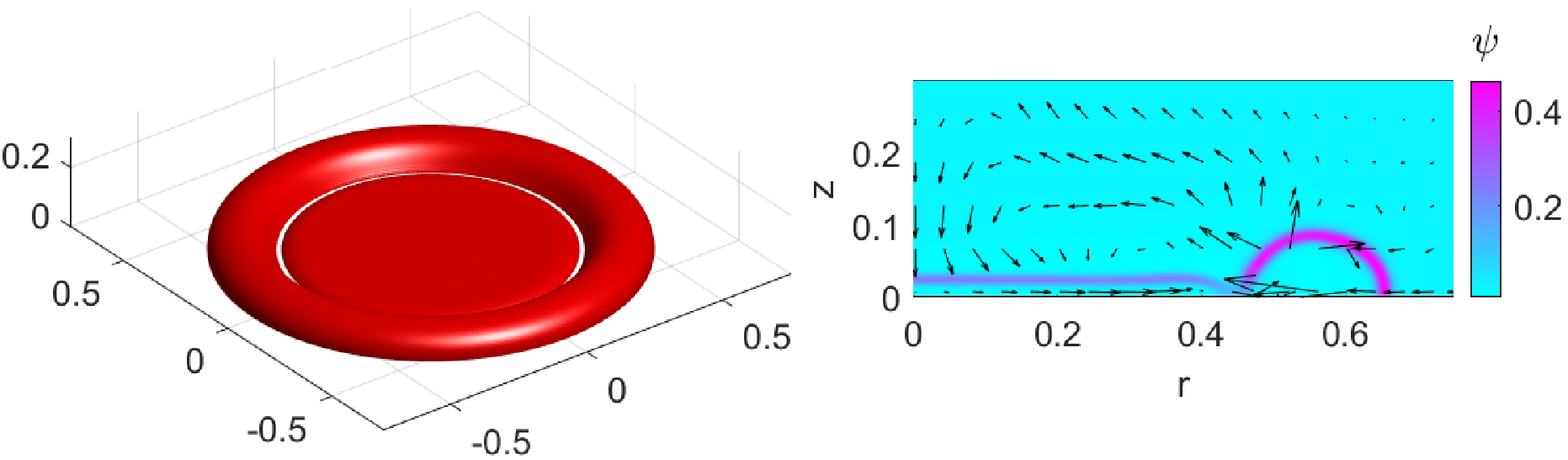}
\put(1,26){\scriptsize{(f)}}
\end{overpic}

\begin{overpic}[trim=0cm -4cm 0cm 0.6cm, clip,scale=0.13]{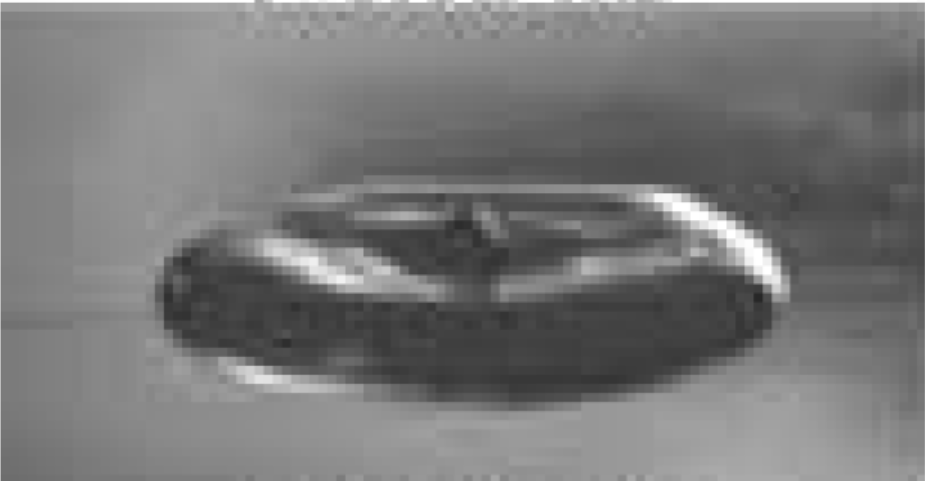}
\put(0,84){\scriptsize{(g)}}
\end{overpic}
\hspace{0.3cm}
\begin{overpic}[trim=0cm 0cm 0cm 0.6cm, clip,scale=0.36]{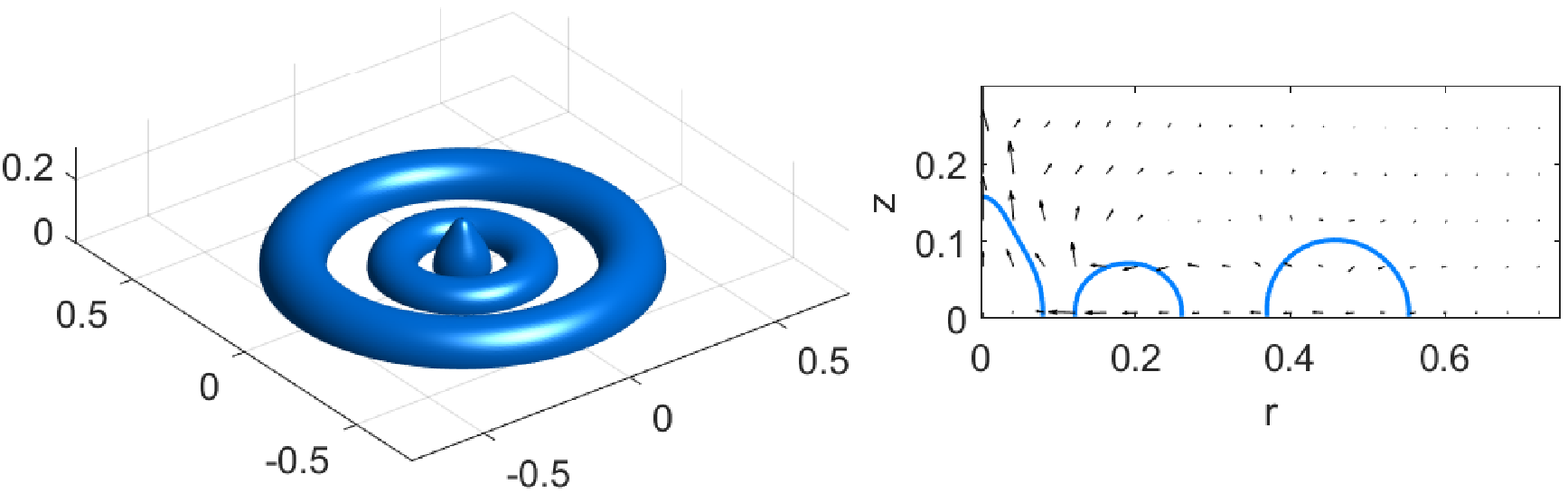}
\put(1,26){\scriptsize{(h)}}
\end{overpic}
\hspace{0.1cm}
\begin{overpic}[trim=0cm 0cm 0cm 0.6cm, clip,scale=0.36]{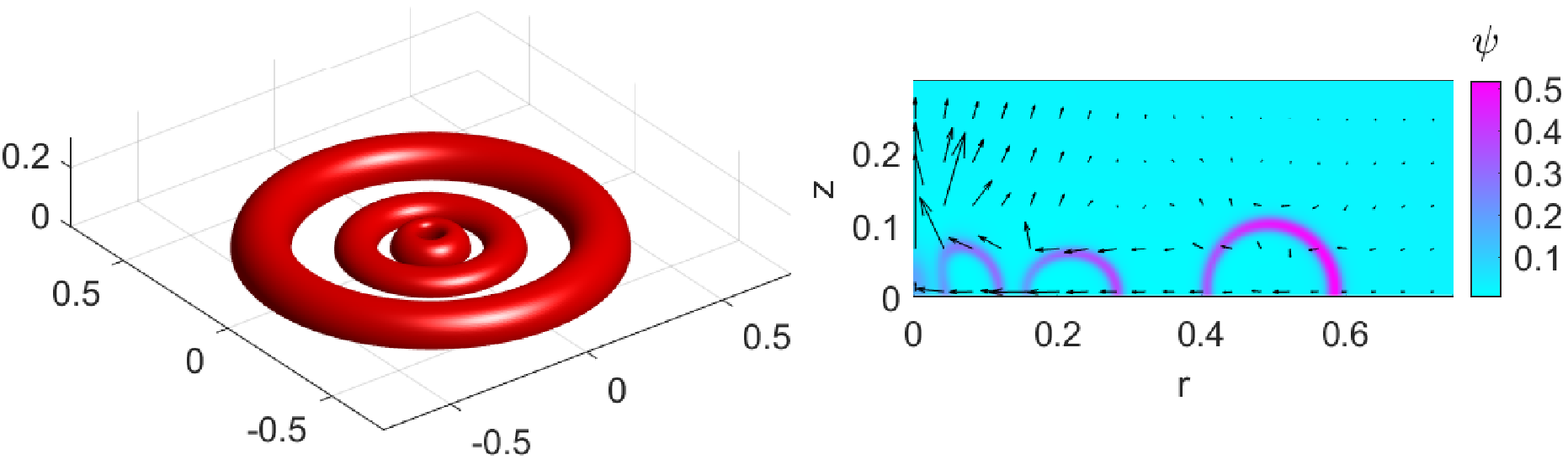}
\put(1,26){\scriptsize{(i)}}
\end{overpic}

\begin{overpic}[trim=0cm -8cm 0cm 0cm, clip,scale=0.13]{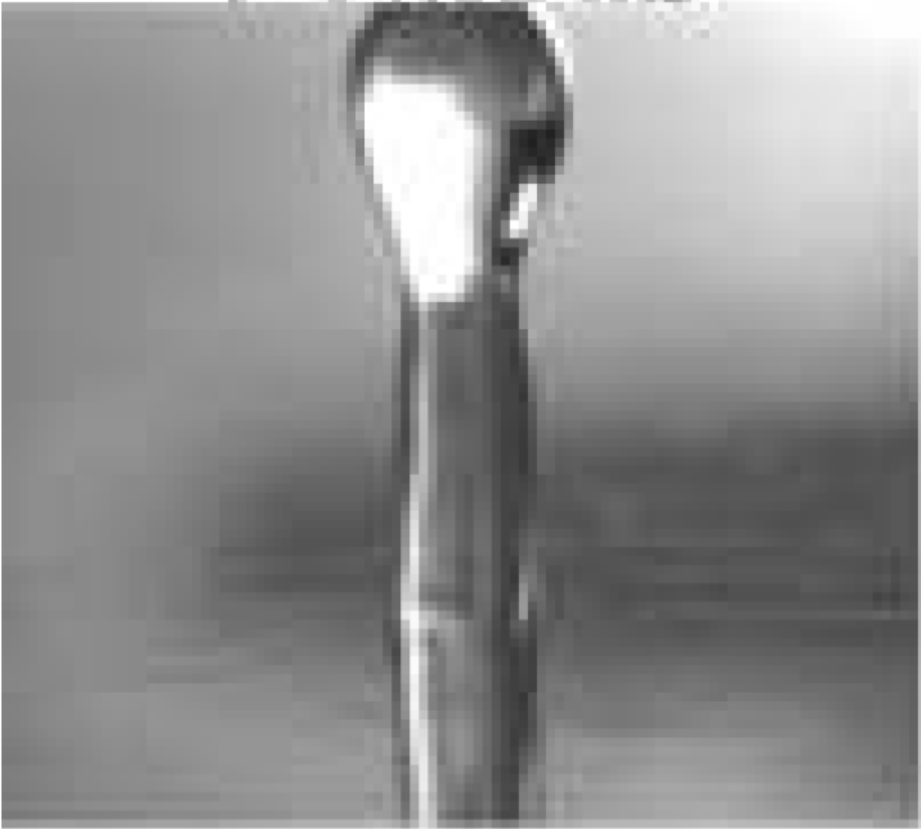}
\put(0,113){\scriptsize{(j)}}
\end{overpic}
\hspace{0.3cm}
\begin{overpic}[trim=0cm 0cm 0cm 0cm, clip,scale=0.36]{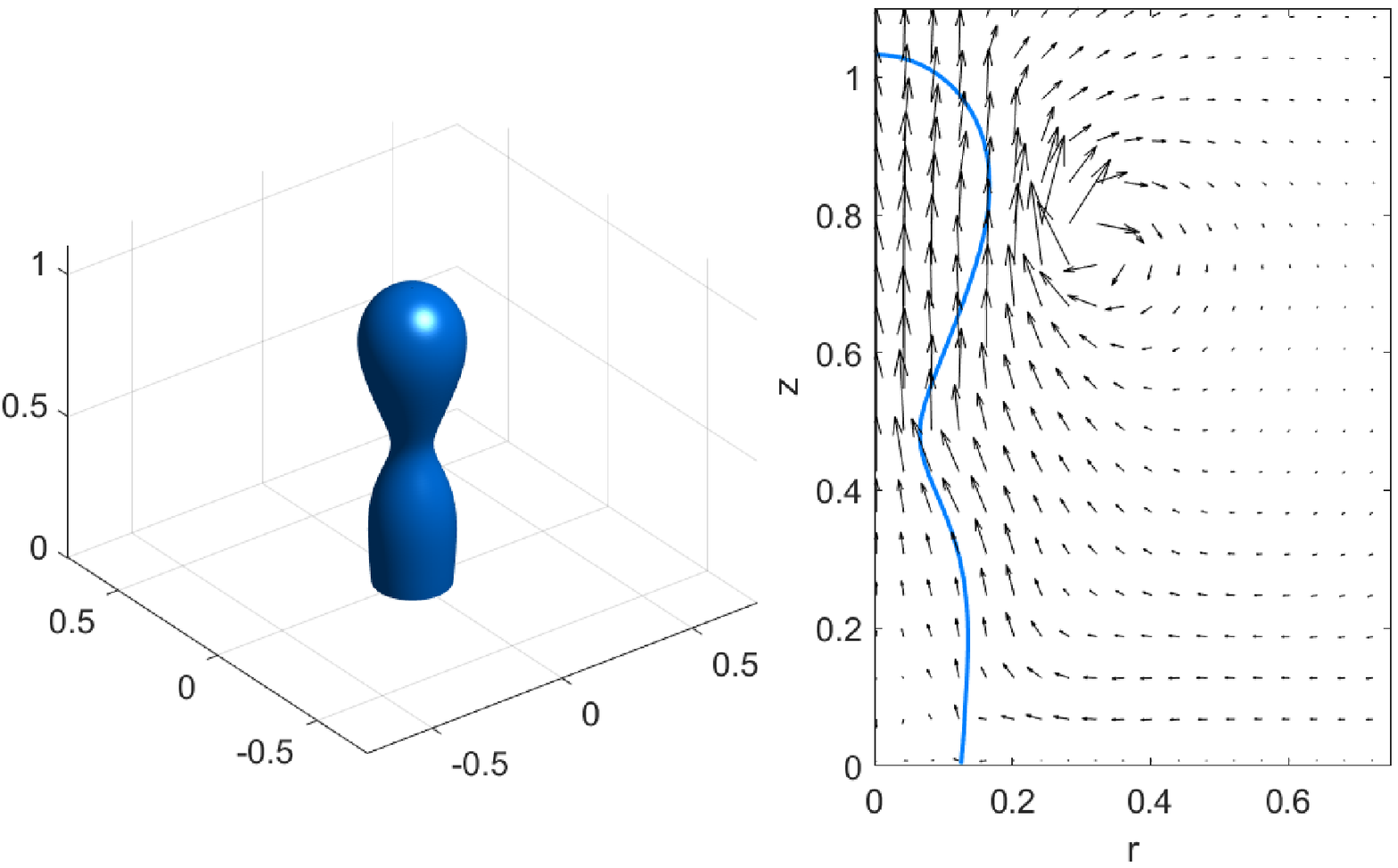}
\put(1,48){\scriptsize{(k)}}
\end{overpic}
\hspace{0.1cm}
\begin{overpic}[trim=0cm 0cm 0cm 0cm, clip,scale=0.36]{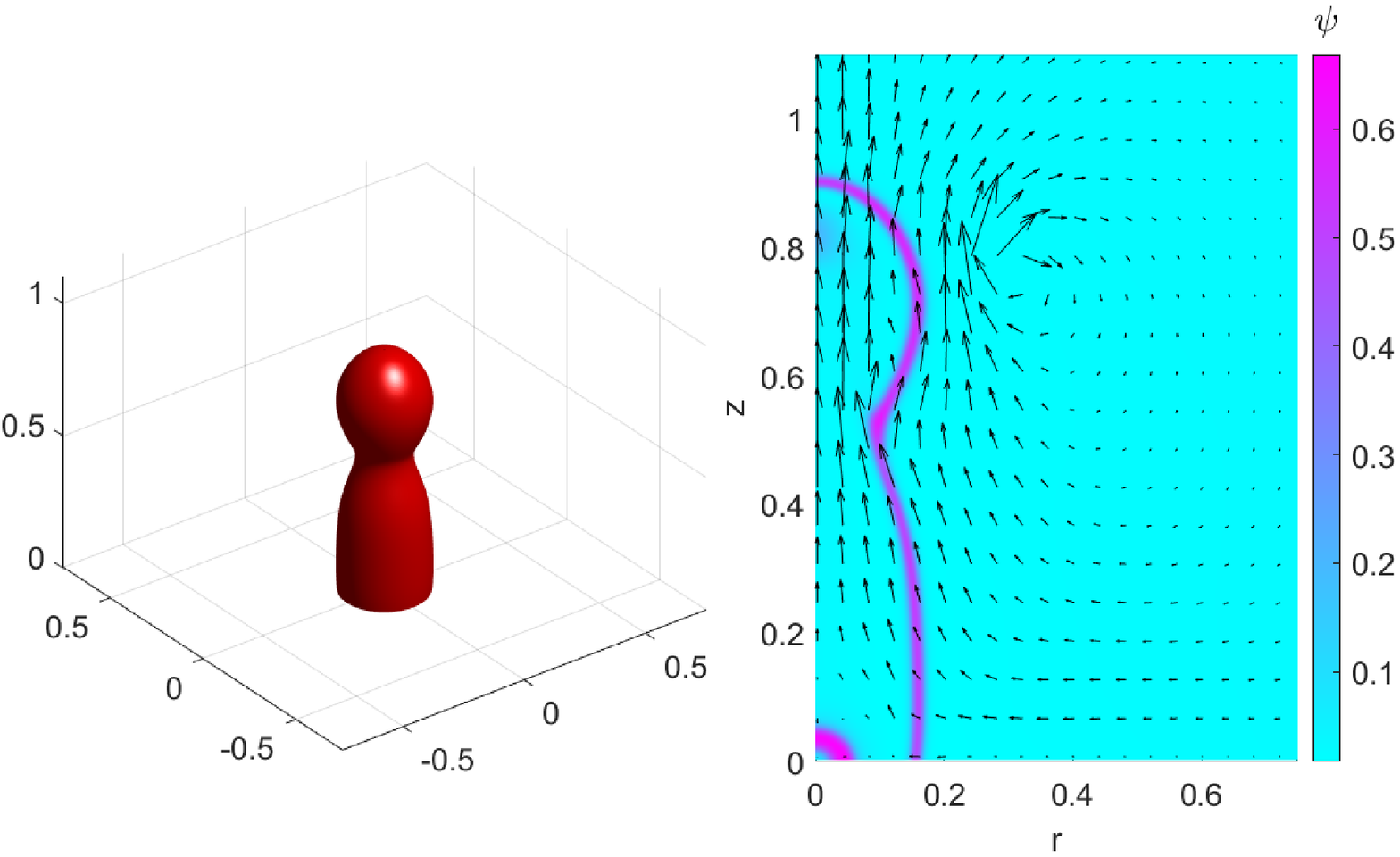}
\put(1,48){\scriptsize{(l)}}
\end{overpic}

\begin{overpic}[trim=0cm -10cm 0cm 0.1cm, clip,scale=0.13]{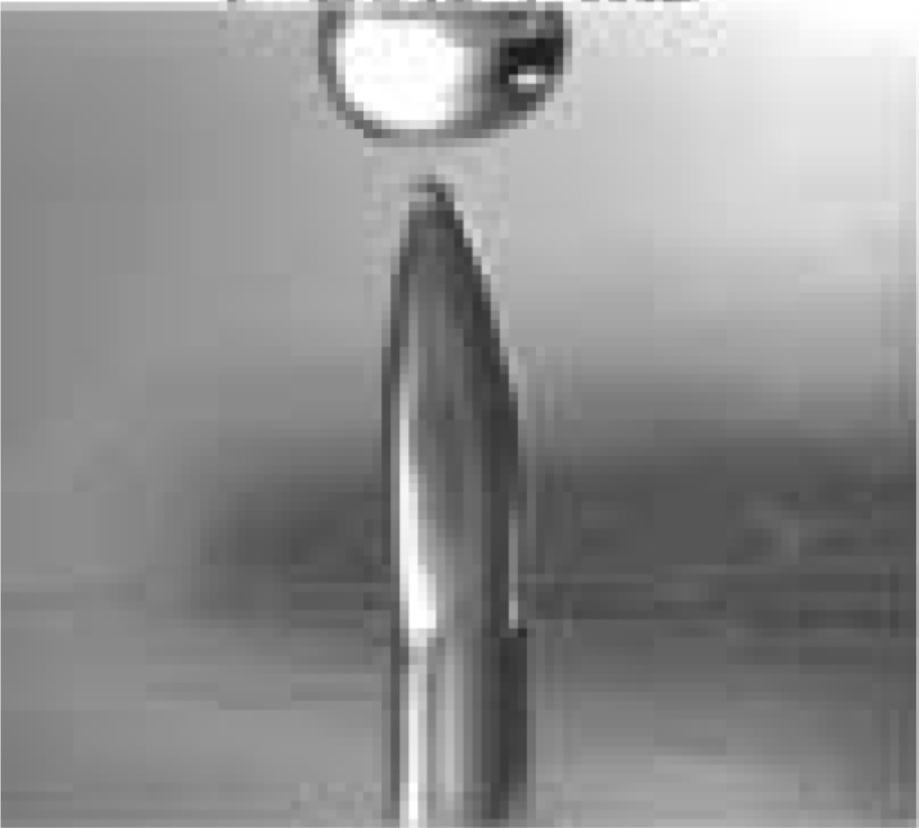}
\put(0,112){\scriptsize{(m)}}
\end{overpic}
\hspace{0.3cm}
\begin{overpic}[trim=0cm 0cm 0cm 0.1cm, clip,scale=0.36]{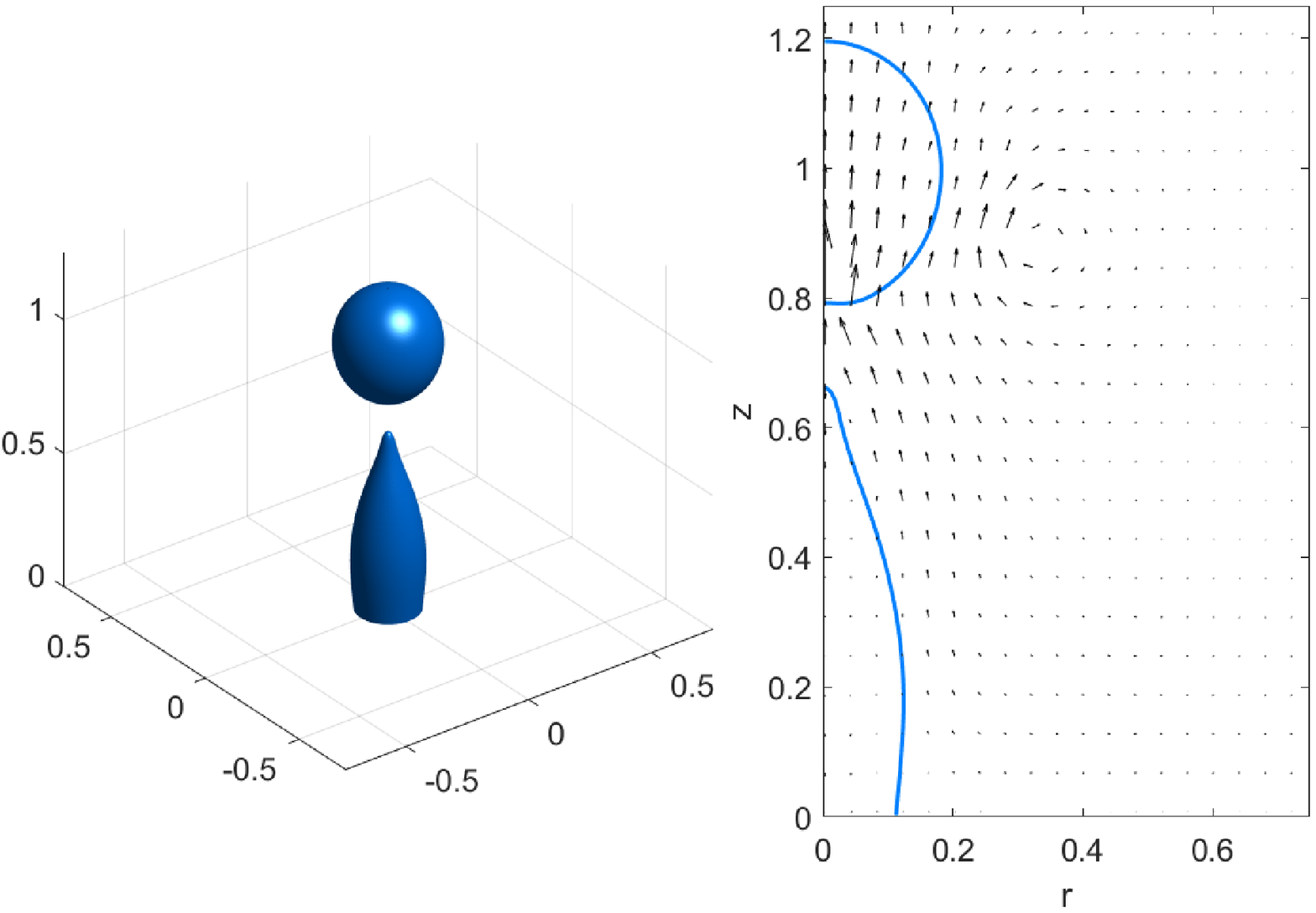}
\put(1,52){\scriptsize{(n)}}
\end{overpic}
\hspace{0.1cm}
\begin{overpic}[trim=0cm 0cm 0cm 0.1cm, clip,scale=0.36]{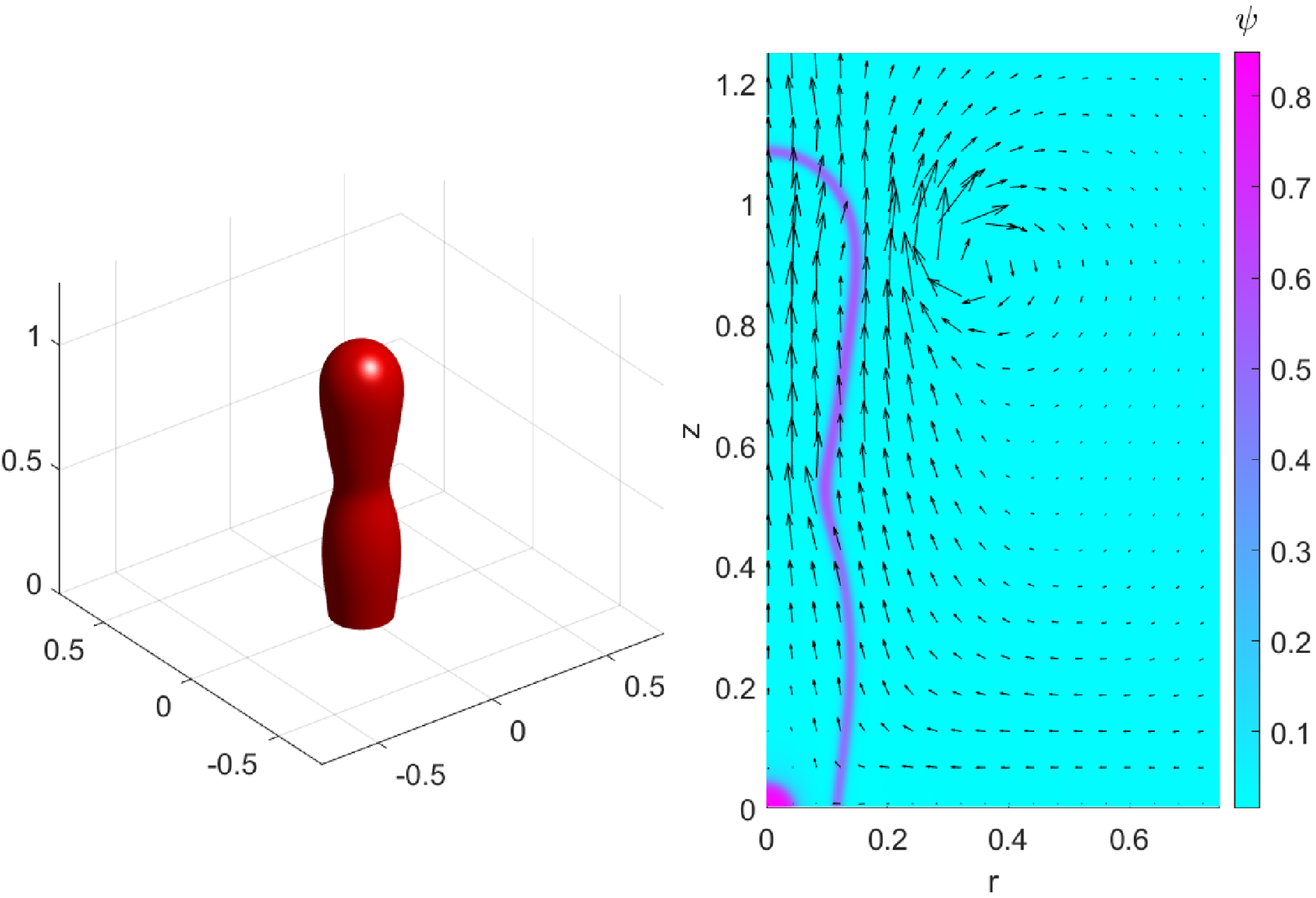}
\put(1,52){\scriptsize{(o)}}
\end{overpic}

\begin{overpic}[trim=0cm -14cm 0cm 0.2cm, clip,scale=0.13]{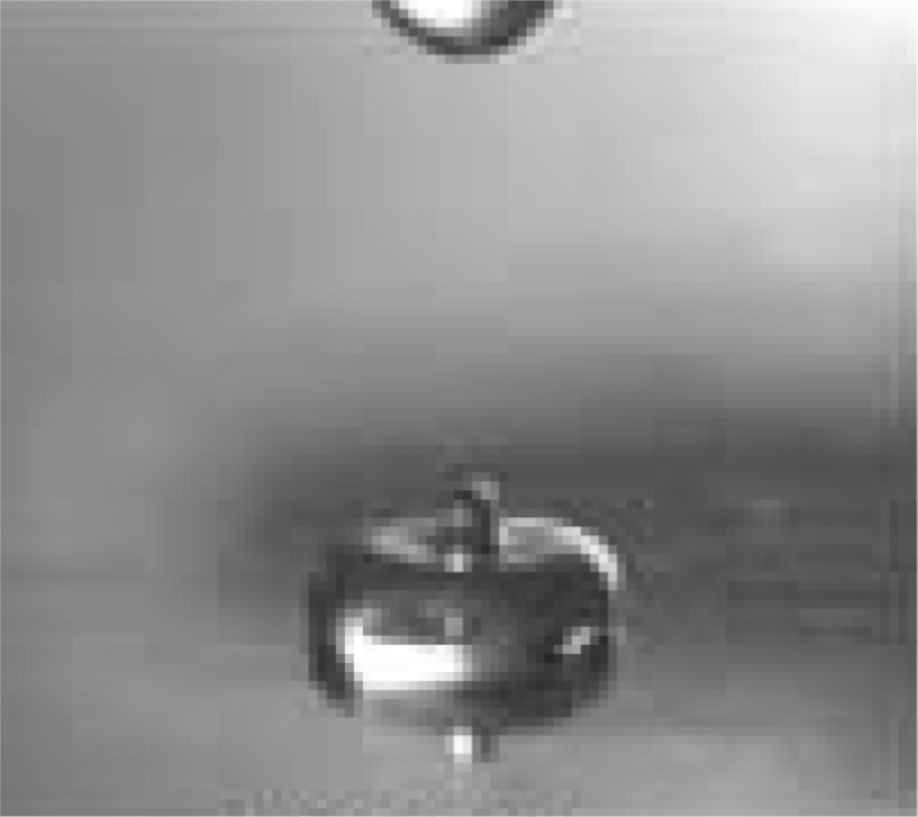}
\put(0,128){\scriptsize{(p)}}
\end{overpic}
\hspace{0.3cm}
\begin{overpic}[trim=0cm 0cm 0cm 0.2cm, clip,scale=0.36]{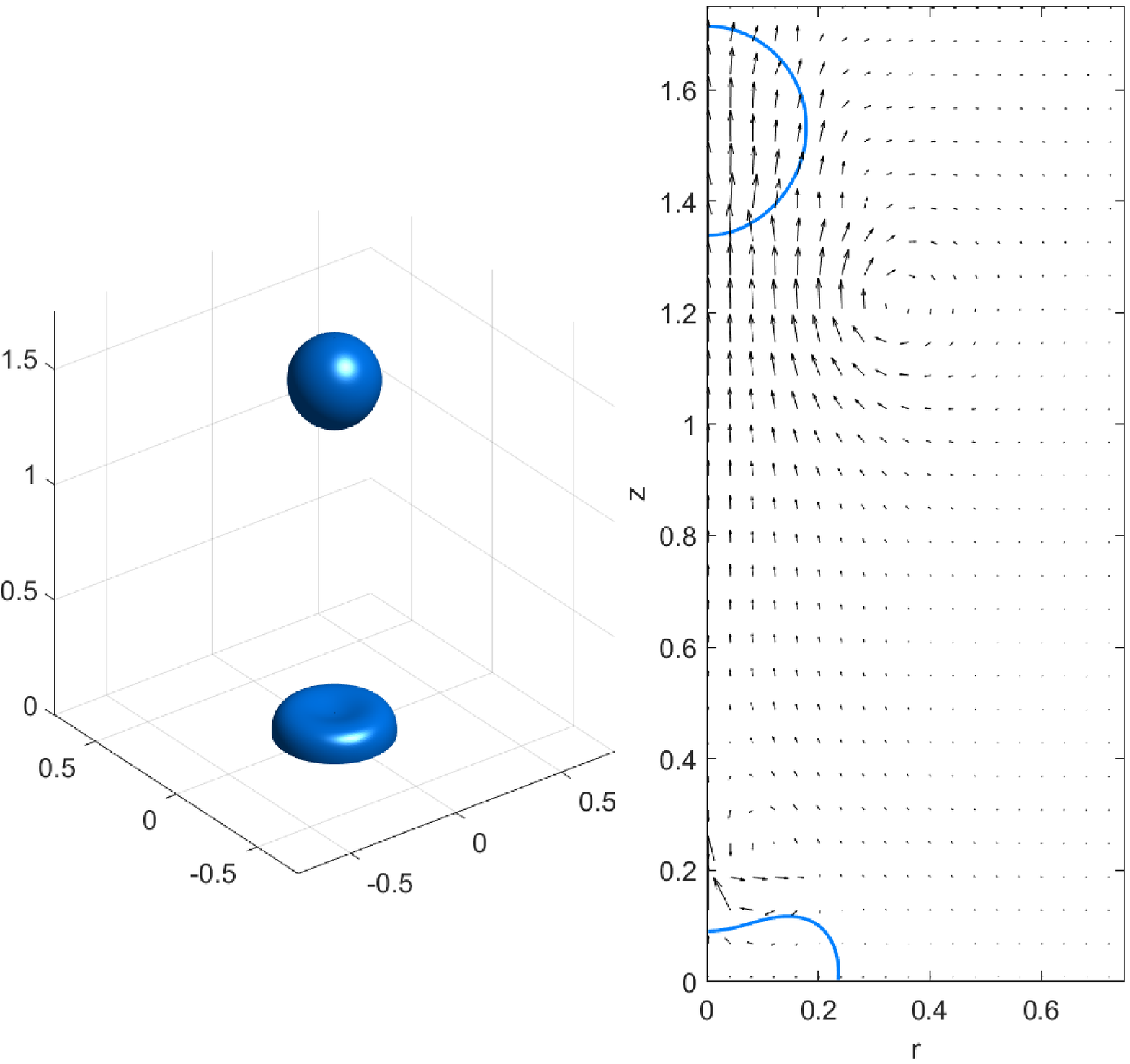}
\put(1,69){\scriptsize{(q)}}
\end{overpic}
\hspace{0.1cm}
\begin{overpic}[trim=0cm 0cm 0cm 0.2cm, clip,scale=0.36]{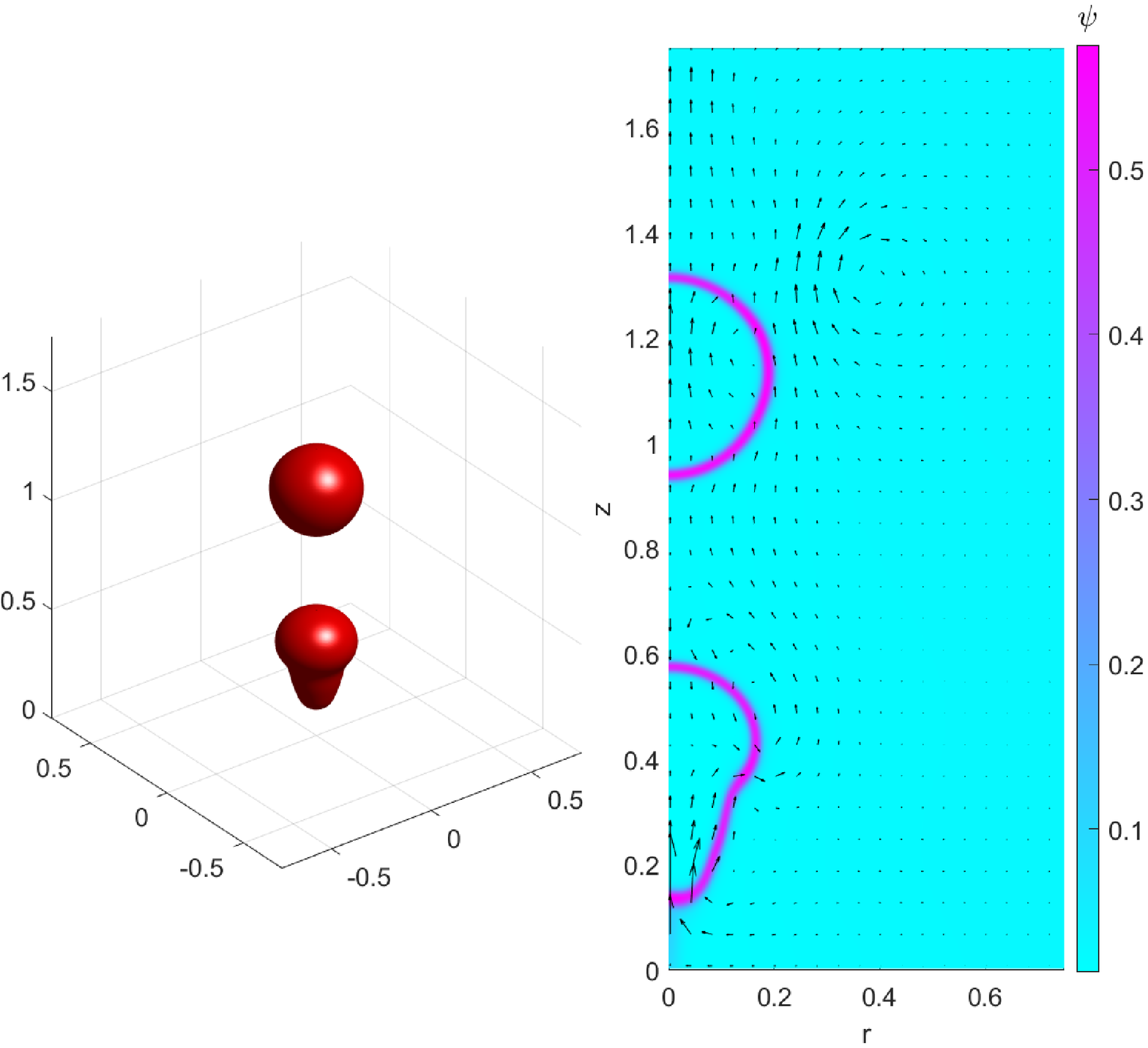}
\put(1,69){\scriptsize{(r)}}
\end{overpic}
\caption{Comparison with experimental results in Example 6: experimental and numerical results for dynamic
process at time $t=1.31$ms (a), $t=1.31$ms (b, c); $t=6.02$ms (d), $t=6.05$ms (e, f); $t=8.21$ms (g), $t=8.38$ms (h, i); $t=20.54$ms (j), $t=20.51$ms (k, l); $t=25.58$ms (m), $t=23.96$ms (n, o); $t=35.17$ms (p), $t=33.42$ms (q, r). Experiment figures courtesy of \cite{Rioboo2002}.}
\label{Reality_1}
\end{figure}

\paragraph{Example 7} In this example, we investigate adherence phenomena for droplet impacting on a smooth glass plate (hydrophilic substrate). We also use the same parameters as in the experiment, which are given as
\begin{equation*}
\mathrm{Re}=7174.4,\quad\quad \mathrm{We}=109.4, \quad\quad \mathrm{L}_{s}=10^{-6}, \quad\quad \theta_s = 10^\circ, \quad\quad  \lambda_\rho=\frac{1}{830},\quad\quad\lambda_\eta=\frac{1}{66.2}.
\end{equation*}
Besides, we take $\mathrm{Pe}_\psi=60$, $\mathrm{Pe}_s=0.5$ and $\Omega =[0, 2] \times[0, 1]$.

The comparison between experimental and numerical results for clean droplet are shown in Fig.~\ref{Reality_2}. Again, almost quantitative agreement between the two sets of results is observed. Different from all previous examples, the droplet keeps on spreading after its impact on the substrate until equilibrium profile is achieved. The major reason for this dynamics is that the substrate is very `sticky' due to the very small Young's angle $\theta_s=10^\circ$. The contact line keeps on advancing until the contact angle achieves Young's angle. The initial kinetic energy is transformed to the surface energy of the droplet interface and the substrate.

The presence of surfactant does not change the impact and spreading dynamics too much. Surfactant transport also takes place from the center of interface towards the moving edge (Fig.~\ref{Reality_2}c-l). The accumulation of surfactant at the moving edge reduces the interfacial tension near the contact line, giving rise to larger uncompensated Young stress and thus faster contact line dynamics (Fig.~\ref{Reality_2}o-u). Eventually, a smaller equilibrium contact angle (than $\theta_s$) is achieved. This phenomenon is also consistent with the simulation result on spreading dynamics (Fig.~\ref{3D_final}) in Sect.~\ref{sec_energy}.
\begin{figure}[t!]
\begin{overpic}[trim=0cm -5cm 0cm 0cm, clip,scale=0.1]{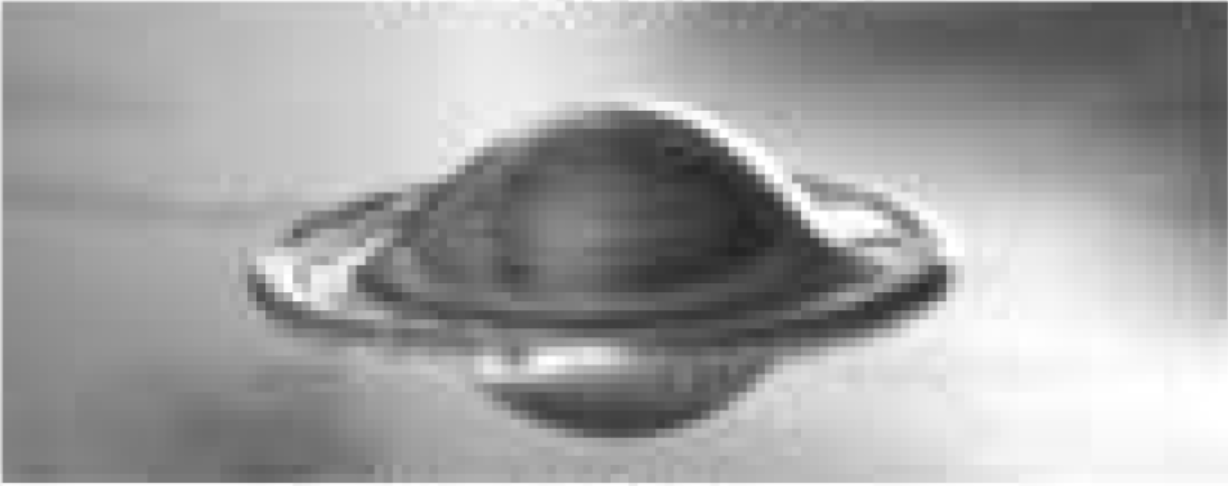}
\put(0,72){\scriptsize{(a)}}
\end{overpic}
\begin{overpic}[trim=0cm 0cm 0cm 0cm, clip,scale=0.3]{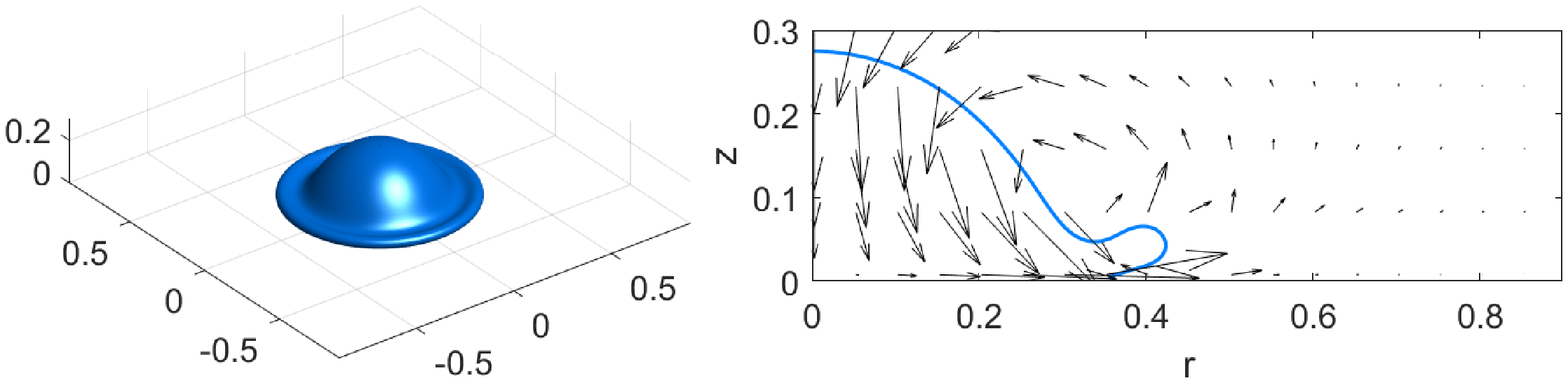}
\put(1,21){\scriptsize{(b)}}
\end{overpic}
\hspace{-0.6cm}
\begin{overpic}[trim=0cm 0cm 0cm 0cm, clip,scale=0.3]{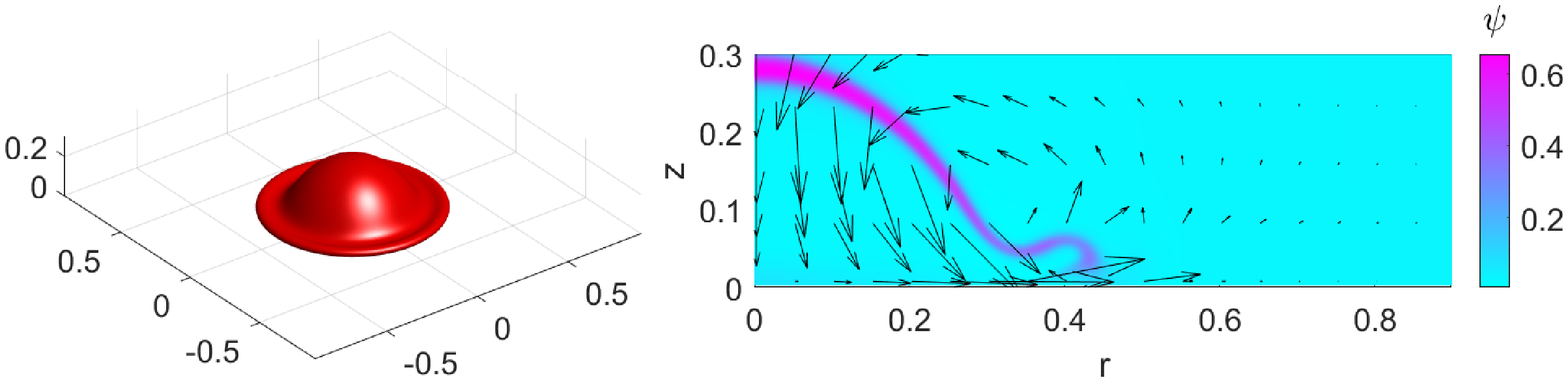}
\put(1,21){\scriptsize{(c)}}
\end{overpic}

\begin{overpic}[trim=0cm -5cm 0cm 0cm, clip,scale=0.1]{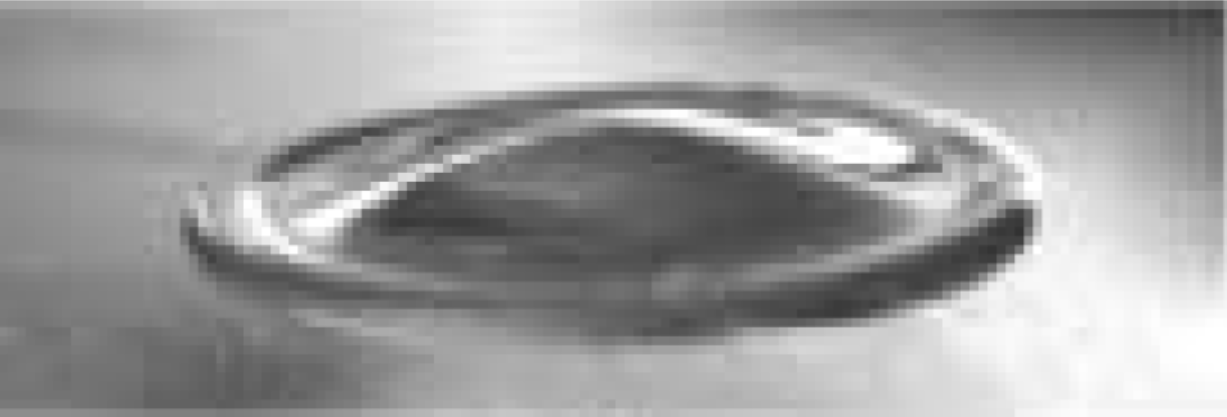}
\put(0,72){\scriptsize{(d)}}
\end{overpic}
\begin{overpic}[trim=0cm 0cm 0cm 0cm, clip,scale=0.3]{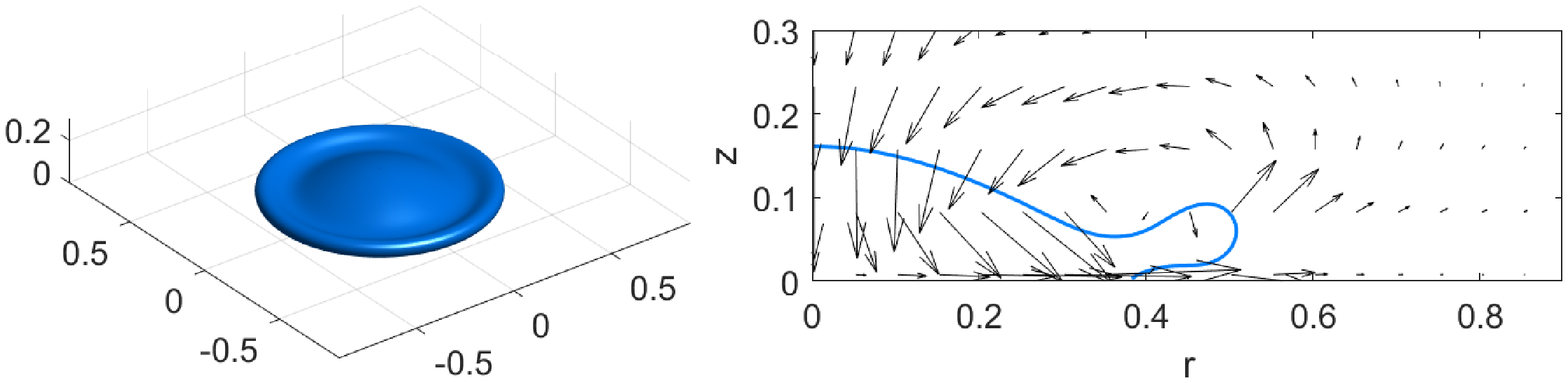}
\put(1,21){\scriptsize{(e)}}
\end{overpic}
\hspace{-0.6cm}
\begin{overpic}[trim=0cm 0cm 0cm 0cm, clip,scale=0.3]{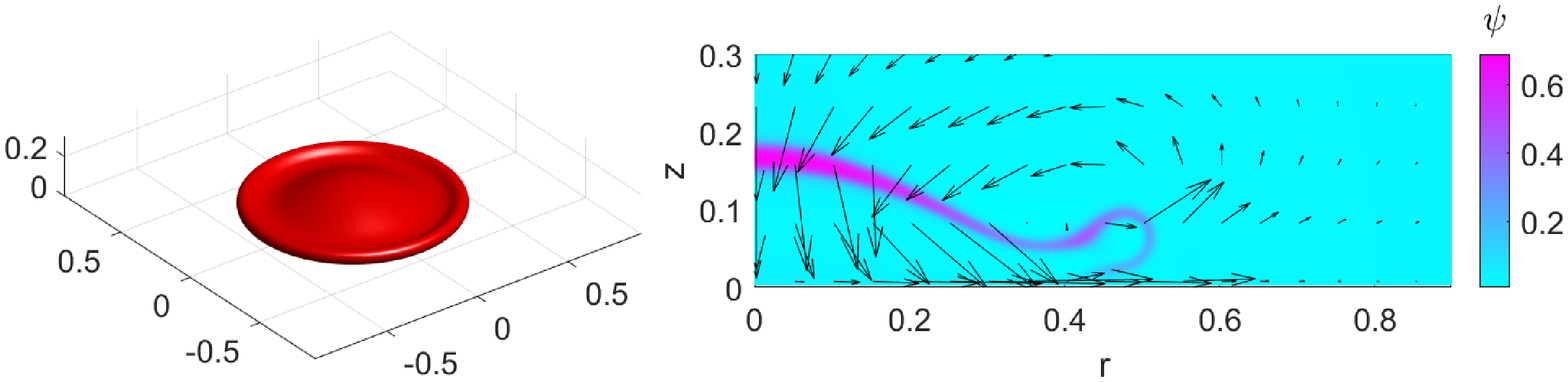}
\put(1,21){\scriptsize{(f)}}
\end{overpic}

\begin{overpic}[trim=0cm -5cm 0cm 0cm, clip,scale=0.1]{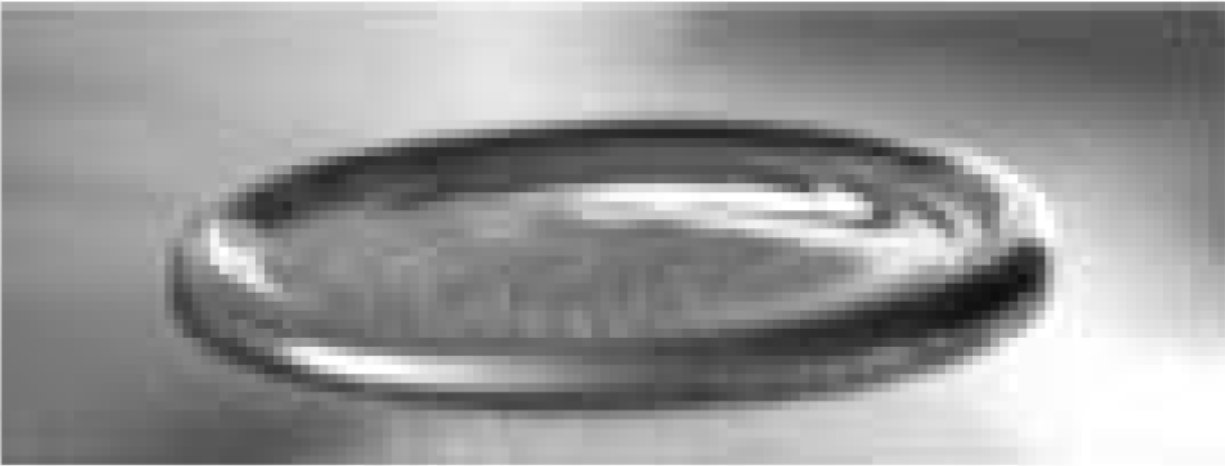}
\put(0,72){\scriptsize{(g)}}
\end{overpic}
\begin{overpic}[trim=0cm 0cm 0cm 0cm, clip,scale=0.3]{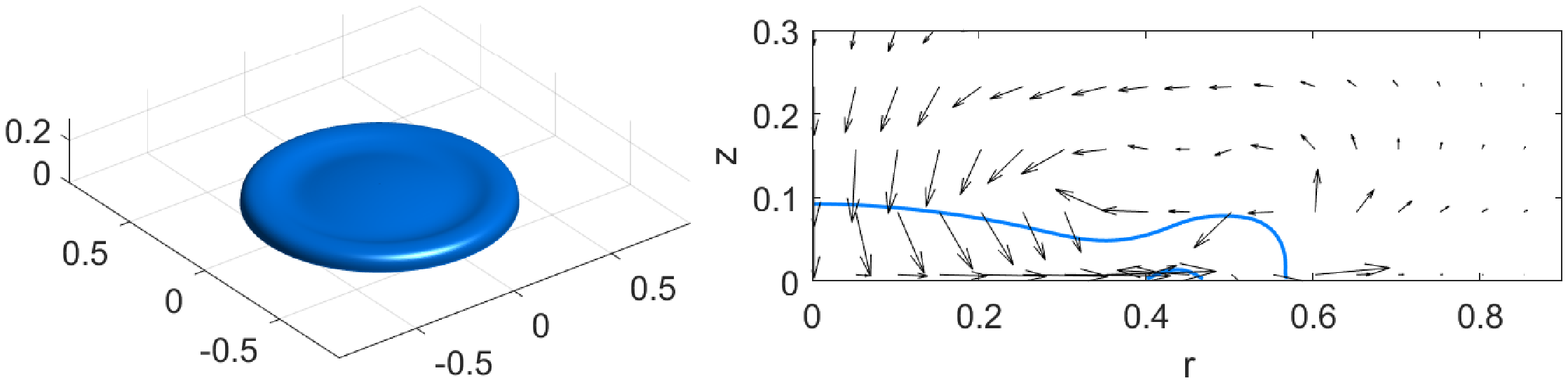}
\put(1,21){\scriptsize{(h)}}
\end{overpic}
\hspace{-0.6cm}
\begin{overpic}[trim=0cm 0cm 0cm 0cm, clip,scale=0.3]{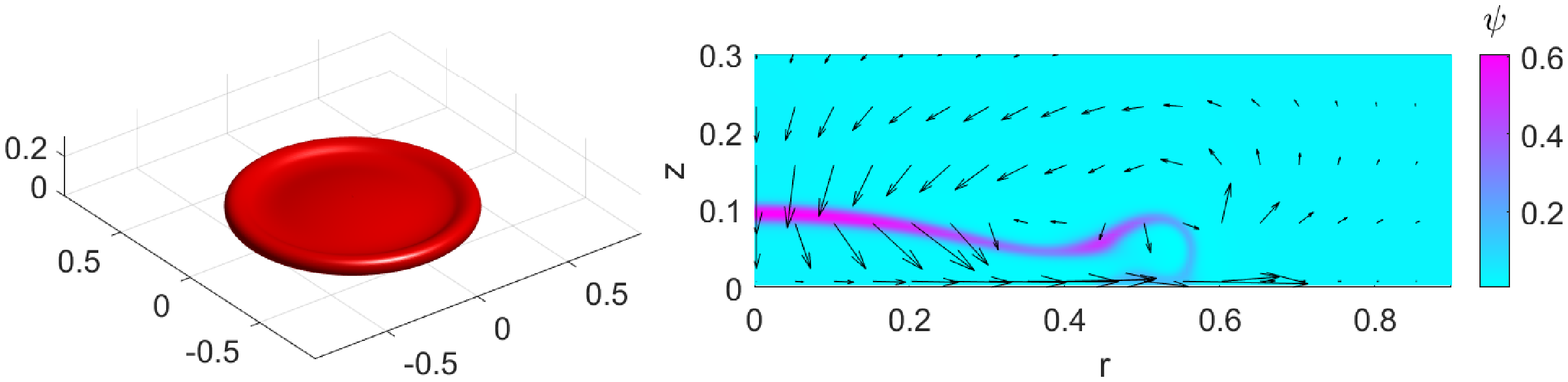}
\put(1,21){\scriptsize{(i)}}
\end{overpic}

\begin{overpic}[trim=0cm -5cm 0cm 0cm, clip,scale=0.1]{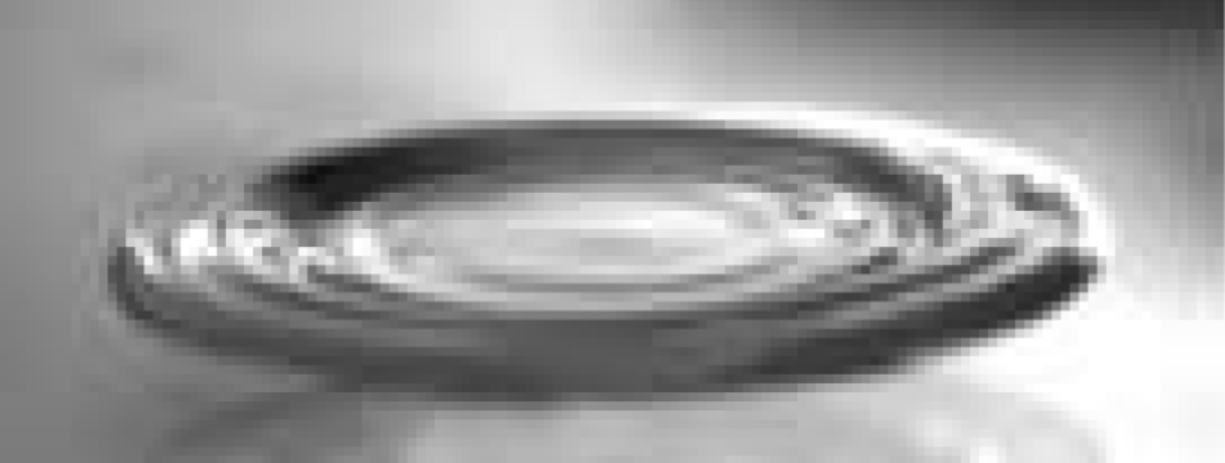}
\put(0,72){\scriptsize{(j)}}
\end{overpic}
\begin{overpic}[trim=0cm 0cm 0cm 0cm, clip,scale=0.3]{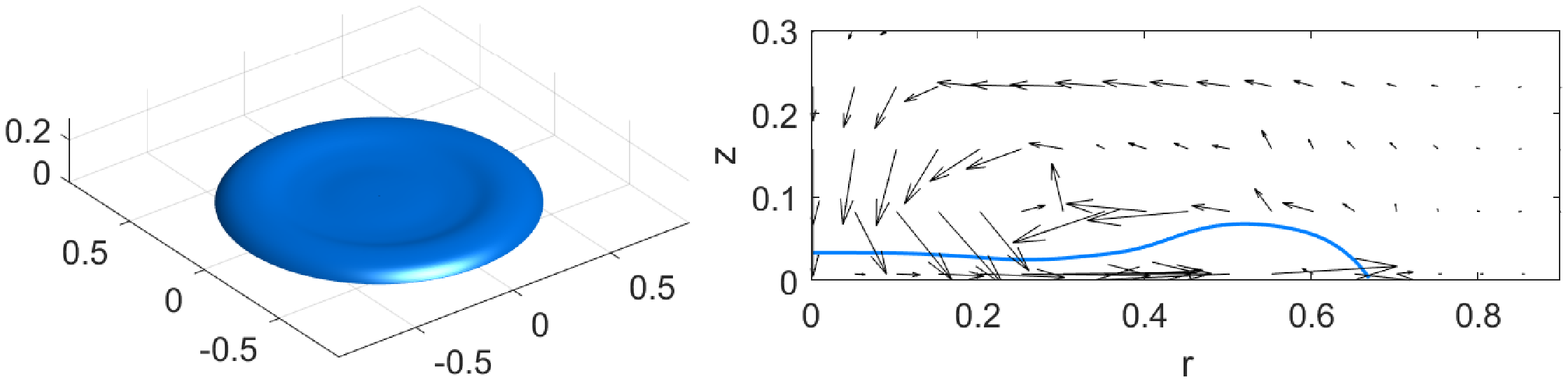}
\put(1,21){\scriptsize{(k)}}
\end{overpic}
\hspace{-0.6cm}
\begin{overpic}[trim=0cm 0cm 0cm 0cm, clip,scale=0.3]{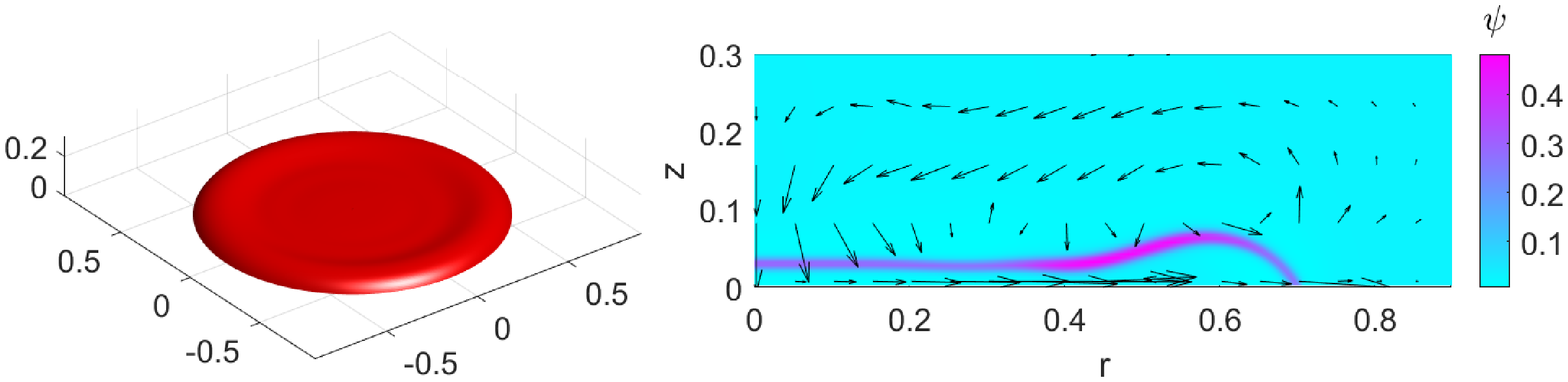}
\put(1,21){\scriptsize{(l)}}
\end{overpic}

\begin{overpic}[trim=0cm -5cm 0cm 0cm, clip,scale=0.1]{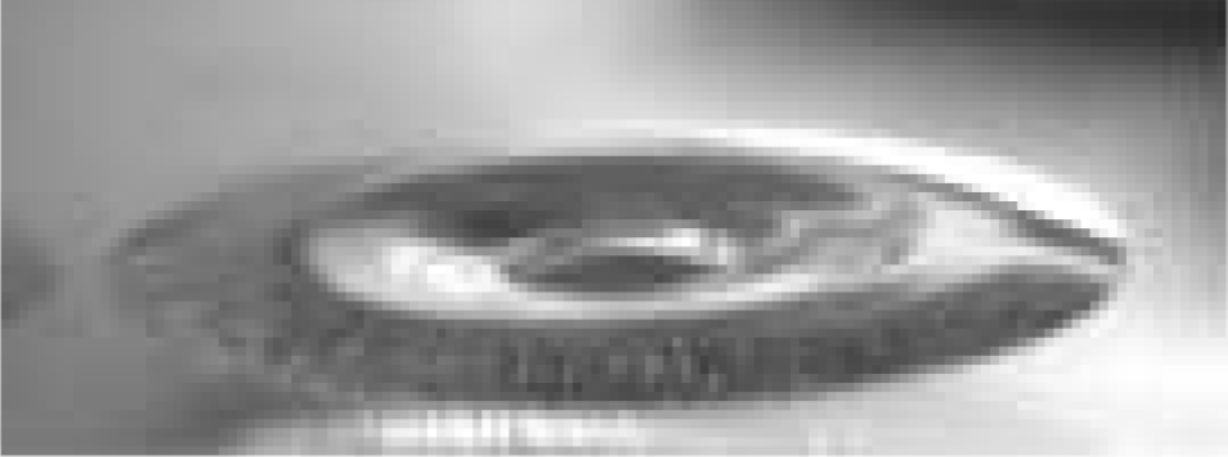}
\put(0,72){\scriptsize{(m)}}
\end{overpic}
\begin{overpic}[trim=0cm 0cm 0cm 0cm, clip,scale=0.3]{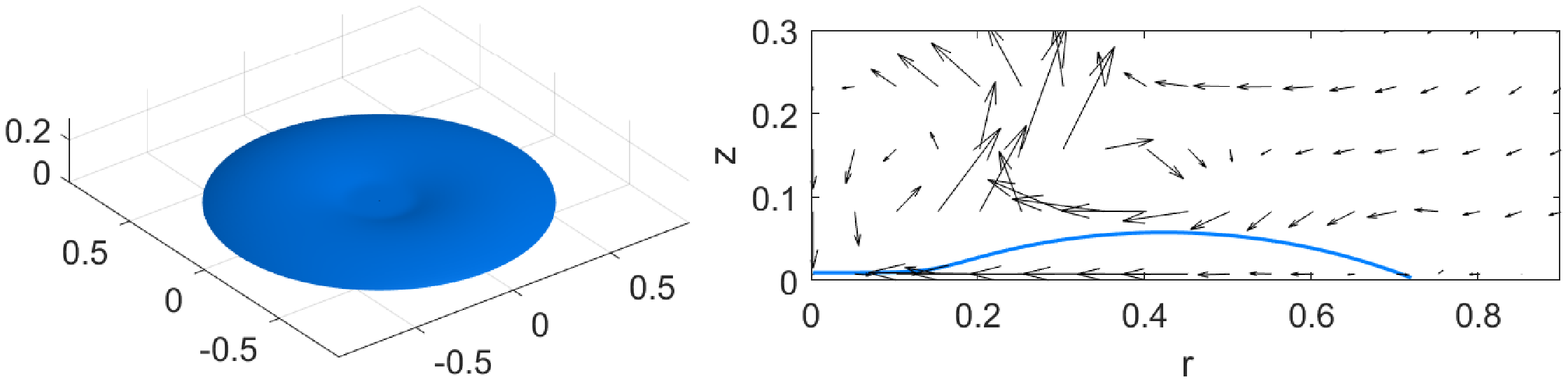}
\put(1,21){\scriptsize{(n)}}
\end{overpic}
\hspace{-0.6cm}
\begin{overpic}[trim=0cm 0cm 0cm 0cm, clip,scale=0.3]{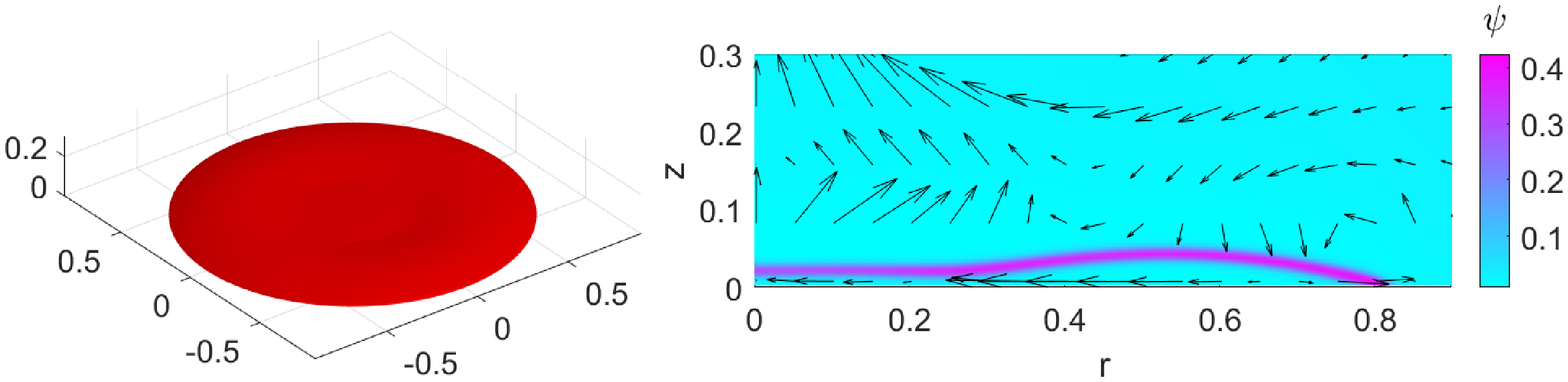}
\put(1,21){\scriptsize{(o)}}
\end{overpic}

\begin{overpic}[trim=0cm -5cm 0cm 0cm, clip,scale=0.1]{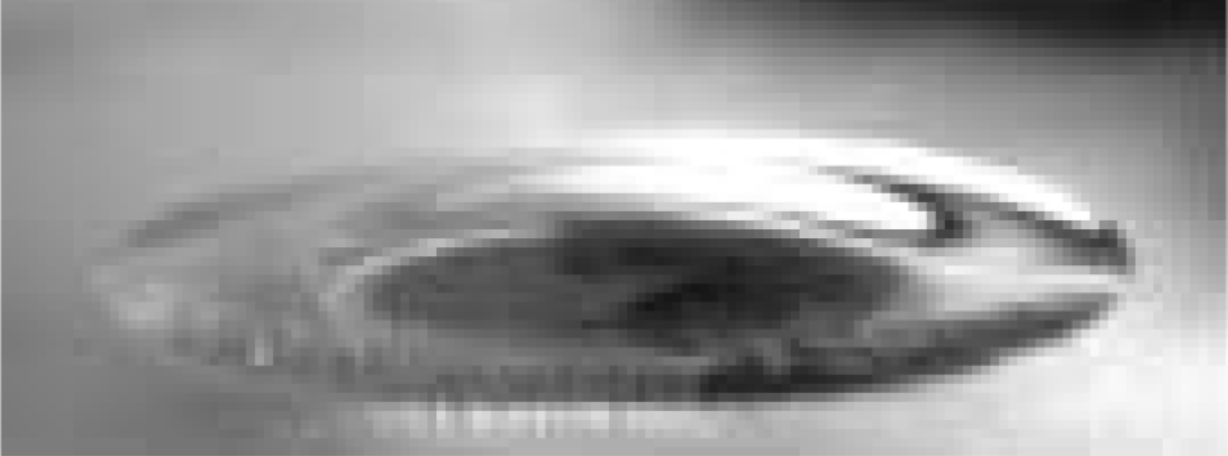}
\put(0,72){\scriptsize{(p)}}
\end{overpic}
\begin{overpic}[trim=0cm 0cm 0cm 0cm, clip,scale=0.3]{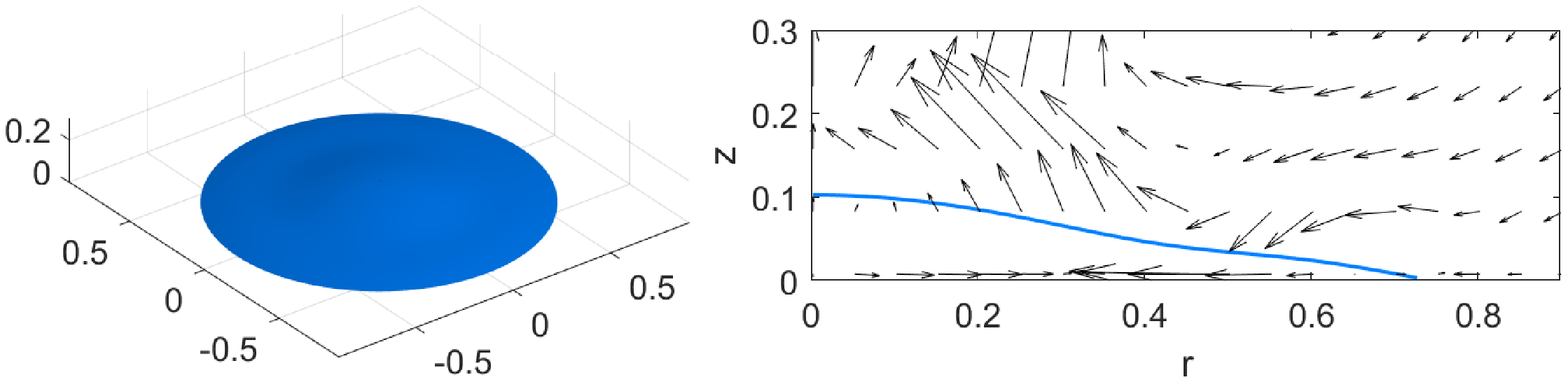}
\put(1,21){\scriptsize{(q)}}
\end{overpic}
\hspace{-0.6cm}
\begin{overpic}[trim=0cm 0cm 0cm 0cm, clip,scale=0.3]{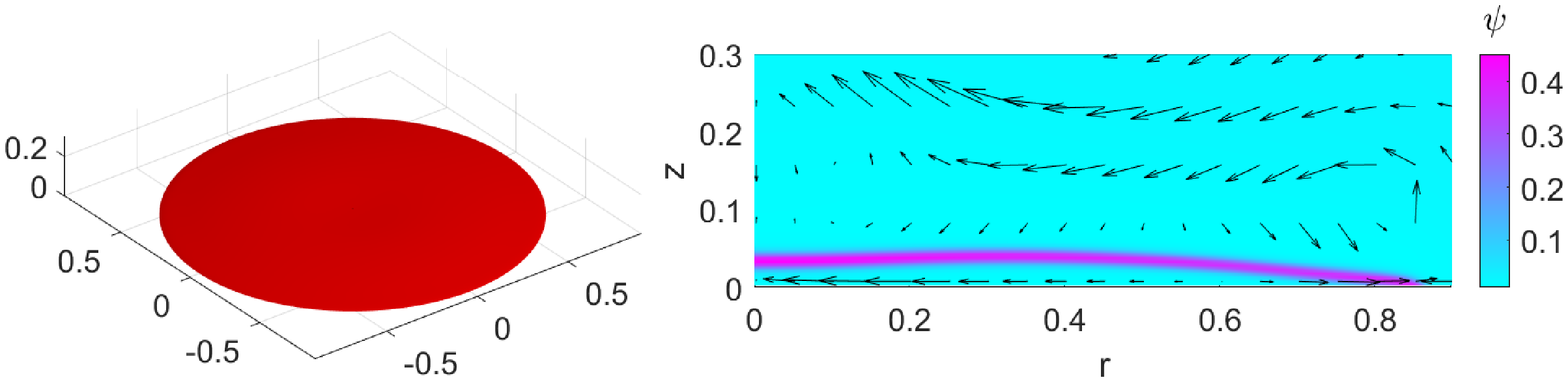}
\put(1,21){\scriptsize{(r)}}
\end{overpic}

\begin{overpic}[trim=0cm -5cm 0cm 0cm, clip,scale=0.1]{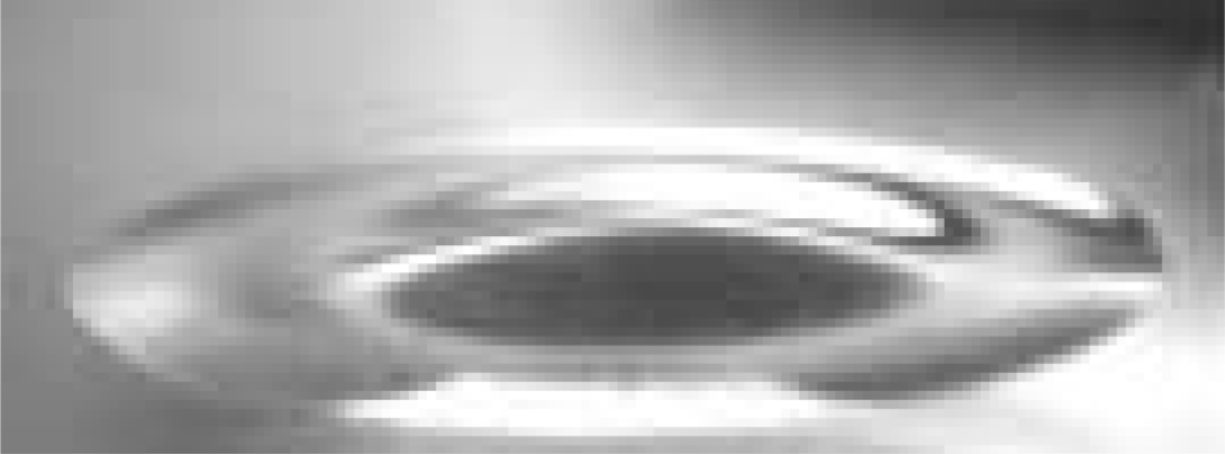}
\put(0,72){\scriptsize{(s)}}
\end{overpic}
\begin{overpic}[trim=0cm 0cm 0cm 0cm, clip,scale=0.3]{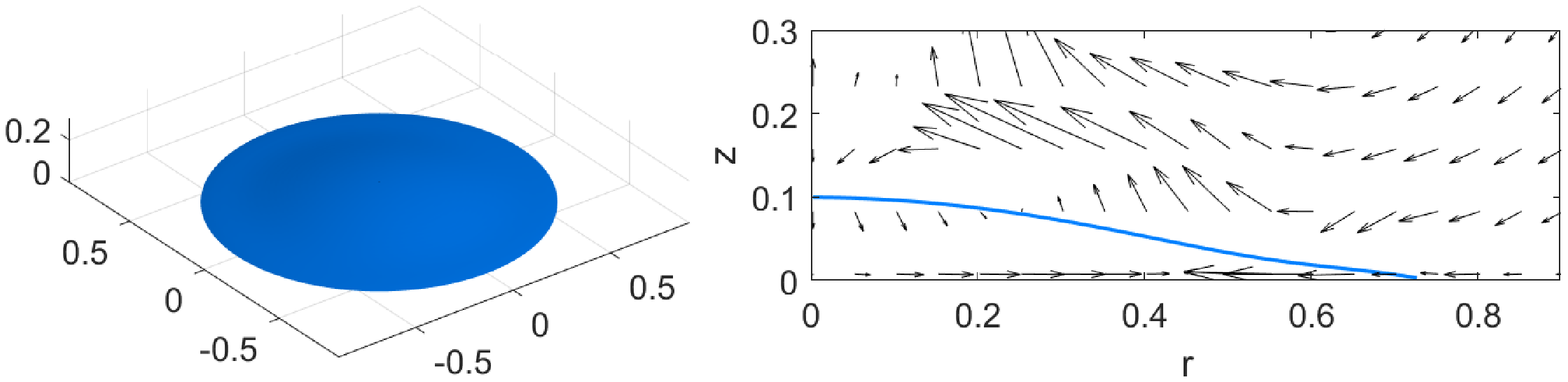}
\put(1,21){\scriptsize{(t)}}
\end{overpic}
\hspace{-0.6cm}
\begin{overpic}[trim=0cm 0cm 0cm 0cm, clip,scale=0.3]{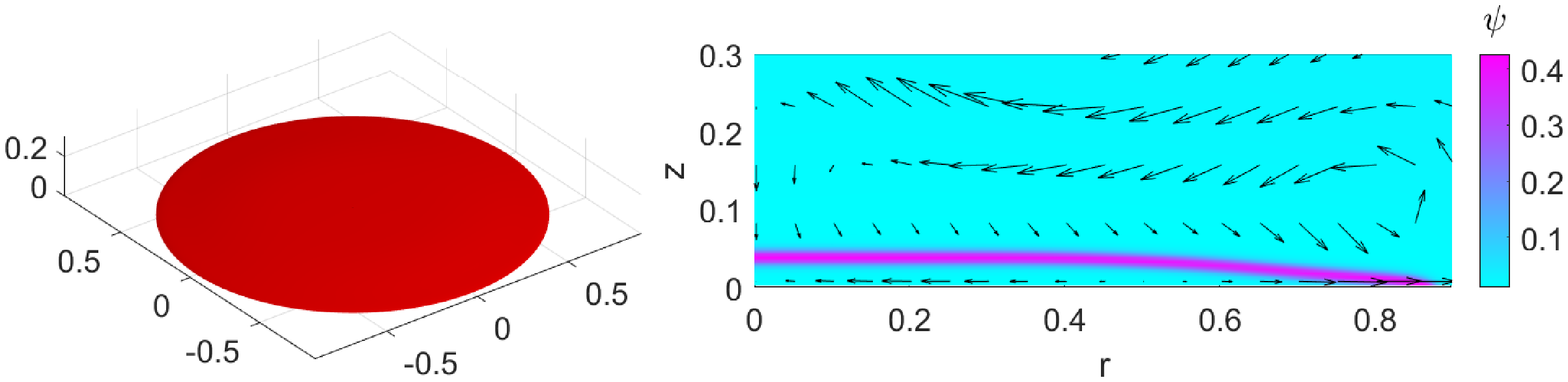}
\put(1,21){\scriptsize{(u)}}
\end{overpic}

\caption{Comparison with experimental results in Example 7: experimental and numerical results for dynamic
process at time $t=1.31$ms (a), $t=1.33$ms (b, c); $t=2.27$ms (d), $t=2.27$ms (e, f); $t=3.16$ms (g), $t=3.10$ms (h, i); $t=6.04$ms (j), $t=6.08$ms (k, l); $t=14.03$ms (m), $t=14.02$ms (n, o); $t=25.58$ms (p), $t=25.56$ms (q, r); $t=29.26$ms (s), $t=29.26$ms (t, u). Experiment figures courtesy of \cite{Rioboo2002}.}
\label{Reality_2}
\end{figure}

\section{Conclusion}\label{sec_Conclusion}
We considered two-phase flow problems with moving contact lines in presence of surfactants, which were modelled by the NS-PFS-MCL system with variable density and viscosity. Based on convex splitting and pressure stabilization techniques, energy stable methods were proposed for this model. In particular, we rigorously proved that the proposed first-order scheme is unconditionally energy stable without modifying the original energy. Moreover, we generalized the methods to be second-order accurate in the form of BDF scheme. The resulting schemes were implemented using linearization techniques in cylindrical coordinates with axisymmetry.


Numerical results were presented to validate accuracy and energy stability of the proposed schemes. We also systematically studied how surfactants affect contact line dynamics, especially in droplet impact.
It was observed that the presence of surfactants makes the droplet more hydrophilic\slash hydrophobic on hydrophilic\slash hydrophobic surfaces. As a result, adding surfactants into an impact droplet can strengthen the adherence effect on hydrophilic surfaces (Example 7) while making splashing more severe on hydrophobic surfaces (Example 5). In general, the dissipation in the impact dynamics of a contaminated drop is smaller than that in the clean case, and topological changes are more likely to occur for contaminated drops. Moreover, adding surfactants may induce some qualitative changes in the droplet impact dynamics: A clean adherent droplet could experience complete bouncing when surfactants are added (Example 2); a clean droplet which completely rebounds becomes partially bounced in the presence of surfactants (Example 3). Lastly, we obtained quantitative agreements with experimental results for impact dynamics in the case without surfactants, and simulated impact processes for contaminated drops as comparisons.

Building linear and fully decoupled schemes satisfying unconditional energy stability with respect to original energy is still open due to the difficulty arising from potential singularity, nonlinear couplings, variable densities and the existence of the MCL. In addition, establishing bound-preserving property for both surfactant concentration $\psi$ and phase-field variable $\phi$ is another challenging issue in numerical analysis and should be taken into account in consideration of real problems. Besides, for the purpose of maintaining divergence-free condition in two-phase flow with variable density, the continuity equation is modified in the proposed model in this work. This drawback can be remedied by considering the quasi-incompressible CHNS model \citep{Lowengrub1998,Guo2017}, which can be potentially generalized to the study of the droplet impact with surfactants.
Dependence of impact phenomena (adherence, bouncing or splashing) on dimensionless parameters, such as $\mathrm{Re}$, $\mathrm{We}$, $\theta_s$ and $\mathrm{Pe_\psi}$, is also of great interest in realistic applications.
This together with real three-dimensional simulations will be our future concern.

\appendix
\section{Appendix}\label{appx}
In this section, the spatial discretization is discussed for the proposed schemes. We first divide $\Omega$ in \eqref{domain} into $n_r \times n_z$ cells, where $n_r$, $n_z$ are the number of cells in the $r$ and $z$ directions. Each cell is characterized by $[(i-1)\Delta r, i\Delta r]\times[(j-1)\Delta z, j\Delta z]$ with the grid sizes $\Delta r=R / n_r$ and $\Delta z=L / n_z$. The cell center, and the centers of its top and right edges are indexed by $(i, j)$, $(i, j+1 / 2)$, and $(i+1 / 2, j)$, respectively. Staggered grids are used: the axial velocity $u_{z}$ and radial velocity $u_{r}$ are evaluated at the edge center, while the order parameters $\phi$ and $\psi$, the chemical potentials $\mu_\phi$ and $\mu_\psi$, and the pressure $p$ are solved at the cell center.

The grid values of variables (denoted by $w$ below) at staggered grids are approximated by (second-order) linear interpolations. For instance, if $w$ is solved at $(i, j+1/2)$, then
\begin{equation*}
w_{i, j}=\frac{w_{i, j-\frac{1}{2}}+w_{i, j+\frac{1}{2}}}{2}, \quad\quad w_{i+\frac{1}{2}, j+\frac{1}{2}}=\frac{w_{i, j+\frac{1}{2}}+w_{i+1, j+\frac{1}{2}}}{2}.
\end{equation*}
Similarly, for variable solved at $(i+1/2, j)$,
\begin{equation*}
w_{i, j}=\frac{w_{i-\frac{1}{2}, j}+w_{i+\frac{1}{2}, j}}{2}, \quad\quad w_{i+\frac{1}{2}, j+\frac{1}{2}}=\frac{w_{i+\frac{1}{2}, j}+w_{i+\frac{1}{2}, j+1}}{2}.
\end{equation*}
For variable solved at cell center $(i, j)$, we can approximate the values at edge centers and corner point by
\begin{equation*}
w_{i, j+\frac{1}{2}}=\frac{w_{i, j+1}+w_{i, j}}{2}, \quad\quad w_{i+\frac{1}{2}, j}=\frac{w_{i+1, j}+w_{i, j}}{2}, \quad\quad   w_{i+\frac{1}{2}, j+\frac{1}{2}}=\frac{w_{i+1, j+1}+w_{i+1, j}+w_{i, j+1}+w_{i, j}}{4}.
\end{equation*}
The differential operators are discretized as
\begin{equation*}
\begin{aligned}
&\Delta w_{i, j}=\frac{r_{i+\frac{1}{2}}(w_{i+1, j}-w_{i, j})-r_{i-\frac{1}{2}}(w_{i, j}-w_{i-1, j})}{r_{i}^{2} \Delta r^{2}}+\frac{w_{i, j+1}-2 w_{i, j}+w_{i, j-1}}{\Delta z^{2}}, \\
&\Delta w_{i+\frac{1}{2}, j}=\frac{r_{i+1}(w_{i+\frac{3}{2}, j}-w_{i+\frac{1}{2}, j})-r_{i}(w_{i+\frac{1}{2}, j}-w_{i-\frac{1}{2}, j})}{r_{i+\frac{1}{2}}^{2} \Delta r^{2}}+\frac{w_{i+\frac{1}{2}, j+1}-2 w_{i+\frac{1}{2}, j}+w_{i+\frac{1}{2}, j-1}}{\Delta z^{2}}, \\
&\Delta w_{i, j+\frac{1}{2}}=\frac{r_{i+\frac{1}{2}}(w_{i+1, j+\frac{1}{2}}-w_{i, j+\frac{1}{2}})-r_{i-\frac{1}{2}}(w_{i, j+\frac{1}{2}}-w_{i-1, j+\frac{1}{2}})}{r_{i}^{2} \Delta r^{2}}+\frac{w_{i, j+\frac{3}{2}}-2 w_{i, j+\frac{1}{2}}+w_{i, j-\frac{1}{2}}}{\Delta z^{2}}.
\end{aligned}
\end{equation*}
In addition, the relaxation boundary condition is treated as follows. Using \eqref{phi} and \eqref{DBC}, we have the following relation
\begin{equation*}
\mathrm{Pe}_s \Delta \mu_\phi=-\mathrm{P e}_{\phi}L(\phi), \quad\quad \text { on }~z=0 .
\end{equation*}
This equation is used in replacement of \eqref{DBC} in the implementation. For $\phi$ and $\mu_\phi$ at $(i,-1/2)$ with $i=0, \ldots, n_r-1$, we have
\begin{equation*}
\mathrm{Pe}_s  [\Delta \mu_\phi^{n+1}]_{i,-\frac{1}{2}}=-\mathrm{P e}_{\phi}L(\phi^{n+1}_{i,-\frac{1}{2}}),
\end{equation*}
where $\Delta \mu_\phi$ is evaluated on $(i,-1 / 2)$ using the boundary condition ${\mu_\phi}_{i,-1}={\mu_\phi}_{i, 0}$ and second-order extrapolation:
\begin{equation*}
[\Delta {\mu_\phi}]_{i,-\frac{1}{2}}
=\frac{r_{i+\frac{1}{2}}({\mu_\phi}_{i+1,0}-{\mu_\phi}_{i, 0})-r_{i-\frac{1}{2}}({\mu_\phi}_{i, 0}-{\mu_\phi}_{i-1,0})}{r_{i} \Delta r^{2}}-\frac{({\mu_\phi}_{i, 2}-5 {\mu_\phi}_{i, 1}+4 {\mu_\phi}_{i, 0})}{2 \Delta z^{2}}.
\end{equation*}

\section*{Acknowledgments}
The authors gratefully acknowledge many helpful discussions with Qian Zhang and Guangpu Zhu during the preparation of the paper. The work of Zhen Zhang was partially supported by the NSFC grant (NO. 12071207), NSFC Tianyuan-Pazhou grant (No. 12126602), the Natural Science Foundation of Guangdong Province (2021A1515010359), and the Guangdong Provincial Key Laboratory of Computational Science and Material Design (No. 2019B030301001). The work of Ming-Chih Lai was supported by NSTC of Taiwan  by grant number 110-2115-M-A49-011-MY3.

\section*{Declarations}
The authors declare that they have no known competing financial interests or personal relationships that could have appeared to influence the work reported in this paper.


\bibliographystyle{model1-num-names}
\bibliography{reference}

\end{document}